\pdfoutput=1
\documentclass[11pt,letterpaper]{article}

\usepackage[utf8]{inputenc}

\usepackage[english]{babel}
\usepackage[utf8]{inputenc}
\usepackage[T1]{fontenc}
\usepackage{babel}
\usepackage[font=small,labelfont=bf]{caption}

\usepackage[margin=1in]{geometry}
\usepackage{amsthm}
\usepackage{amsmath}
\usepackage{amsfonts}
\usepackage{bbm}
\usepackage{bm}
\usepackage{graphicx}
\usepackage{tikz}
\usepackage[colorinlistoftodos]{todonotes}
\usepackage[colorlinks=true, allcolors=blue]{hyperref}
\usepackage[ruled,linesnumbered,algo2e]{algorithm2e}
\usepackage{cleveref}
\usepackage{subfigure}
\usepackage{enumitem}

\newtheorem{theorem}{Theorem}
\newtheorem{lemma}{Lemma}

\newtheorem{claim}{Claim}
\newtheorem{definition}{Definition}
\newtheorem{corollary}{Corollary}
\newtheorem{proposition}{Proposition}

\newcommand{\defeq}{:=}
\newcommand{\norm}[1]{\left\lVert#1\right\rVert}
\newcommand{\inprod}[2]{\left\langle#1, #2\right\rangle}
\newcommand{\eps}{\epsilon}

\newcommand{\argmax}{\textup{argmax}}
\newcommand{\argmin}{\textup{argmin}} 
\newcommand{\R}{\mathbb{R}}

\newcommand{\N}{\mathbb{N}}

\newcommand{\E}{\mathbb{E}}

\newcommand{\Nor}{\mathcal{N}}

\definecolor{burntorange}{rgb}{0.8, 0.33, 0.0}

\newcommand{\xbold}{\mathbf{x}}

\newcommand{\calI}{\mathcal{I}}
\newcommand{\Par}[1]{\left(#1\right)}
\newcommand{\Brack}[1]{\left[#1\right]}

\newcommand{\ind}{\mathrm{in}}
\newcommand{\out}{\mathsf{out}}

\newcommand{\tmix}{t_{\mathrm{mix}}}

\newcommand{\calA}{\mathcal{A}}
\newcommand{\Atot}{A_{\mathsf{tot}}}
\newcommand{\Atoti}{A_{\mathsf{tot},i}}
\newcommand{\calS}{\mathcal{S}}
\newcommand{\calX}{\mathcal{X}}
\newcommand{\PP}{\mathbf{P}}
\newcommand{\QQ}{\mathbf{Q}}

\newcommand{\II}{\mathbf{I}}
\newcommand{\ee}{\mathbf{e}}
\newcommand{\zero}{\mathbf{0}}
\renewcommand{\aa}{\mathbf{a}}
\newcommand{\pp}{\mathbf{p}}
\newcommand{\ppi}{\bm{\pi}}
\newcommand{\llam}{\bm{\lambda}}
\newcommand{\tildeppi}{\bm{\widetilde{\pi}}}
\newcommand{\tildepi}{\widetilde{\pi}}

\newcommand{\qq}{\mathbf{q}}
\renewcommand{\r}{\mathbf{r}}
\newcommand{\RR}{\mathbf{R}}

\newcommand{\VV}{\textbf{V}}

\newcommand{\subin}{_\mathsf{in}}
\newcommand{\subino}{_{\mathsf{in}1}}
\newcommand{\subint}{_{\mathsf{in}2}}
\newcommand{\subout}{_\mathsf{out}}
\newcommand{\subaux}{_\mathsf{aux}}
\newcommand{\med}{\mathsf{med}}
\newcommand{\poly}{\mathrm{poly}}
\newcommand{\calM}{\mathcal{M}}
\newcommand{\calG}{\mathcal{G}}
\newcommand{\calT}{\mathcal{T}}
\newcommand{\calC}{\mathcal{C}}

\newcommand{\Vbold}{\mathbf{V}}

\newcommand{\tbsg}{\mathsf{TBSG}}
\newcommand{\ssg}{\mathsf{SimSG}}
\mathchardef\mhyphen="2D
\newcommand{\onestate}{\mathsf{O}\mhyphen\tbsg}

\newcommand{\onestatessg}{\mathsf{O}\mhyphen\ssg}

\newcommand{\lrm}{\mathsf{long}}
\newcommand{\srm}{\mathsf{short}}
\newcommand{\aux}{\mathsf{aux}}
\newcommand{\ar}{\mathsf{AR}}
\newcommand{\slrm}{\mathsf{sl}}

\newcommand{\outility}{\upsilon}

\newcommand{\gcircuit}{\mathsf{GCircuit}}

\newcommand{\ronly}{\mathsf{LocReward}}
\newcommand{\ppad}{\mathsf{PPAD}}
\newcommand{\LLam}{\mathbf{\Lambda}}
\newcommand{\lineofaend}{\mathsf{EndOfALine}}
\newcommand{\brouwer}{\mathsf{Brouwer}}
\newcommand{\osBsf}{\mathsf{os-Bellman}}
\newcommand{\Vsf}{\mathsf{value}}
\newcommand{\tail}{\mathsf{tail}}
\newcommand{\head}{\mathsf{head}}
\newcommand{\Abs}{\mathsf{Abs}}

\newcommand{\codeStyle}[1]{{\bfseries #1} }
\newcommand{\codeInput}{\codeStyle{Input:}}	
	
\newcommand{\codeReturn}{\codeStyle{Return:}}

\SetCommentSty{mycommfont}
\SetKwInput{KwInput}{Input} 
\SetKwInput{KwParameter}{Parameters}                %
\SetKwInput{KwOutput}{Output}              %
\SetKwInput{KwReturn}{Return}              %
\SetKwComment{Comment}{$\triangleright$\ }{}
\SetCommentSty{mycommfont}

\definecolor{burntorange}{rgb}{0.8, 0.33, 0.0}

\title{The Complexity of Infinite-Horizon\\
	General-Sum Stochastic Games}

\author{
Yujia Jin \\
Stanford University \\
\texttt{\href{mailto:yujiajin@stanford.edu}{yujiajin@stanford.edu}}
\and
	Vidya Muthukumar \\
Georgia Institute of Technology \\
\texttt{\href{mailto:vmuthukumar8@gatech.edu}{vmuthukumar8@gatech.edu}}
\and
	Aaron Sidford \\
Stanford University \\
\texttt{\href{mailto:sidford@stanford.edu}{sidford@stanford.edu}}
}

\date{}

\begin{document}

\maketitle
\thispagestyle{empty}

\begin{abstract}

We study the complexity of computing stationary Nash equilibrium (NE) in $n$-player infinite-horizon general-sum stochastic games. 
We focus on the problem of computing NE in such stochastic games when each player is restricted to choosing a stationary policy and rewards are discounted. 
 First, we prove that computing such NE is in $\ppad$ (in addition to clearly being $\ppad$-hard). 
Second, we consider turn-based specializations of such games where at each state there is at most a single player that can take actions and show that these (seemingly-simpler) games remain $\ppad$-hard.
Third, we show that under further structural assumptions on the rewards computing NE in such turn-based games is possible in polynomial time. 
Towards achieving these results we establish structural facts about stochastic games of broader utility, including monotonicity of utilities under single-state single-action changes and reductions to settings where each player controls a single state.

\end{abstract}

\newpage
\tableofcontents
\thispagestyle{empty}
\newpage
\pagenumbering{arabic}

\section{Introduction}\label{sec:intro}

Stochastic games~\cite{filar2012competitive,bacsar1998dynamic} are a fundamental mathematical model for dynamic, non-cooperative interaction between multiple players. 
Multi-player dynamic interaction arises naturally in a diverse set of contexts including natural resource competition~\cite{levhari1980great}, monetary interaction in markets~\cite{karatzas1997strategic}, packet routing~\cite{altman1994flow}, and computer games~\cite{silver2017mastering,silver2016mastering}.
Such games have also been of increased study in reinforcement learning (RL); there have been a number of successes in transferring results from single-player RL to multiplayer RL under zero-sum and cooperative interaction, but comparatively less success for general-sum interaction (see e.g.\ \cite{zhang2021multi} for a survey). 

We consider the broad class of general-sum, simultaneous, tabular, $n$-player stochastic games \cite{shapley1953stochastic,fink1964equilibrium,takahashi1964equilibrium}, which we henceforth refer to as $\ssg$s.\footnote{See Section~\ref{sec:prelim} for the more formal definition and description of our notational conventions.}
$\ssg$s are parameterized by a (finite) state space $\calS$ and disjoint (finite) action sets $\calA_{i,s}$ for each player $i$ and state $s$. The players choose a joint strategy $\ppi$, consisting of distributions $\ppi_{i,s}^t$ over the actions $\calA_{i,s}$ for each player $i \in [n]$ at each state $s \in \calS$ at time-step $t \geq 0$. The game then proceeds in time-steps, where in each time-step $t \geq 0$, the game is at a state $s^t \in \calS$ and each player $i \in [n]$ samples independently from $\calA_{i,s^t}$ according to $\ppi_{i,s^t}^{t}$. The set of actions $\aa^t$ chosen at time-step $t$ then yields an immediate reward $\r_{i,s,\aa^t}$ to each player $i$, and causes the next state $s^{t +1}$ to be sampled from a distribution $\pp_{s,\aa^t}$. Each player $i$ aims to maximize her own long-term value as a function of the rewards they receive, i.e.~$\r_{i,s,\aa^t}$. 
For any fixed strategy the states in a $\ssg$ evolve as a Markov chain\footnote{This Markov structure is commonly assumed across the stochastic games literature, particularly when stationary strategies are considered~\cite{shapley1953stochastic,filar2012competitive,condon1992complexity} and some recent literature~\cite{bai2020near,jin2021v,song2021can} refers to $\ssg$s as \emph{Markov games}. There are studied generalizations of $\ssg$s that allow non-stationary or non-Markovian dynamics~\cite{bacsar1998dynamic}, but are outside the scope of this paper.} and the single-player specialization of $\ssg$s, i.e. when $n = 1$, is a Markov decision processes (MDP)~\cite{Bertsekas19,puterman2014markov}.

Our focus in this paper is on computing (\emph{approximate}) \emph{Nash equilibrium} in the multiplayer \emph{general-sum setting} of $\ssg$s. The term \emph{general-sum} emphasizes that we do not impose any shared structure on the immediate reward functions across players (in contrast to the special case of zero-sum games where $n = 2$ and $\r_2 = - \r_1$).
A \emph{Nash equilibrium (NE)}~\cite{nash1951non} is defined as a joint strategy $\ppi$ such that no player can gain in reward by deviating (keeping the other players' strategies fixed). NE is a solution concept of fundamental interest and importance in both static~\cite{myerson1997game} and dynamic games~\cite{filar2012competitive,bacsar1998dynamic}.
While NE are known to always exist in $\ssg$~\cite{shapley1953stochastic,fink1964equilibrium,takahashi1964equilibrium}, they are challenging to compute efficiently; current provably efficient algorithms from computing (approximate) NE for $\ssg$s make strong assumptions on the rewards~\cite{hu2003nash,littman2001friend}.

One setting for which the complexity of computing general-sum NE is relatively well-understood is the \emph{finite-horizon} model, where all players play up to a horizon of finite and known length $H$ and wish to optimize their total reward.
Even with just two players, and a single state (or multiple states but a horizon length $H = 1$), the problem of NE computation in $\ssg$ is $\ppad$-hard as it generalizes computing NE for a two-player normal-form game which is known to be $\ppad$-complete~\cite{chen2006settling,daskalakis2009complexity}.
On the other hand, by leveraging stochastic dynamic programming techniques~\cite{filar2012competitive}, one can show that the complexity of NE computation in finite-horizon $\ssg$ and normal-form games is polynomial-time equivalent: in particular, NE-computation remains $\ppad$-complete.
This dynamic programming technique is also broadly applicable to solution concepts that are comparatively tractable, such as correlated equilibrium (CE)~\cite{papadimitriou2008computing}.
Exploiting this property, recent work~\cite{jin2021v,song2021can,mao2022provably} has shown that simple decentralized RL algorithms can provably learn and converge to the set of CE's in a finite-horizon $\ssg$.

The central goal of this work is to broaden our understanding of the complexity of $\ssg$s. We ask, ``\emph{how brittle is the property of $\ppad$-completeness of finite-horizon $\ssg$s?}'', specifically to: 
\begin{itemize}
\item \emph{\textbf{Infinite time horizon}: what if players optimize rewards over an infinite time horizon?}
\item \emph{\textbf{Turn-based games}: what if each state is controlled only by a single player?}
\item \emph{\textbf{Localized rewards}: what if rewards are only received for a player at states they control?} 
\end{itemize} 

In this paper we systematically address these questions and provide theoretical foundations for understanding the complexity of infinite-horizon stochastic games. Our key results include complexity-class characterizations, algorithms, equivalences and structural results regarding such games. 
For a brief summary of our main complexity characterizations, see Table~\ref{tab:complexityfull}.

\paragraph{Infinite time horizon:} 
First, we consider \emph{infinite-horizon} $\ssg$s in which each player seeks to maximize rewards over an infinite time horizon while following a \emph{stationary strategy}.
A stationary strategy is one in which action distributions are independent of the time-step (i.e.\ $\ppi_{i,s}^{t_1} = \ppi_{i,s}^{t_2}$ for all $i \in [n]$, $s \in \calS$, and  $t_1,t_2 \geq 0$).
We focus on the discounted-reward model, and defer discussion of the alternative average-reward model to~\Cref{apdx:average}.
(The single-player version of such games is known as a discounted Markov decision process (DMDP) and has been the subject of extensive study in optimization~\cite{ye2011simplex}, operations research~\cite{puterman2014markov}, and machine learning~\cite{sutton2018reinforcement}.)
Stationary strategies are especially attractive to study owing to their succinctness in representation compared to non-stationary strategies and the fact that stationary \emph{policies} (the 1-player analog of strategies) can attain the optimal value in single-player DMDPs~\cite{Bertsekas19,puterman2014markov}. 

Despite the fact that stationary NE are always known to exist in infinite-horizon $\ssg$s~\cite{fink1964equilibrium,takahashi1964equilibrium}, existence does not appear to directly follow from the straightforward proof of existence in finite-horizon $\ssg$. In particular, the dynamic programming technique for finite-horizon $\ssg$s breaks down for infinite-horizon  $\ssg$s~\cite{zinkevich2006cyclic} and does not directly imply membership in $\ppad$. 
Nevertheless, as described in~\Cref{sec:resultsssg}, we show that the stationary NE-computation problem for $\ssg$ remains in $\ppad$.
To prove this result we establish a number of key properties of discounted $\ssg$s (and, thereby, an alternative NE existence proof) that are crucial for several of the results in this paper and may be independently useful for future $\ssg$ algorithm design (see \Cref{sec:foundation}).

\paragraph{Turn-based games:}
We then consider \emph{turn-based} variants of $\ssg$s, which we henceforth refer to as $\tbsg$. Formally, $\tbsg$s are the specialization of $\ssg$s where for each state there is at most one-player that has a non-trivial set of distinct actions to choose from. $\tbsg$s are common in the literature and encompass the popular instantiations of game-play for which large-scale RL has yielded empirical success~\cite{silver2017mastering,silver2016mastering}.
Additionally, they have been extensively studied in the case of two players and zero-sum rewards~\cite{shapley1953stochastic,condon1992complexity,etessami2010complexity,hansen2013strategy,sidford2020solving}.

Whereas it was natural to suspect that discounted $\ssg$s would be $\ppad$-complete, the computational complexity of computing NE for $\tbsg$s seems less clear.
The trivial proof of $\ppad$-hardness for $\ssg$s breaks down even for the case of multiplayer $\tbsg$ --- specializing to a single-state game reduces the problem to trivial independent reward maximization by each player, rather than a simultaneous normal-form game.
More generally, $\tbsg$s seem to have more special structure than $\ssg$s owing to the restriction of a single player controlling each state.
As a quick illustration of this structure, note that non-stationary NE for general-sum, \emph{finite-horizon} $\tbsg$s can be computed in polynomial time by a careful application of the multi-agent dynamic programming technique.
Further, in~\Cref{sec:nonstat}, we extend this technique to show that \emph{non-stationary} NE for $\tbsg$s can be computed in polynomial time for a polynomially bounded discount factor.

Despite this seemingly special structure of $\tbsg$s, one of the main contributions of our work (described in~\Cref{sec:resultstbsg}) is to show that computing a multiplayer \emph{stationary} NE for $\tbsg$ is $\ppad$-hard even for a constant discount factor $\gamma \in (0,1)$.
This shows a surprising and non-standard divergence between the non-stationary and stationary solution concepts in infinite-horizon stochastic games.
Moreover, it even implies the hardness of stationary coarse-correlated equilibrium (CCE) computation in $\ssg$s (owing to a stationary NE in $\tbsg$s being a special case), which is a relaxed notation of equilibrium that allows for more computationally-efficient methods in two-player normal-form games (in contrast to $\ssg$s).
Our hardness results hold even for $\tbsg$s for which each player controls a different state, and each player receives a non-zero reward (allowed to be either positive or negative) at at most $4$ states, including her own.

\paragraph{Localized rewards:}
Finally, with the hardness of discounted general-sum $\ssg$s and $\tbsg$s established, we ask \emph{``under what further conditions on reward functions are there polynomial-time algorithms for $\tbsg$s?''}
As described in~\Cref{sec:resultslocalrewards}, we show that further localizing the reward structure such that each player receives a reward of the same sign \emph{only} at a single state which she controls changes the complexity picture and leads to a polynomial-time algorithm.
We show that for these specially structured $\tbsg$s, a pure NE always exists and is polynomial-time computable via approximate best-response dynamics (also called \emph{strategy iteration} in the stochastic games literature~\cite{hansen2013strategy}).
These results are derived via a connection to potential games~\cite{monderer1996potential} modulo a monotonic transformation of the utilities.
While the connection to potential game theory yields \emph{approximate} NE, we also design a more combinatorial, graph-theoretic algorithm that computes \emph{exact} NE in polynomial-time if, additionally, the transitions in $\tbsg$ are deterministic.

\paragraph{Summary and additional implications:}

In summary, we show that (a) stationary NE computation for infinite-horizon $\ssg's$ is in $\ppad$, (b) stationary NE computation for infinite-horizon $\tbsg$s is $\ppad$-hard, and (c) stationary pure NE computation for infinite-horizon $\tbsg$s when each player receives a consistently-signed reward at one controlled state is polynomial-time solvable.

Beyond shedding light on the complexity of infinite horizon general-sum stochastic games, our work yields several insights and implications of additional interest. 
On the one hand, our hardness result for stationary NE in infinite-horizon $\tbsg$ implies the hardness of slightly more complex solution concepts such as stationary coarse-correlated equilibrium (CCE) in $\ssg$ (as the former is a special case of the latter).
On the other hand, our $\ppad$ membership result for $\ssg$ (which includes $\tbsg$ as a special case) is interesting as in $\ssg$s the utility that a player receives is a non-convex function of her actions, and general-sum non-convex games lie in a complexity class suspected to be harder than $\ppad$~\cite{schoenebeck2012computational}; indeed, even the zero-sum case is $\ppad$-hard~\cite{daskalakis2021complexity}. 
Further, many of the results in this paper crucially utilize special structure that we prove (in Lemma~\ref{lem:quasi-mono}) of a monotonic change with upper and lower-bounded slope (which we refer to as \emph{pseudo-linear}) on each player's value function when she changes her policy at only \emph{one} state.
This observation has powerful consequences for many of our results and allows us to leverage several algorithmic techniques that are normally applied only to linear and piecewise-linear utilities.
As one example, it yields a particularly simple existence proof of stationary NE in $\ssg$ compared to past literature~\cite{fink1964equilibrium,takahashi1964equilibrium}.
We hope these results facilitate the further study of infinite-horizon stochastic games.

\begin{table}[t]
\begin{center}
\begin{small}
\begin{tabular}{|l|c|c|c|}
\hline
Setting &  $\ssg$ & $\tbsg$ & $\tbsg$ (localized rewards) \\
\hline\hline
Finite-horizon & $\ppad$ & Polynomial & Polynomial \\
\hline
Infinite-horizon & $\ppad$ (\Cref{thm:ppad-membership}) & $\ppad$ (\Cref{thm:PPAD-hard}) & Polynomial (\Cref{lem:pureinP_potential})  \\
\hline
\end{tabular}
\end{small}
\caption{ 
Summary of $\ssg$ complexity characterization. Characterizations are for computing non-stationary NE in the finite-horizon case, and for computing stationary NE in the infinite-horizon case.
\label{tab:complexityfull}}
\end{center}
\vskip -0.4in
\end{table}

\paragraph{Paper Organization:} We cover notation and fundamental definitions in \Cref{sec:prelim}, an overview of our results and techniques in \Cref{sec:overview}, and related work in \Cref{sec:related-work-short}. Main results are proved in \Cref{sec:axiom,sec:ppad-membership,sec:tbsg,sec:pure} and 
additional technical facts and settings are in \Cref{app:qmcounterexample,ssec:bellman,apdx:average}.

\section{Preliminaries}\label{sec:prelim}

Here we introduce notation and basic concepts for $\ssg$s and $\tbsg$s we use throughout the paper.

\paragraph{Simultaneous stochastic games ($\ssg$s).} This paper focuses on computing NE of \emph{multi-agent general-sum simultaneous stochastic games ($\ssg$s)} in infinite-horizon settings.
Unless stated otherwise, we consider discounted infinite-horizon $\ssg$s and denote an instance by tuple $\mathcal{G} = (n, \calS ,\mathcal{A},\pp,\r,\gamma)$. $n$ denotes the  number of players (agents), $\calS$ denotes a finite state space, and $\calA$ denotes the finite set of actions available to the players where for player $i\in [n]$ and $s \in \calS$ the possible actions of player $i$ at states $s$ are $\calA_{i,s}$. We say player $i\in[n]$  \emph{controls state $s\in\calS$} if $\calA_{i,s}\neq\emptyset$. We use $\calI_s = \{i\in[n]|\calA_{i,s}\neq\emptyset\}\subseteq [n]$ to denote the players controlling state $s$, and $\calA_s$ to denote the joint action space of all players controlling state $s$, i.e.\ for any $\aa_s \in\calA_s$,  $\aa_s = (a_{i,s})_{i\in\calI_s}$ where $a_{i,s}\in\calA_{i,s}$.  
We denote the action space size for player $i$ by $\Atoti := \sum_{s \in \calS} |\calA_{i,s}|$ and the joint action space size by $\Atot := \sum_{i \in [n]} \Atoti$. 
We let $\pp$ denote the transition probabilities, where $\pp_{s,\aa_s}\in\Delta^\calS \defeq \{\xbold \in \R^\calS_{\geq 0} | \sum_{s \in \calS} x_s = 1\}$ is a distribution over states for all $s\in\calS$ and $\aa_s\in\calA_s$.  $\r$ denotes the instantaneous rewards, where $r_{i,s,\aa_s}$ with $|r_{i,s,\aa_s}| \leq 1$  is the reward of player $i$ at state $s$ if the players controlling it play $\aa_s\in\calA_s$. $\gamma \in (0, 1)$ denotes a discount factor.

\paragraph{$\ssg$ notation and simplifications.} Recall that we use $\calI_s = \{i\in[n]|\calA_{i,s}\neq\emptyset\}\subseteq [n]$ to denote the players controlling state $s$. Additionally, we use $\calS_i=\{s\in\calS|\calA_{i,s}\neq\emptyset\}\subseteq\calS$ to denote states that are controlled by player $i$. Without loss of generality, we assume that for each player $i$ there exists at least one state $s\in\calS$ where $\calA_{i,s}\neq\emptyset$ (i.e.\ $|\calS_i|\ge 1$), since otherwise we can remove the corresponding player $i$ from the game. Also, we assume for each state $s\in\calS$ there is at least a player $i$ such that $\calA_{i,s}\neq\emptyset$ (i.e.\  $|\calI_s|\ge 1$). This is because for any $s\in\calS$, if $\calA_{i,s}=\emptyset$ for all $i\in[n]$ and the transition from the state is $\pp_s \in \Delta^\calS$, this is equivalent to setting $\calA_{1,s}=\{a_s\}$ and $\pp_{s,a_s}=\pp_s$.

\paragraph{$\ssg$ model and objectives.} A $\ssg$ proceeds as follows. It starts from time step $t=0$ and initial state $s^0 \in \calS$ drawn from initial distribution $\qq$.
In each turn $t \geq 0$ the game is at a state $s^t$. At state $s^t$, each player $i\in\calI_{s^t}$ plays an action $a^t_{i}\in \calA_{i, s^t}$. The joint action $\aa^{t} = (a^t_{i})_{i\in \calI_{s^t}}\in\calA_{s^t}$ then yields reward $r_{i,s^t, \aa^t}$ for each player $i\in[n]$.The next state $s^{t+1}$ is then sampled (independently) by $\pp_{s^t,\aa^t}\in\Delta^\calS$. The goal of each player $i \in [n]$ is to maximize their expected infinite-horizon discounted reward, or known as \emph{\emph{value} of the game for player $i$}, defined as 
$v_i = \E[\sum_{t\ge 0}\gamma^t r_{i,s^t, \aa^t}]$.%

\paragraph{$\ssg$ policies and strategies.} Unless stated otherwise, for each player  $i\in[n]$ we restrict to considering randomized stationary policies, i.e. $\ppi_i = \Par{\ppi_{i,s}}_{s\in\calS_i}$ where $\ppi_{i,s}\in\Delta^{\calA_{i,s}}$, and use $\pi_{i,s}(a)$ to denote the probability of player $i$ playing action $a\in\calA_{i,s}$ at state $s$. We call a collection of policies for all players, i.e. $\ppi=(\ppi_i)_{i\in[n]}$, a \emph{strategy}. For a strategy $\ppi$ we use $\ppi_{-i}$ to denote the collection of policies of all players other than player $i$, i.e. $\ppi_{-i} \defeq (\ppi_j)_{j\in [n]\setminus\{i\}}$; we do not distinguish between orders of $\ppi_{i,s}$ in the set $\ppi$ when clear from context (e.g. see definition of NE in \eqref{def:NE-approx}). Further, we use $\PP^{\ppi}\in\R^{\calS\times\calS}$  and $\r^{\ppi}$ to denote the probability transition kernel and instantaneous reward, respectively, under strategy $\ppi$, where
\begin{equation}\label{def:prelim-pi-r}
	\begin{gathered}
		\PP^{\ppi}(s,\cdot) \defeq  \sum_{\aa_s\in\calA_s}\Par{\prod_{i\in\calI_s}\pi_{i,s}(a_{i,s})}\pp_{s,\aa_s}\in\Delta^{\calS}
		\text{ and }
		\r_i^{\ppi}(s)  \defeq \sum_{\aa_s\in\calA_s}\Par{\prod_{i\in\calI_s}\pi_{i,s}(a_{i,s})}r_{i,s,\aa_s}.
	\end{gathered}
\end{equation}
Under strategy $\ppi$, we define the value function of each player $i \in [n]$ at state $s \in \calS$ to be 
\begin{equation}\label{eq:def-V-i-pi}
V_i^{\ppi}(s) \defeq \E\left[\sum_{t\ge 0}\gamma^t r_{i,s^t,\aa^t}|s_0 = s, a^t_{j,s^t}\sim \ppi_{j,s^t}~\text{for all}~j,t\right]
= \ee_s^\top (\II - \gamma \PP^{\ppi})^{-1} \r^{\ppi} ~.
\end{equation}
The value of a strategy to player $i$ starting at initial distribution $\qq\in\Delta^\calS$ is defined as
\begin{equation}\label{eq:def-v-i-pi}
v_i^{\ppi,\qq} \defeq \E_\qq^{\ppi}\left[\sum_{t\ge 0}\gamma^t r_{i,s^t,\aa^t}\right] = \E\left[\sum_{t\ge 0}\gamma^t r_{i,s^t,\aa^t}|s_0\sim \qq, a^t_{j,s^t}\sim \ppi_{j,s^t}~\text{for all}~j,t\right] = \langle \qq,\VV_i^\pi\rangle.
\end{equation}

\paragraph{Nash equilibrium (NE) in $\ssg$s.} Given any $\epsilon\ge 0$, we call a strategy $\ppi$ an \emph{ $\eps$-approximate  Nash Equilibrium (NE) ($\epsilon$-NE)} if for each player $i\in[n]$ 
\begin{align}\label{def:NE-approx}
	& u_i(\ppi_i, \ppi_{-i})\ge u_i(\ppi_i', \ppi_{-i})-\epsilon,~~\text{for any}~~ \ppi_{i,s}' \in \Delta^{\calA_{i,s}}.
\end{align}
where $u_i(\ppi)$ for all $i \in [n]$ is a real value (as a function of $\ppi$) referred to as \emph{utility of player $i$ under strategy $\ppi$}. For general $\ssg$s, unless specified otherwise, we let $u_i(\ppi) \defeq u_i(\ppi_i, \ppi_{-i}) = \inprod{\qq}{\VV_i^{\ppi}}$, i.e.\ the value function with initial distribution $\qq = \frac{1}{|\calS|}\ee_{\calS}$. Further, we call any $0$-approximate NE an \emph{exact NE} and when we refer to a NE we typically mean an $\epsilon$-NE for inverse-polynomially small $\epsilon$. We use the term \emph{approximate NE} to refer to an $\epsilon$-NE for constant $\epsilon$.

\paragraph{Turn-based stochastic games ($\tbsg$s).} $\tbsg$s are the class of $\ssg$s where each state is controlled by at most one player (i.e.\ $|\calI_s| \leq 1$), or equivalently, the states controlled by each of the players are disjoint (i.e\ $\calS_i\cap\calS_j=\emptyset$ for any $i\neq j$, $i,j\in[n]$). Equivalently (by earlier assumptions), a $\tbsg$ is a  $\ssg$ with $|\calI_s| = 1$ for all $s \in \calS$; accordingly, we use $\calI_s = \{i_s\}$ to denote the single player that is controlling state $s$ in a $\tbsg$. 
Since  $\calS=\cup_{i\in[n]}\calS_i$ in a $\tbsg$ we denote an instance by $\mathcal{G} = (n,\calS = \cup_{i\in[n]}\mathcal{S}_i,\mathcal{A},\pp,\r,\gamma)$. Following $\ssg$ notation, we have $\calA=(\calA_{i,s})_{i\in[n], s\in\calS_i}$ and $\calA_s = \calA_{i,s}$ if and only if $s\in\calS_i$ as well as $\pp = (\pp_{s,a})_{s\in\calS, a_s\in\calA_s}$ and $\r=(r_{i,s,a_s})_{i\in[n],s\in\calS, a_s\in\calA_s}$. 
When clear from context, we also use $\ppi_s \defeq \ppi_{i_s,s}$ for all $s\in\calS$. Using this notation, the probability transition kernel $\PP^{\ppi}\in\R^{\calS\times\calS}$ and instantaneous reward $\r^{\ppi}\in\R^\calS$ under strategy $\ppi$ are
\begin{equation}\label{def:prelim-pi-r-tbsg}
\begin{gathered}
		\PP^{\ppi}(s,\cdot) = \sum_{a_{s}\in\calA_s}\pi_{i_s,s}(a_{s})\pp_{s,a_s}\in\Delta^{\calS}
		\text{ and }
		\r_i^{\ppi}(s)  = \sum_{a_s\in\calA_s}\pi_{i_s,s}(a_{s})r_{i,s,a_s}\,.
\end{gathered}
\end{equation}

\paragraph{Game variations.} Here we briefly discuss variants of discounted $\ssg$s we consider.

\begin{itemize}
	\item \textbf{Number of players}: We focus on $n$-player games and our hardness results use that $n$ can scale with the problem size. Establishing the complexity of computing general-sum NE for $\tbsg$s with a constant number of players, e.g.\ $n=2$, remains open.

\item \textbf{Number of states each player controls}: We use $\onestatessg$ (and $\onestate$) to denote the class of $\ssg$s (and $\tbsg$s) where each player only controls one state $|\calS_i|=1$ (note that it is possible that $|\calI_s|>1$, for some $s\in\calS$ in an $\onestatessg$). For simplicity, in $\onestate$ instances we denote the state space by $\calS_i = \{s_i\}$ for each $i\in[n]$ and thus $\calS  = \cup_{i\in[n]}\{s_i\}$ and let $\calA_i \defeq \calA_{s_i} = \calA_{i,s_i}$. %
For $\onestatessg$s and $\onestate$s, we use $\outility_i(\ppi) \defeq V_i^{\ppi}(s_i)$ to denote the value of player $i$ under strategy $\pi$ with initial distribution $\ee_{s_i}$. 
Unless specified otherwise, we use $\outility(\cdot)$ as the utility function in the definition of NE for $\onestatessg$s and $\onestate$s; in \Cref{ssec:bellman} we prove that these two notions of approximate NE are equivalent up to polynomial factors.%

\item \textbf{Different types of strategies (and policies).} We focus on stationary strategies in the majority of this paper, but at times we consider non-stationary strategies where the distribution over actions chosen at each time-step is allowed to depend on $t$. 
Further,  we call a policy $\ppi_i$ a \emph{pure (or deterministic) policy} if it maps a state to a single action for that player, i.e.\ if $\ppi_{i,s} = \ee_{a_{i,s}}$ for some $a_{i,s}\in\calA_{i,s}$ for each  $s\in\calS_i$ and call a strategy $\ppi = (\ppi_i)_{i\in[n]}$  a \emph{pure strategy} if all policies $\ppi_{i}$ are pure. 
Some of the results in paper restrict to consider pure strategies and we extend the definitions of NE to these cases by restricting to such strategies in \eqref{def:NE-approx}.

\end{itemize}

\section{Overview of results and techniques}\label{sec:overview}

Here we provide an overview of our main results and techniques for establishing the complexity of computing stationary NEs in discounted infinite-horizon general-sum $\ssg$s. First, in \Cref{sec:foundation} we cover foundational structural results regarding such $\ssg$s that we use throughout the paper.
In \Cref{sec:resultsssg} we discuss how we show that the problem of computing stationary NE in such $\ssg$s is in $\ppad$. We then consider the $\tbsg$ specialization of this problem and discuss how we show that computing stationary NE in such $\tbsg$s is $\ppad$-hard (\Cref{sec:resultstbsg}), but polynomial-time solvable under additional assumptions on rewards  (\Cref{sec:resultslocalrewards}). 

Although we focus on discounted $\ssg$s and $\tbsg$s in the body of the paper, in \Cref{apdx:average}, we extend our results to the \emph{average-reward model} where the rewards are not discounted, but instead amortized over time. We show that results analogous to our main results hold for under the assumption of bounded mixing times. These extensions are achieved by building upon tools established in~\cite{jin2021towards} for related discounted and average-reward MDPs. 

\subsection{Foundational properties}\label{sec:foundation}

Here we introduce two types of foundational structure we demonstrate for infinite-horizon $\ssg$s, which both our positive and negative complexity characterizations crucially rely on. These structures use the fact that when fixing the strategies of all but one of the players in a $\ssg$, the problem reduces to a single-agent DMDP.

The first property we observe is that when changing the action of a player at any single state in a $\ssg$ from one distribution to another, the utility for that player changes in a \emph{monotonic} manner, with slope that is both upper and lower-bounded. 
We refer to this type of change as \emph{pseudo-linear}. 
This property is equivalent to showing the following theorem that the utilities are pseudo-linear in the special class of $\ssg$s where each player controls only one state, i.e.\ $\onestatessg$s.

\begin{theorem}[Pseudo-linear utilities in $\onestatessg$s, restating~\Cref{coro:quasi-mono-ssg}]\label{thm:quasi-mono-ssg} 
	Consider any $\onestatessg$ instance $\mathcal{G} = (n, \calS ,\mathcal{A},\PP,\RR,\gamma)$ any initial distribution $\qq$, and some player $i\in[n]$. Her utility function $u_i(\ppi_i, \ppi_{-i}) = v_i^{\ppi,\qq}$, when fixing other players' strategy $\ppi_{-i}$, is \emph{pseudo-linear} in $\ppi_i$, i.e.\  for any $\ppi_i$, $\ppi_i'\in\Delta^{\calA_i}$ ordered such that $u_i(\ppi_i,\ppi_{-i}) \leq u_i(\ppi_i',\ppi_{-i})$ and any $\theta\in[0,1]$, we have 
\begin{equation}\label{eq:def-quasi-mono-bounded-velo-overview}
(1-\gamma)\theta (u_i(\ppi'_i,\ppi_{-i})-u_i(\ppi))\le u_i(\theta\ppi'_i+(1-\theta)\ppi_i,\ppi_{-i})-u_i(\ppi)\le \frac{1}{1-\gamma}\theta (u_i(\ppi'_i,\ppi_{-i})-u_i(\ppi)).
\end{equation}
\end{theorem}

First, to see why $\onestatessg$s have pseudo-linear utilities, we note that when we consider a linear combination of policies $\ppi_i'$ and $\ppi_i$ for player $i$ and fix the other players' strategy $\ppi_{-i}$, it is equivalent to considering a DMDP in which a single player linearly changes her policy on a single state $s$ between two actions $a$ and $a'$.
In this case, the difference in transition matrices is of rank-$1$ and we can use the Sherman-Morrison formula to exactly characterize the change in utility as
\begin{align}
	&u_i(\theta\ppi'_i+(1-\theta)\ppi_i,\ppi_{-i}) = u_i(\ppi_i,\ppi_{-i}) +\theta \frac{\qq^\top \QQ \ee_s \cdot \left[(r_{s,a'} - r_{s,a})+\gamma  (\pp_{s,a'}-\pp_{s,a})^\top \QQ \r^{\ppi} \right]}{1-\gamma\theta  (\pp_{s,a'}-\pp_{s,a})^\top \QQ \ee_s} \label{eq:nonconvexity}
\end{align}
where $\QQ \defeq \left(\II-\gamma\PP^{(\ppi_i,\ppi_{-i})}\right)^{-1}$. 
We then bound the difference in utilities arising from changing the transitions; we show that $(\pp_{s,a'}-\pp_{s,a})^\top \QQ \ee_s \in[-1/(1-\gamma),1]$ by utilizing a specific Markov chain interpretation of the utilities.
This implies the more fine-grained property in~\eqref{eq:def-quasi-mono-bounded-velo-overview} that the utility function is \emph{pseudo-linear} with bounded slope.

Theorem~\ref{thm:quasi-mono-ssg} describes powerful structure on the utility functions of each player that we leverage for our membership \emph{and} hardness results. Although utilities for $\ssg$s may be non-linear and non-convex (in fact, even under a single-state policy change, \eqref{eq:nonconvexity} may be either convex or concave in $\theta$ depending on the sign of $D$ and $[u_i(\ppi'_i,\ppi_{-i})-u_i({\ppi})]$) with complex global correlations, for any fixed player Theorem~\ref{thm:quasi-mono-ssg}  shows that utilities are not too far from linear.

This pseudo-linear structure is key to many of our subsequent proofs. For example, the pseudo-linear property in~\eqref{eq:def-quasi-mono-bounded-velo-overview} implies a distinct proof of existence of NE for $\onestatessg$ that is considerably simpler than the classic existence proofs for $\ssg$s~\cite{fink1964equilibrium,takahashi1964equilibrium}.
We describe how pseudo-linearity is used in each of our proofs of membership of $\ssg$ (\Cref{sec:resultsssg}), hardness of $\tbsg$ (\Cref{sec:resultstbsg}) and polynomial-time algorithms for pure NE in special cases (\Cref{sec:resultslocalrewards}).
While pseudo-linearity is useful for several of our results, it appears to tie closely with NE in $\onestatessg$s\footnote{Monotonicity structure in stochastic games has been studied previously, and~\cite{lozovanu2018stationary} claimed that a version of this structure holds for all $\tbsg$s, including ones in which one player can control multiple states.
However, it appears that the restriction to $\onestatessg$ or (equivalently) considering the change in actions only at a single state is key and we prove in \Cref{app:qmcounterexample} that without this, monotonicity may not hold.
}.
To leverage the pseudo-linearity property more broadly for $\ssg$s, we make the following important structural observation of $\ssg$s, which implied that computing an approximate NE of general $\ssg$ ($\tbsg$) instances is polynomial-time reducible to computing an approximate NE of some corresponding $\onestatessg$ ($\onestate$) instances.
\begin{theorem}[Approximate-NE equivalences for $\ssg$s and $\onestatessg$s, restating~\Cref{lem:redx-MtO}]\label{lem:redx-MtO-overview}
	   There exists a linear-time-computable mapping between the original $\ssg$ and a linear-time-computable corresponding $\onestatessg$ instance, such that for any $\eps \geq 0$ a strategy $\ppi$ is an $\epsilon$-approximate mixed NE of the original $\ssg$ if its induced policy $\ppi'$ is a $((1-\gamma)\epsilon/|\calS|)$-approximate mixed NE in the corresponding $\onestatessg$ (Definition~\ref{def:mspp_ospp_ssg}).
\end{theorem}

To prove \Cref{lem:redx-MtO-overview}, we leverage a key property of the induced single-player MDP for player $i$ when the other players' policies $\ppi_{-i}$ are fixed: the policy improvement property of coordinate-wise (i.e.\ \emph{asynchronous}) policy iteration~\cite{Bertsekas19,puterman2014markov}. In general, \Cref{lem:redx-MtO-overview} implies that an algorithm applicable to all $\onestatessg$ instances can also be adapted to solve $\ssg$ instances. This allows us to transfer the benefits of pseudo-linearity in the more specialized $\onestatessg$ classes to all infinite-horizon $\ssg$s, despite the absence of monotonicity structure in the latter.

\subsection{Complexity of NE in $\ssg$s}\label{sec:resultsssg}

Here we describe how we leverage our structural results on infinite horizons $\ssg$s to show that computing NE of $\ssg$s is in $\ppad$ and thereby obtain a full complexity characterization of such games (they are $\ppad$-complete). Our main complexity result for $\ssg$s is \Cref{thm:ppad-membership-overview}.

\begin{theorem}[Complexity of NE in $\ssg$, restating~\Cref{thm:ppad-membership}]\label{thm:ppad-membership-overview}
The problem of computing an $\eps$-approximate NE for infinite-horizon $\ssg$ class is $\ppad$-complete for a polynomially-bounded discount factor $\frac{1}{1 - \gamma} = \poly(\Atot)$ and accuracy $\eps = \Omega(1/\poly(\Atot))$.%
\end{theorem}

Showing hardness in \Cref{thm:ppad-membership} is relatively trivial: it follows immediately by considering $\gamma \rightarrow 0$ and noting that choosing the optimal stationary policy for one step involves computing a NE for an arbitrary multiplayer normal-form game, which is known to be $\ppad$-hard~\cite{daskalakis2009complexity}.

The more interesting component  of the proof of  \Cref{thm:ppad-membership} is the proof of $\ppad$ membership. This proof is provided in~\Cref{ssec:ppad-membership} and leverages the foundational structure of $\ssg$ discussed in~\Cref{sec:foundation} and additional properties of DMDPs.
In particular, making use of the Brouwer fixed point argument (see, e.g.~\cite{dasgupta1986existence}) that shows the existence of NE, we construct two different types of Brouwer functions on strategies as below:
\begin{align}
f_\Vsf:~&\ppi\rightarrow \mathbf{y}~\text{such that}~y_{i,a}(\ppi) = \frac{\pi_{i}(a) + \max(u_i(\ee_{a},\ppi_{-i}) - u_i(\ppi),0)}{1 + \sum_{a' \in \calA_{i,s}} \max(u_i(\ee_{a'},\ppi_{-i}) - u_i(\ppi),0)}~\text{for}~\onestatessg,\label{eq:Brouwer-functions}\\
f_\osBsf:~&\ppi\rightarrow \mathbf{y}~\text{such that}~y_{i,s,a}(\ppi) = \frac{\pi_{i,s}(a) + \max([r_{i,s,a} + \gamma \pp_{s,a}^\top \VV_i^{\ppi}]-V_i^{\ppi}(s),0)}{1 + \sum_{a' \in \calA_{i,s}} \max([r_{i,s,a} + \gamma \pp_{s,a}^\top \VV_i^{\ppi}]-V_i^{\ppi}(s),0)}~\text{for}~\ssg~.\nonumber
\end{align}
Both of these functions satisfy the property that $f(\ppi) = \ppi$ if and only if $\ppi$ is a NE, and are reminiscent of the Brouwer functions used in original $\ppad$-membership arguments that are tailored to linear utilities~\cite{daskalakis2009complexity}. Each function leads to a different  $\ppad$-membership proof and we include both due to the interesting distinct properties of $\ssg$s that they utilize.

Our proof based on $f_\Vsf$ uses both the linear-time equivalence between $\onestatessg$ and $\ssg$ provided in~\Cref{lem:redx-MtO-overview}, and the pseudo-linear structure of  $\onestatessg$ utilities in~\Cref{thm:quasi-mono-ssg}. 
The most non-trivial step involves showing that approximate Brouwer fixed points correspond to approximate NE (Lemma~\ref{lem:brouwerapproxNE}), for which we critically use our established property of pseudo-linearity.
We also show that this proof strategy generalizes to show $\ppad$-membership of any $n$-player-$k$-action game with pseudo-linear utilities (\Cref{thm:ppad-membership-general}, under other mild conditions), which we think may be of independent interest. 

Our alternative proof based on $f_\osBsf$ builds upon the structural fact that small Bellman errors suffice to argue about approximation of NE in~\eqref{lem:bellman-suff} of~\Cref{ssec:bellman}.
Here the crucial observation is that fixing all other players' policies, the Bellman errors are linear in policy-space for a single player.
As a consequence we can apply the more standard analysis~\cite{daskalakis2009complexity} to argue that when $\ppi$ is an approximate fixed point of $f_\osBsf$, the Bellman update error $\max([r_{i,s,a} + \gamma \pp_{s,a} V_i^{\ppi}]-V_i^{\ppi}(s)$ is close to $0$.
This in turn maps back to an approximate NE using the sufficient conditions on Bellman-error for NE (\Cref{ssec:bellman}).

\vspace{-1em} 
\paragraph{Clarification and contextualization with recent prior work~\cite{deng2021complexity}:} 

After initial drafting of this manuscript, we were pointed to the recent work of~\cite{deng2021complexity}, which claims to have already shown the $\ppad$-membership of general $\ssg$s.
However, we were unable to verify their proof; in particular, we do not know how to derive the $6$-th line from the $5$-th line in proving Case 2 of Lemma 4 in~\cite{deng2021complexity} (analogous to our Lemma~\ref{lem:brouwerapproxNE}).
Like us, the authors of~\cite{deng2021complexity} also use the Brouwer function $f_\Vsf$ (more commonly known as Nash's Brouwer function and originally designed for linear utilities); however, unlike us, they do not establish or use any special pseudo-linear structure on the value functions.
In our proof of Lemma~\ref{lem:brouwerapproxNE}, this structure is key to establishing $\ppad$-membership and used for the most non-trivial part of the proof --- that the approximate fixed points of the Brouwer function $f_\Vsf$ are equivalent to approximate NE. 

\subsection{Complexity of NE in $\tbsg$s}\label{sec:resultstbsg}

Here we consider the specialization of infinite-horizon $\ssg$s to $\tbsg$s. 
Recall that in a $\tbsg$, each state is controlled by only one player and, thus, players take turns in controlling the Markov process.
We ask the fundamental question, \emph{how hard is it to compute stationary NE in $\tbsg$s}?

Unlike their non-turn-based counterparts, it is no longer clear that this problem is $\ppad$-hard: the aforementioned direct encoding of NE of arbitrary two-player normal-form games no longer applies when $|\calS| = 1$ or $\gamma \to 0$.
Moreover, in~\Cref{sec:nonstat} and~\Cref{ssec:lpalg} we show that approximate NE computation for $\tbsg$ is in polynomial-time if: (a) non-stationary NE are allowed, \emph{or} (b) the number of states $|\calS|$ is held to a constant; note that equilibrium computation for $\ssg$s remains $\ppad$-hard even under these simplifications. 

Though prior work on general-sum $\tbsg$s is limited, the special case of 2-player zero-sum $\tbsg$s has been well studied~\cite{condon1992complexity,shapley1953stochastic,hansen2013strategy,sidford2020solving} and are known to possess additional structure beyond $\ssg$s.
For example,~\cite{shapley1953stochastic} showed that a pure NE always exists for zero-sum $\tbsg$s and~\cite{hansen2013strategy} showed that NE is computable in strongly polynomial time when the discount factor is constant.
However, this structure does not carry over to the general-sum case and \cite{zinkevich2006cyclic} shows that there are $\tbsg$s with only mixed NE (which hints at possible hardness). 
In~\Cref{ssec:ppad-hard}, we prove the following theorem and establish $\ppad$-hardness of computing NEs of $\tbsg$s.

\begin{theorem}[Complexity of $\tbsg$  NEs, restating \Cref{thm:PPAD-hard}, informal]\label{thm:ppad-hardness-overview} 
Approximate NE-computation in infinite-horizon $\gamma$-discounted $\tbsg$s with any $\gamma\in[1/2,1)$ is $\ppad$-complete. 
\end{theorem}   
We prove \Cref{thm:PPAD-hard} by reducing the problem of generalized approximate circuit satisfiability ($\eps$-$\gcircuit$, formally defined in~\Cref{def:gcircuit}) to $\onestate$s; $\eps$-$\gcircuit$ is known to be $\ppad$-hard for even sufficiently small constant $\eps > 0$~\cite{rubinstein2018inapproximability}.
This reduction is, at a high-level, the approach taken in the first proofs of $\ppad$-hardness of normal-form games~\cite{chen2006settling,daskalakis2009complexity} as well as more recent literature (e.g.~hardness for public goods games~\cite{papadimitriou2021public}); though it has been predominantly applied to games with linear or piecewise linear utilities.
The key ingredients of our reduction are the implementation of certain circuit gates, i.e. $G_=$ (equal), $G_\alpha$ (set to constant $\alpha$), $G_{\times}$ (multiply), $G_+$ (sum), $G_{-}$ (subtraction), $G_{>}$ (comparison), $G_{\land}$ (logic AND), $G_{\lor}$ (logic OR), $G_{\neg}$ (logic NOT), through $\onestate$ \emph{game gadgets} which carefully encode these gates in an $\onestate$. 

As an illustration, here we show how to implement an approximate equal gate $G_=$ between input and output players (corresponding to input and output states), i.e. $p\subout \in [p\subin - \eps, p\subin + \eps]$ at \emph{any} approximate NE.
This gadget includes $3$ players (states): $\ind$ $(s\subin)$, $\out$ $(s\subout)$ and $\aux$ $(s\subaux)$. Figure~\ref{fig:mixedgamegadgets-intro} illustrates the transitions in the $\tbsg$ instance and Table~\ref{tab:equal-reward-aux-intro} partially specifies the instantaneous rewards. 
Here, the reader should think of $p\subin$ as the probability of player $\mathsf{in}$ choosing action $a^1\subin$ and $p\subout$ as the probability of player $\out$ choosing action $a^1\subout$.
Our game gadgets are crucially \emph{multiplayer} in that they allow flexible choice of instant rewards for different players (e.g.\ $\mathsf{in}$, $\aux$ and $\mathsf{out}$).
We consider the case of exact NE as a warmup; in particular, we hope to show that exact NE necessitates $\pi\subin(a^1\subin) = \pi\subout(a^1\subout)$.
Just as in the typically implemented graphical game gadgets~\cite{daskalakis2009complexity,chen2006settling}, our hope is to enforce this equality constraint through a proof-by-contradiction argument that goes through two steps.
As an illustration of the contradiction argument, suppose that $\pi\subin(a^1\subin) > \pi\subout(a^1\subout)$.
Our optimistic hope would be to choose the rewards and transitions so that the value function of player $\aux$ at his own state under choice of $\ppi\subout, \ppi\subin$ satisfies
\begin{align}\label{eq:overview-V-value}
V\subaux^{(\ee_{a\subaux^1}, \ppi\subin, \ppi\subout)} = \gamma \ppi\subin(a^1\subin)~~\text{ and }~~V\subaux^{(\ee_{a\subaux^2}, \ppi\subin, \ppi\subout)} = \gamma\ppi\subout(a^1\subout).
\end{align}
If~\eqref{eq:overview-V-value} were satisfied, $\aux$ player would have to take pure strategy $a^1\subaux$ at exact NE, which would transit to state $s\subin$. 
As reflected in the reward table (\Cref{tab:equal-reward-aux-intro}) this would be a bad event for player $\mathsf{out}$ due to the negative reward she accrues at state $s\subin$.
Consequently, she would prefer to take action $a^1\subout$ as much as possible, i.e. $\pi\subout(a^1\subout) \approx 1$, which would lead to the desired contradiction.
(A symmetric contradictory argument would work for the case $\pi\subin(a^1\subin) < \pi\subout(a^1\subout)$, ensuring that the system balances and necessitates $\pi\subin = \pi\subout$ at an exact NE.) 

\begin{figure}[t]
	\centering
		\includegraphics[width = 0.7\textwidth]{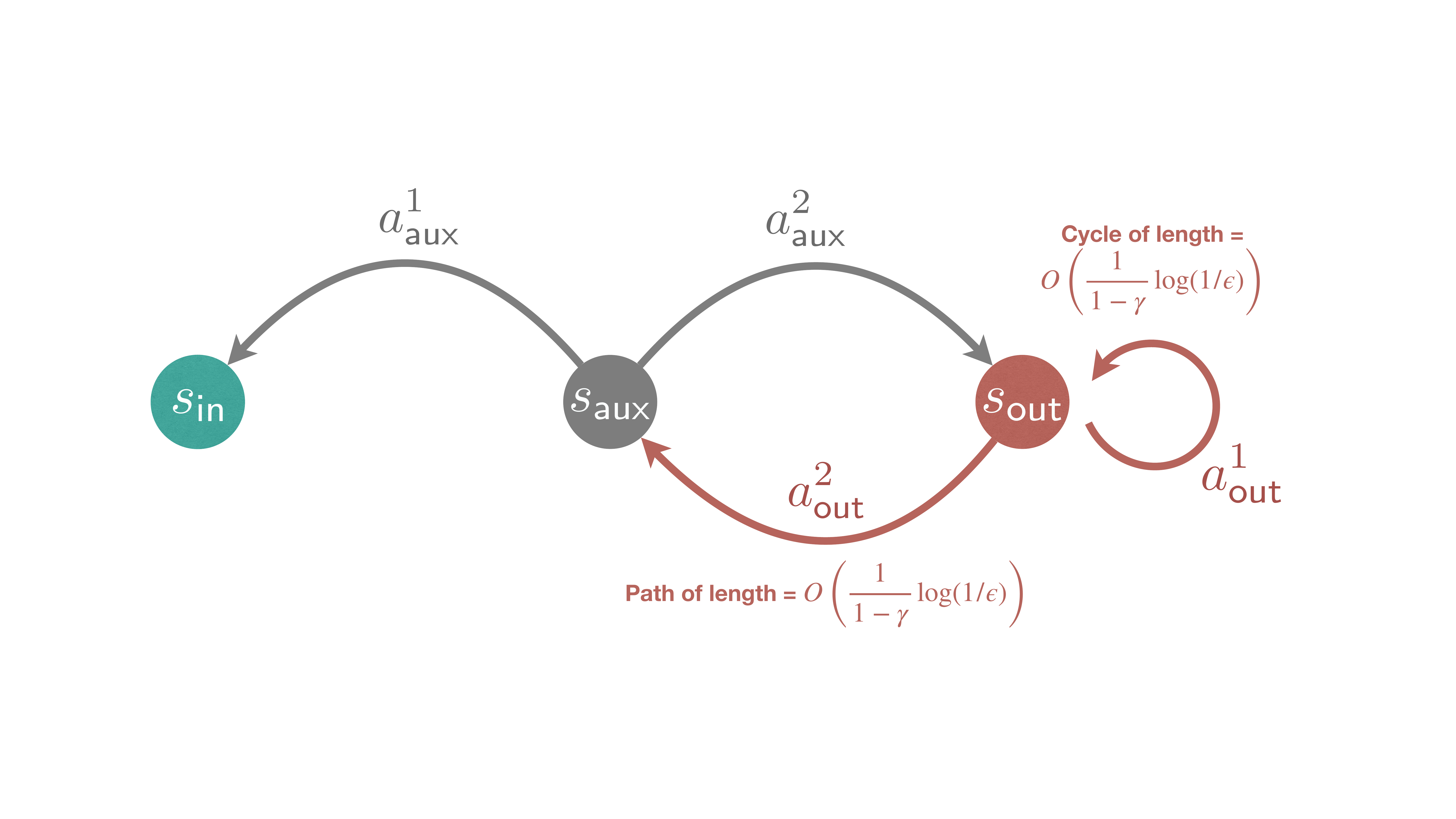}
    \captionof{figure}{Illustration of states and transitions for ``equal gadget'' to implement $G_{=}$. The transitions in red encode a cycle or path of length $L$, where $L = \lceil \frac{4}{1 - \gamma} \log (1/\eps)\rceil$ for some constant $\gamma$, $\epsilon$.
	}\label{fig:mixedgamegadgets-intro}
\end{figure}
 
 \begin{table}[t]
\begin{center}
\begin{tabular}{llll}
\multicolumn{1}{l|}{}     & \multicolumn{1}{l|}{$s\subin,a\subin^1$} & \multicolumn{1}{l|}{$s\subin,a\subin^2$} & \multicolumn{1}{l}{$s\subaux, a^1\subaux$} \\ \hline
\multicolumn{1}{l|}{$\aux$}   &
\multicolumn{1}{l|}{$1/2$} & \multicolumn{1}{l|}{$0$}   & 
\multicolumn{1}{l}{$0$}   \\ \hline
\multicolumn{1}{l|}{$\mathsf{out}$}  &
\multicolumn{1}{l|}{$-1/4$}   &
\multicolumn{1}{l|}{$-1/4$} & \multicolumn{1}{l}{$0$}   
\end{tabular}
\caption{ Instantaneous rewards of player $\mathsf{aux},\mathsf{out}$, informal: Constant $\gamma$ is the given problem discount factor, $\gamma^L\le \epsilon^2$ is tiny.}\label{tab:equal-reward-aux-intro}
\end{center}
\end{table}

However, creating an equal gadget through $\onestate$ is much more intricate than a corresponding graphical game gadget due to the twin challenges of \emph{nonlinearity} and \emph{common structure} in players' utilities.
For one, the pseudo-linear structure described in~\Cref{sec:foundation} only ensures approximate linearity up to multiplicative constants; the more fine-grained equality required in~\eqref{eq:overview-V-value} is far more difficult to achieve (and unclear whether possible).
Moreover, unlike the definitional local structure between players in graphical games~\cite{kearns2013graphical}, $\tbsg$s have significant \emph{global} structure between players (as players represent states that transit to one another).
In other words the players' utility functions depend on \emph{all} of the other players and not just their immediate neighbors.
As a consequence of this global structure, a naive combination of individual gadgets could sizably change the value functions and break the local circuit operations.

We work around these two issues by creating long cycles and paths with ``dummy states" for the actions the $\mathsf{out}$ player takes such that the only non-zero rewards are collected outside these dummy states. 
We show that this elongation of paths simultaneously induces \emph{approximate linearity} and \emph{localization to neighbors} in the $\onestate$ instance.
When the path has length $L = O(\frac{1}{1-\gamma}\log(\frac{1}{\eps}))$, we satisfy \eqref{eq:overview-V-value} in an approximate sense up to tiny  $\poly(\eps)$ errors.
Further, this almost-linear structure turns out to be robust to transitions that are ``further away'' from $s\subout$.
This ensures that the $\mathsf{out}$ player can then be used as an input for subsequent gadgets connected in series, and enables a successful combination of the gadgets without changing the NE conditions at each state.

It remains to translate these ideas from an exact-NE argument to an approximate-NE argument.
For this, the pseudo-linearity property that we established in~\Cref{ssec:quasi-mono} proves to be especially useful.
In particular, when $\ppi\subin(a\subin^1)>\ppi\subout(a^1\subout)+\epsilon$, we can adapt the bounded ``slope'' argument in~\eqref{eq:def-quasi-mono-bounded-velo-overview} and observe that
\begin{align*}
V\subaux^{(\ee_{a^1\subaux}, \ppi\subin, \ppi\subout)} - V\subaux^{(\theta \ee_{a^1\subaux}+(1-\theta)\ee_{a^2\subaux}, \ppi\subin, \ppi\subout)}& \ge(1-\theta)(1-\gamma)\left[V\subaux^{(\ee_{a^1\subaux}, \ppi\subin, \ppi\subout)}-V\subaux^{(\ee_{a^2\subaux}, \ppi\subin, \ppi\subout)}\right]\\
& \ge (1-\theta)\gamma (1-\gamma)\epsilon,
\end{align*}
which ensures that $\theta$, i.e. the probability that player $\aux$ takes action $a^1\subaux$, must be close enough to $1$ at any approximate NE. We use this pseudo-linearity multiple times to formally relax the exact NE argument under approximation, in order to implement $G_{=}$ gate for $\epsilon$-$\gcircuit$. 

Ultimately, our proof of this theorem sheds further light on the problem's structure and shows that hardness is fairly resilient in general-sum stochastic games. Even in the special case where each player controls a single state and receives non-zero reward at at most $4$ states (\emph{or} alternatively, all players have non-negative but dense reward structure), the problem is still $\ppad$-hard. 

\paragraph{Contextualization with independent concurrent work:}
In independent and concurrent work, the authors of~\cite{daskalakis2022complexity} were additionally able to prove that the computation of NE of even $2$-player $\tbsg$s is $\ppad$-hard. 
We note that beyond claims of hardness in $\tbsg$s, each of \cite{daskalakis2022complexity} and this work contain disjoint results of independent interest. For instance,~\cite{daskalakis2022complexity} provides a polynomial-time algorithm for finding \emph{non-stationary} Markov CCEs for $\ssg$s.
On the other hand, this work focuses exclusively on stationary equilibrium concepts.
In addition to hardness of $\tbsg$s, we show the $\ppad$-membership of general games with pseudo-linear utilities including $\ssg$s (see~\Cref{sec:resultsssg}) and provide polynomial-time algorithms for finding stationary NEs for $\tbsg$s under extra assumptions on the reward structure (see~\Cref{sec:resultslocalrewards}). 

\subsection{Efficient algorithms for $\tbsg$s under localized rewards}\label{sec:resultslocalrewards}

As shown above, the problem of finding an approximate NE for infinite-horizon $\tbsg$s is $\ppad$-complete even under a variety of additional structural assumptions.
For example we show that even when each player only controls one state, all transitions are deterministic, and all players receive non-negative rewards (possibly in many states) computing a NE in a $\tbsg$ is $\ppad$-hard. 

Towards characterizing what features are critical to the hardness of the problem, we specialize further and ask what happens if we further restrict each player to receive reward only at the \emph{single state} that they control. We call this class of games $\ronly$ and consider the class of \emph{fixed-sign} $\ronly$ $\onestate$, i.e.\ $r_{i,s_i,\cdot}\ge0$ (or $r_{i,s_i,\cdot}\le0$) for all $i\in[n]$ and $r_{i,s',\cdot}=0$ for any $s'\neq s_i$.%

The intuitive reason for why this special structure is helpful is that it creates a qualitative symmetry in the players' incentives: all of them wish to either reach (in the case of non-negative rewards) or avoid (in the case of negative rewards) their own controlling state. Mathematically, we observe that  given a strategy $\ppi$, the utility function has the following structure 
\begin{align}
u_i(\ppi) & = V^{\ppi}_i	=  \ee_{s_i}^\top (\II-\gamma\PP^{\ppi})^{-1}\r^{\ppi} = \frac{E_i(\ppi_{-i})}{\det(\ppi)}r_i^{\ppi_i}(s_i)
\text{ where }
\det(\ppi) \defeq \det\left(\II-\gamma\PP^{\ppi}\right)>0
\label{eq:loc_utility}
\end{align}
and $E_i(\ppi_{-i})$ is the determinant of the $(i,i)^{th}$ minor of the matrix $\II-\gamma\PP^{\ppi}$.  
The last equality for $u_i(\ppi)$ in \eqref{eq:loc_utility} used the matrix inversion formula and the fact that $r_i^{\ppi_j}(s_j)=0$ for any $j\neq i$. 
Consequently, the numerator of $u_i(\ppi)$ is separable in $\ppi_i$ and $\ppi_{-i}$ and the denominator is common to all players $i \in [n]$.
This implies that, following a logarithmic transformation, a fixed-sign $\ronly$ $\onestate$ game is equivalent to a \emph{potential game}~\cite{monderer1996potential}, with potential function $\Phi(\ppi) \defeq \log (\det(\ppi)^{-1} \prod_{i\in[n]}r_i^{\ppi_i}(s_i))$.
That is, for any $i\in[n]$ and $\ppi_i$, $\ppi_i'$ we have $\Phi(\ppi_i',\ppi_{-i}) - \Phi(\ppi) = \log V_i^{(\ppi_i'.\ppi_{-i})}-\log V_i^{(\ppi_i,\ppi_{-i})}$.

The potential game structure of $\ronly$ $\onestate$ automatically implies the existence of a pure NE for all fixed-sign $\ronly$ games.
Further, it follows~\cite{monderer1996potential} that (approximate) best response dynamics (also known as \emph{strategy iteration} in the stochastic games literature~\cite{hansen2013strategy})  provably decrease this potential by a polynomial factor, until it achieves a pure-strategy approximate NE.
This yields a polynomial-time algorithm for computing approximate NE as stated below.

\begin{theorem}[Restating~\Cref{lem:nonneg-existence,lem:pureinP_potential,lem:nonpos-existence,lem:pureinP_potential-nonpos}]\label{thm:positive}
Consider a $\ronly$ $\onestate$ instance $\calG  = (n,\calS = \cup_{i\in[n]}\{s_i\},\mathcal{A},\pp,\r,\gamma)$ where all rewards are non-negative (or non-positive). Then the game has a pure NE, and given some accuracy $\epsilon$, approximate best-response dynamics find an $\epsilon$-approximate pure NE in time $\mathrm{poly}(\Atot, \frac{1}{1-\gamma}, \frac{1}{\epsilon})$.
\end{theorem}

We also show that under further assumptions it is possible to compute an exact NE through a different set of algorithms inspired by graph problems.
Specifically, we consider a special sub-class of fixed-sign $\ronly$ $\onestate$ where we impose two additional structural assumptions: (a) all transitions are deterministic, and (b) all rewards on each player's own state are independent of actions. 
Under these refinements, all players in an non-negative $\ronly$ instance are incentivized to go through a shortest-cycle to maximize its utility, while all players in an non-positive $\ronly$ instance are incentivized to go through a cycle that is as long as possible (or, most ideally follow a path to a cycle that doesn't return to the player's controlled state). 
Accordingly, we design graph algorithms (\Cref{alg:ronly-non-negative-graph} and \Cref{alg:ronly-non-positive-graph}) that locally, iteratively find the cycle and path structure that corresponds to an exact NE. 
Our results show that best-response dynamics (i.e.~strategy iteration) and graph-based algorithms can work in general-sum $\tbsg$s beyond zero-sum setting~\cite{hansen2013strategy}.

Finally, note that our positive results really require \emph{both} of the assumptions of (a) reward only at a single state (b) rewards of the same sign (see \Cref{tab:complexity-loc} for a summary). 
From~\Cref{sec:resultstbsg} we already know when relaxing the first condition, finding an approximate mixed NE is $\ppad$-hard. 
The second condition is also important, as relaxing it (i.e.\ allowing both positive and negative rewards in the $\ronly$ $\onestate$ model) may preclude even the existence of pure NE~\cite{zinkevich2006cyclic}. In fact, we show in~\Cref{ssec:mixed-sign} that even determining whether or not a pure NE exists is NP-hard via a reduction to the Hamiltonian path problem (but whether \emph{mixed} NE are polynomial-time computable under this modification remains open).
Ultimately, this gives a more complete picture of what transformations change the problem from being $\ppad$-complete to being polynomial time solvable.

\begin{table}[t]
\begin{center}
\begin{small}
\begin{tabular}{|l|c|c|}
\hline
Setting ($\onestate$) &  Localized rewards ($\ronly$) & General rewards \\
\hline\hline
Fixed-sign rewards & Polynomial (pure NE) & $\ppad$-complete  \\
\hline
Mixed-sign rewards & NP-hard (pure NE), open problem (mixed NE) & $\ppad$-complete   \\
\hline
\end{tabular}
\end{small}
\caption{ 
Summary of complexity characterization for $\onestate$ under various reward assumptions.
\label{tab:complexity-loc}}
\end{center}
\vskip -0.1in
\end{table}

\section{Related work}\label{sec:related-work-short}

Here we highlight prior work that is most closely related to our results. 

\paragraph{General-sum stochastic game theory:}
Central questions in stochastic game theory research involve (a) the existence of equilibria and (b) the convergence and complexity of algorithms that compute these equilibria.
Existence of equilibria is known in significantly more general formulations of stochastic games than the tabular $\ssg$s that are studied in our paper (see,~e.g.~the classic textbooks~\cite{filar2012competitive,bacsar1998dynamic}).
Relevant to our study, the first existence proofs of general-sum tabular stochastic games appeared in~\cite{fink1964equilibrium,takahashi1964equilibrium}.
They are based on Kakutani's fixed point theorem, and so \emph{non-constructive} in that they do not immediately yield an algorithm.
This is a departure from the zero-sum case, where Shapley's proof of existence~\cite{shapley1953stochastic} is constructive and directly leverages the convergence of infinite-horizon dynamic-programming.

Indeed, the recent survey paper on multi-agent RL~\cite{zhang2021multi} mentions the search for computationally tractable and provably convergent (to NE) algorithms for $\ssg$ as an open problem.
Algorithms that are known to converge to NE in $\ssg$ require strong assumptions on the heterogeneous rewards --- such as requiring the one-step equilibrium to be unique at each iteration~\cite{hu2003nash,greenwald2003correlated}, or requiring the players to satisfy a ``friend-or-foe'' relationship~\cite{littman2001friend}.

An important negative result in the literature was the shown failure of convergence of infinite-horizon dynamic-programming algorithms for general-sum $\ssg$s~\cite{zinkevich2006cyclic}.
More generally, they uncover a fundamental identifiability issue by showing that more than one equilibrium value (and, thereby, more than one NE) can realize identical action-value functions.
This identifiability issue suggests that any iterative algorithm that uses action-value functions in its update (including policy-based methods like policy iteration and two-timescale actor-critic~\cite{konda1999actor}) will fail to converge for similar reasons. 

Since then, alternative algorithms that successfully \emph{asymptotically} converge to NE have been developed for general-sum $\ssg$s based on two-timescale approaches~\cite{prasad2015two} and homotopy methods~\cite{borkovsky2010user,herings2004stationary}.
However, these algorithms are intricately coupled across players and states in a more intricate way and, at the very least, suffer a high complexity per iteration.
Finite-time guarantees for these algorithms do not exist in the literature.
A distinct approach that uses linear programming is also proposed~\cite{dermed2009solving}, but this algorithm also suffers from exponential iteration complexity.
Algorithms that are used for general-sum $\ssg$ in practice are largely heuristic and directly minimize the Bellman error of the strategy~\cite{perolat2017learning} (which we defined in~\Cref{ssec:bellman}).

Interestingly, this picture does not significantly change for $\tbsg$s despite their significant structure over and above $\ssg$s.
The counterexamples of~\cite{zinkevich2006cyclic} are in fact $2$-player, $2$-state, and $2$-action-per-state $\tbsg$s.
Our $\ppad$-hardness results for $\tbsg$ resolve an open question that was posed by~\cite{zinkevich2006cyclic}, who asked whether alternative methods (using Q-values and equilibrium-value functions) could be used to derive stationary NE in $\tbsg$ instead.
In particular, we show that the stationary NE is not only difficult to approach via popular dynamics, but is fundamentally hard.

Very recently, a number of positive results for finite-horizon non-stationary CCE in $\ssg$s were provided~\cite{song2021can,jin2021v,mao2022provably}.
These results even allow for independent learning by players.
A natural question is whether an infinite-horizon \emph{stationary} CCE could be extracted from these results.
Since $\tbsg$ NE is a special case of $\ssg$ CCE, our $\ppad$-hardness result answers this question in the negative.
In general, tools that are designed for computing and approaching non-stationary equilibria cannot be easily leveraged to compute or approach stationary equilibria due to the induced nonconvexity in utilities and the failure of infinite-horizon dynamic programming.
Our paper fills this gap and provides a comprehensive characterization of complexity of computing stationary NE for infinite-horizon multi-player $\ssg$s and $\tbsg$s.

A trivial observation is that the problem of exact computation for general-sum stochastic games is only harder than approximation; in general, exact computation for NE of stochastic games is outside the scope of this paper and we refer readers to~\cite{filos2022fixp} for recent hardness result following that thread.

\paragraph{The zero-sum case:}
There is a substantial literature on equilibrium computation, sample complexity and learning dynamics in the case of zero-sum $\ssg$ and $\tbsg$.
For a detailed overview of advances in learning in zero-sum stochastic games, see the survey paper~\cite{zhang2021multi}. 
In contrast, our results address the general-sum case.
Positive results for zero-sum $\tbsg$, such as the property of strongly-polynomial-time computation of an exact NE with a constant discount factor~\cite{hansen2013strategy}, leverage special structure that does not carry over to the general-sum case.
In particular, a pure NE always exists for a zero-sum $\tbsg$ owing to the convergence of Shapley's value iteration~\cite{shapley1953stochastic}.
~\cite{zinkevich2006cyclic} showed that a pure NE need not exist for general-sum $\tbsg$s.
We further show in~\Cref{sec:pure} that pure NE are NP-hard to compute (at least in part due to their possible lack of existence).
On the more positive side, we also characterize specializations of general-sum $\tbsg$s for which pure NE always exist and are polynomial-time computable.

It is crucial to note that our results only address the equilibrium computation problem of general-sum $\ssg$s and $\tbsg$s with a constant discount factor.
When the rewards are zero-sum this is known to be polynomial-time~\cite{shapley1953stochastic} and additionally strongly polynomial-time in the case of $\tbsg$~\cite{hansen2013strategy}.
Whether it is possible to compute an (exact or approximate) NE in even zero-sum $\tbsg$s with an increasing discount factor remains open~\cite{andersson2009complexity}.
This open problem has important connections to simple stochastic games~\cite{condon1992complexity,etessami2010complexity}, mean-payoff games~\cite{gurvich1988cyclic,zwick1996complexity}, and parity games~\cite{emerson1991tree,voge2000discrete,jurdzinski2008deterministic}.

\paragraph{Algorithmic game theory for normal-form and market equilibria:} 
The $\ppad$ complexity class was introduced by~\cite{papadimitriou1994complexity} to capture the complexity of all total search problems (i.e. problems for which a solution is known)~\cite{megiddo1991total} that are polynomial-time reducible to the problem of finding at least one unbalanced vertex on a directed graph.
~\cite{daskalakis2009complexity} first showed that NE computation for $n$-player $k$-action normal-form games lies in $\ppad$.
By definition, the utilities of normal-form games are always \emph{linear} in the mixed strategies.
This is not the case for $\ssg$ or $\tbsg$, whose utilities are not even convex in their argument.
The membership of nonconvex general-sum games in $\ppad$ is not obvious.
For example,~\cite{daskalakis2021complexity} recently showed $\ppad$-hardness of even \emph{zero-sum} constrained nonconvex-nonconcave games.
Moreover, the complexity of all general-sum games satisfying a succinct representation and the property of polynomial-time evaluation of expected utility (which includes $\ssg$) is believed to lie in a strictly harder complexity class than $\ppad$~\cite{schoenebeck2012computational}.
General-sum nonlinear game classes that are known to be in $\ppad$ primarily involve market equilibrium~\cite{chen2009settling,vazirani2011market,chen2017complexity,garg2017settling} and Bayes-NE of auctions~\cite{filos2021complexity} and make distinct assumptions of either a) \emph{separable concave and piecewise linear (SPLC)} assumptions on the utilities or b) \emph{constant-elasticity-of-substitution (CES) utilities}~\cite{chen2017complexity}.
They also utilize in part linearity in sufficient conditions for NE (e.g. Walras's law for market equilibrium).
These structures, while interesting in their own right, are also not satisfied by $\ssg$s or $\tbsg$s.
The pseudo-linear property of $\ssg$s that we uncover in~\Cref{ssec:quasi-mono} is key to showing $\ppad$-membership.
Our subsequent proof in~\Cref{ssec:ppad-membership} is a useful generalization of the traditional proof for linear utilities~\cite{daskalakis2009complexity} to pseudo-linear utilities.

In addition to being in $\ppad$, general-sum normal-form games were established to be $\ppad$-hard by~\cite{chen2006settling,daskalakis2009complexity}.
Since then, $\ppad$-hardness has been shown for several structured classes of normal-form games~\cite{mehta2014constant,liu2018approximation,chen2015complexity,deligkas2020tree,papadimitriou2021public} as well as for weaker objectives in normal-form games such as constant-additive approximation~\cite{daskalakis2013complexity,rubinstein2016settling,rubinstein2018inapproximability} and smoothed-analysis~\cite{boodaghians2020smoothed}.
Our approach to prove $\ppad$-hardness for $\tbsg$ takes inspiration from the approach to prove $\ppad$-hardness for $n$-player graphical games~\cite{kearns2007graphical,kearns2013graphical} (which was subsequently used to prove $\ppad$-hardness for constant-player normal-form games by~\cite{daskalakis2009complexity}).
In particular we construct \emph{game gadgets} to implement real-valued arithmetic circuit operations through $\tbsg$ NE.
As summarized in~\Cref{sec:resultstbsg}, the details of our $\tbsg$ game gadgets are significantly more intricate than the corresponding graphical game gadgets due to the additional challenges of global shared structure across players and the nonlinearity of the utilities.
These challenges do not manifest in graphical games as, by definition, they only possess local structure and satisfy linearity in utilities.
Whether $\tbsg$s are directly reducible to graphical games or bimatrix games remains an intriguing open question.

\paragraph{Relation of $\tbsg$ to other game-theoretic paradigms:}

We conclude our overview of related work with a brief summarization of solution concepts and paradigms that are partially related to $\tbsg$'s.
First, the class of sequential or \emph{extensive-form games} is known to lie in $\ppad$~\cite{young2014handbook} and is trivially $\ppad$-hard due to normal-form games being a special case. 
We note that computation of non-stationary equilibria in the finite-horizon $\ssg$ and $\tbsg$ are special cases of these. 
Second, the solution concept of (coarse) correlated equilibrium (CCE) is polynomial-time computable, in contrast with NE, even for multiplayer games with linear utilities~\cite{papadimitriou2008computing}.
Since $\tbsg$ involves a non-trivial action set for only one player at each state, the solution concepts of NE and CCE all become equivalent for both stationary and non-stationary equilibria. 
On the positive side, this may imply the convergence of recently designed finite-horizon learning dynamics~\cite{song2021can,jin2021v,mao2022provably} to $\tbsg$ NE.
On the negative side, our $\ppad$-hardness of approximation of stationary NE in $\tbsg$ (Section~\ref{ssec:ppad-hard}) implies hardness of \emph{stationary} CCE equilibria in $\ssg$s.
Finally, we contextualize our NP-hardness results on certain decision problems (i.e. \emph{does there exist an equilibrium with certain properties?}) in~\Cref{ssec:mixed-sign}.
In normal-form games, such decision problems are known to be NP-hard~\cite{conitzer2002complexity}.
\section{Foundational properties of $\ssg$}\label{sec:axiom}

We begin by discussing some structural properties of $\ssg$ and $\tbsg$.
These structural properties all follow from the observation that when fixing strategies of all other players $\ppi_{-i}$, the game degenerate to a single-player Markov decision process for player $i$. 
Such structure has the following implications: 
\begin{enumerate}
\item It ensures the utilities to be all monotonic (either decreasing or increasing) with upper and lower-bounded slope, which we call \emph{pseudo-linearity}, along any linear path between two policies changing for a single player, as long as each player only controls a single state (i.e. the instance is in $\onestatessg$).
This result is shown in~\Cref{ssec:quasi-mono}.
\item It allows us to show a polynomial-time equivalence between computing an approximate NE in a general instance $\ssg$ (in which a player may control multiple states) and computing an approximate NE in a correspondingly defined instance of $\onestatessg$.
This result is shown in~\Cref{ssec:redx-m2o}.
\end{enumerate}

These foundational lemmas are repeatedly used in subsequent sections.

\subsection{Pseudo-linear utilities for $\onestatessg$}\label{ssec:quasi-mono}

In this section we show a strong version of quasi-monotonicity property of the utility functions of players in a $\onestatessg$ instance, which we refer to as \emph{pseudo-linearity} throughout the paper.

As a stepping stone to this result, we first prove this property for a single-player Markov decision process (MDP) for which only one state has a non-trivial action space of size $two$.
We formally define this type of MDP below.

\begin{definition}[Two Action MDP]\label{def:twoactionmdp}
A \emph{two-action MDP} $\calM = (\calS,\calA,\pp,\r,\gamma)$ fixes a single state $s \in \calS$ for which $\calA_s = \{a,a'\}$ and sets $|\calA_{s'}| = 1$ for all $s' \neq s$.
Rewards and transition probabilities for state $s \in \calS$ and action $a \in \calA_s$ are denoted by $r_{s,a}$ and $\pp_{s,a}$ respectively.
It suffices to consider the continuum of policies $\ppi^\theta$ such that $\pi_s(a') = \theta \in [0,1]$ (note that by the definition of this simplified MDP, this is the only state at which the policy needs to be specified).
We define as shorthand the corresponding value of the policy $\ppi^\theta$ (starting at initial distribution $\qq$ over states) as $v^{(\theta),\qq}$.
\end{definition}

\begin{lemma}[Pseudo-linearity of value in two-action MDP]\label{lem:quasi-mono}
	The value of a policy $\ppi^\theta$ for any two-action MDP is monotonic in $\theta \in [0,1]$.
	Moreover, whenever $v^{(1),\qq}\neq v^{(0),\qq}$ we have
	\begin{equation}\label{eq:bounded-velo}
		\frac{v^{(\theta),\qq}-v^{(0),\qq}}{v^{(1),\qq}-v^{(0),\qq}} \in\left[(1-\gamma) \theta
		~ , ~ 
		\frac{\theta}{1-\gamma}  \right],~\forall \theta\in[0,1].
	\end{equation}
	In other words, the monotonicity is strict with an upper and lower-bounded slope.
\end{lemma}

\begin{proof}
	We recall the expression of the value function of MDP to be
	\begin{equation}\label{def:mdp-v}
		\begin{aligned}
			v^{(\theta),\qq} &= \qq^\top \left(\II-\gamma\PP^{(\theta)}\right)^{-1}\r^{(\theta)},\\
			~\text{where we define}~~\PP^{(\theta)}(s,\cdot) & \defeq \theta \pp_{s,a'} + (1 - \theta) \pp_{s,a} ~\text{for}~s,~\text{and}~\PP^{(\theta)}(s',\cdot) \defeq \pp_{s'}~\text{for}~s'\neq s,\\
			\r^{(\theta)}(s) & \defeq \theta r_{s,a'} + (1 - \theta) r_{s,a}~\text{for}~s,~\text{and}~\r^{(\theta)}(s') \defeq r_{s'}~\text{for}~s'\neq s.
		\end{aligned}
	\end{equation}
	By the Sherman-Morrison-Woodbury formula specialized to a rank-$1$ update, we have
	\begin{align*}
		\left(\II-\gamma\PP^{(\theta)}\right)^{-1} & = \left(\II-\gamma\PP^{(0)}-\gamma \theta \ee_s(\pp_{s,a'}-\pp_{s,a})^\top\right)^{-1}\\
		& = \left(\II-\gamma\PP^{(0)}\right)^{-1}+\frac{\gamma\theta\left(\II-\gamma\PP^{(0)}\right)^{-1}\ee_s(\pp_{s,a'}-\pp_{s,a})^\top\left(\II-\gamma\PP^{(0)}\right)^{-1}}{1-\gamma\theta(\pp_{s,a'}-\pp_{s,a})^\top\left(\II-\gamma\PP^{(0)}\right)^{-1}\ee_s}
	\end{align*}
	Plugging this back into the value function expression (Equation~\eqref{def:mdp-v}), we then get
	\begin{align*}
		& v^{(\theta),\qq} = \qq^\top\left(\II-\gamma\PP^{(0)}\right)^{-1}\r^{(\theta)}+\frac{\gamma\theta\cdot\qq^\top \left(\II-\gamma\PP^{(0)}\right)^{-1}\ee_s(\pp_{s,a'}-\pp_{s,a})^\top\left(\II-\gamma\PP^{(0)}\right)^{-1}\r^{(\theta)} }{1-\gamma\theta \left(\pp_{s,a'}-\pp_{s,a}\right)^\top \left(\II-\gamma\PP^{(0)}\right)^{-1}\ee_s}\\
		& \hspace{2em} = v^{(0),\qq}+\theta\left(r_{s,a'}-r_{s,a}\right)\cdot\qq^\top \left(\II-\gamma\PP^{(0)}\right)^{-1}\ee_s\\
		& \hspace{2em}+\frac{\gamma\theta\cdot\qq^\top \left(\II-\gamma\PP^{(0)}\right)^{-1}\ee_s(\pp_{s,a'}-\pp_{s,a})^\top\left(\II-\gamma\PP^{(0)}\right)^{-1}\r^{(\theta)} }{1-\gamma\theta \left(\pp_{s,a'}-\pp_{s,a}\right)^\top \left(\II-\gamma\PP^{(0)}\right)^{-1}\ee_s}.
	\end{align*}
	Now, we define as shorthand $E^{(0),\qq}(s) = \qq^\top \left(\II-\gamma\PP^{(0)}\right)^{-1}\ee_s$.
	Note that $E^{(0),\qq}(s)$ is simply the \emph{expected visitation frequency} of state $s$ for initial distribution $\qq$ and probability transition $\PP^{(0)}$.
	Similarly, we define $\Delta E^{(0)}(s)=(\pp_{s,a'}-\pp_{s,a})^\top\left(\II-\gamma\PP^{(0)}\right)^{-1}\ee_s$ and note that this is the \emph{difference} of expected visitation frequency of state $s$ between initial distribution $\pp_{s,a'}$ and $\pp_{s,a}$.
	Finally, we define $\Delta V^{(0)} = (\pp_{s,a'}-\pp_{s,a})^\top\left(\II-\gamma\PP^{(0)}\right)^{-1}\r^{(0)}$ to be the difference of the value of policy $\ppi^0$ between initial state distributions $\pp_{s,a'}$ and $\pp_{s,a}$. 
	Using this notation, the equality above simplifies to
	\begin{equation}\label{eq:quasi-mono-main}
		\begin{aligned}
			v^{(\theta),\qq} & = v^{(0),\qq} + \theta\left(r_{s,a'}-r_{s,a}\right)\cdot E^{(0),\qq}(s)+\frac{\gamma\theta\cdot E^{(0),\qq}(s)\cdot\left( \Delta V^{(0)}+\theta\left(r_{s,a'}-r_{s,a}\right)\Delta E^{(0)}(s)\right)}{1-\gamma\theta \cdot \Delta E^{(0)}(s)}\\
			& = v^{(0),\qq}+ \theta\cdot \frac{ E^{(0),\qq}(s)\cdot \left[\left(r_{s,a'}-r_{s,a}\right)+\gamma \Delta V^{(0)}\right]}{1-\gamma\theta \cdot \Delta E^{(0)}(s)}.
		\end{aligned}
	\end{equation}
	
	Next, we bound the quantity $1-\gamma\theta \Delta E^{(0)}(s)$, which will imply that $1-\gamma\theta \cdot \Delta E^{(0)}(s)>0$ for any $\gamma\in[0,1)$ and any $\theta\in[0,1]$.
	Since the numerator of the right hand side of Equation~\eqref{eq:quasi-mono-main} only depends on $\theta$ linearly, monotonicity follows immediately as a result.

	We proceed to bound $1 - \gamma \theta \Delta E^{(0)}(s)$.
	To see this we note that by definition of $\Delta E^{(0)}(s)$ and its meaning in terms of expected visitation, we have 
	\begin{align*}
	& \pp_{s,a'}^\top\left(\II-\gamma\PP^{(0)}\right)^{-1}\ee_s  = \E^{\ppi^{(0)}}_{\pp_{s,a'}}\left[\sum_{t\ge 0}\gamma^t\mathbf{1}_{s^t = s}\right]\\
	& \hspace{2em} = \E^{\ppi^{(0)}}_{\pp_{s,a'}}\left[\sum_{t\ge 0}\gamma^t\mathbf{1}_{s^t = s}|\text{will visit state }s~\right]+ \E^{\ppi^{(0)}}_{\pp_{s,a'}}\left[\sum_{t\ge 0}\gamma^t\mathbf{1}_{s^t = s}|\text{will never visit state }s~\right]\\
	& \hspace{2em} \le  \mathbb{P}\left(\text{starting from}~\pp_{s,a'}~\text{will visit}~s\right)\cdot \E_{\ee_{s}}^{\ppi^{(0)}}\left[\sum_{t\ge 0}\gamma^t\mathbf{1}_{s^t = s}\right] + 0\\
	& \hspace{2em} \le p\cdot \ee_s^\top \left(\II-\gamma\PP^{(0)}\right)^{-1}\ee_s.
	\end{align*}
Here for the last inequality we let $p = \mathbb{P}\left(\text{starting from}~\pp_{s,a'}~\text{will visit}~s\right)\in[0,1]$. This further impliesx
	\begin{align*}
		0\le \pp_{s,a'}^\top \left(\II-\gamma\PP^{(0)}\right)^{-1}\ee_s & \le  p	\cdot \ee_s^\top \left(\II-\gamma\PP^{(0)}\right)^{-1}\ee_s \\
		& = p\cdot \ee_s^\top \left(\II+\gamma \PP^{(0)}\left(\II-\gamma \PP^{(0)}\right)^{-1}\right)\ee_s = p\left(1+\gamma \pp_{s,a}^\top \left(\II-\gamma\PP^{(0)}\right)^{-1}\ee_s\right),
	\end{align*}
	where for the last equality we utilize the fact that $\pp_{s,a} = \ee_s^\top \PP^{(0)}$. 
	Subtracting $\pp_{s,a}^\top \left(\II-\gamma\PP^{(0)}\right)^{-1}\ee_s$ on both sides, we get
	\begin{align*}
		\Delta E^{(0)}(s) & = \pp_{s,a'}^\top \left(\II-\gamma\PP^{(0)}\right)^{-1}\ee_s -\pp_{s,a}^\top \left(\II-\gamma\PP^{(0)}\right)^{-1}\ee_s \\
		& \le  p\left(1+\gamma \pp_{s,a}^\top \left(\II-\gamma\PP^{(0)}\right)^{-1}\ee_s\right) - \pp_{s,a}^\top \left(\II-\gamma\PP^{(0)}\right)^{-1}\ee_s\in \left[-\frac{1}{1-\gamma},p\right]\subseteq\left[-\frac{1}{1-\gamma},1\right],
	\end{align*}
	Above, the last step uses the Cauchy-Schwarz inequality and the fact that $\norm{\II-\gamma \PP^{(0)}}_\infty\le \frac{1}{1-\gamma}$. 
	Noting that $\gamma, \theta \leq 1$ immediately implies that $1-\gamma\theta\cdot \Delta E^{(0)}(s)>0$.
	This completes the proof of monotonicity.

	We now show the bounded-slope property.
	Substituting $\theta = 1$ into Equation~\eqref{eq:quasi-mono-main} yields
	\begin{align}\label{eq:quasi-mono-main1}
	v^{(1),\qq}-v^{(0),\qq} = \frac{ E^{(0),\qq}(s)\cdot \left[\left(r_{s,a'}-r_{s,a}\right)+\gamma \Delta V^{(0)}\right]}{1-\gamma \cdot \Delta E^{(0)}(s)}.
	\end{align}
	Then, as long as $v^{(1),\qq}\neq v^{(0),\qq}$, we can divide Equation~\eqref{eq:quasi-mono-main} by Equation~\eqref{eq:quasi-mono-main1} to get
	\[
	\frac{v^{(\theta),\qq}-v^{(0),\qq}}{v^{(1),\qq}-v^{(0),\qq}} = \frac{\theta( 1-\gamma \cdot \Delta E^{(0)}(s))}{1-\gamma\theta \cdot \Delta E^{(0)}(s)}\in\left[(1-\gamma)\theta, \left(\frac{1}{1-\gamma}\right) \theta\right].
	\]
	This is exactly Equation~\eqref{eq:bounded-velo} and completes the proof of the lemma.
\end{proof}
We now leverage Lemma~\ref{lem:quasi-mono} to show quasi-monotonicity of each player's policy for any $\onestatessg$ instance.
This immediately follows as a corollary from the following observation for $\onestatessg$: Consider a player $i$ and her controlling state $s_i$.
Fix the other players' policies $\ppi_{-i}$, and consider two candidate policies for player $i$ denoted by $\ppi$ and $\ppi'$.
Then, the induced MDP for player $i$ is a two-action MDP in the sense of Definition~\ref{def:twoactionmdp}, and pseudo-linearity follows as an immediate consequence.
The complete statement of pseudo-linearity (quasi-monotonicity with bounded slope) for $\onestatessg$ is provided below.

\begin{corollary}[Pseudo-linear utility under $\onestatessg$]\label{coro:quasi-mono-ssg}
	Consider an $\onestatessg$ instance $\mathcal{G} = (n, \calS ,\mathcal{A},\PP,\RR,\gamma)$, and any initial distribution $\qq$, recall the definition of the utility function $u_i(\ppi) = v_i^{\ppi,\qq}$ given any strategy $\ppi$. 
	Fix a player $i$ and other players' policies $\ppi_{-i}$, the the player $i$'s utility function $u_i(\ppi_i, \ppi_{-i})$ is quasi-monotonic in $\ppi_i$.
	In other words, for any two candidate policies $\ppi_i$, $\ppi_i'\in\Delta^{\calA_i}$ and any $\theta \in [0,1]$, we have
	\begin{equation}\label{eq:def-quasi-mono}
		\min\{u_i(\ppi_i,\ppi_{-i}),u_i(\ppi_i',\ppi_{-i})\}\le u_i(\theta\ppi_i+(1-\theta)\ppi_i',\ppi_{-i})\le \max\{u_i(\ppi_i,\ppi_{-i}),u_i(\ppi_i',\ppi_{-i})\}.
	\end{equation}
	Further when $u_i(\ppi_i,\ppi_{-i})\neq u_i(\ppi_i',\ppi_{-i}) $, we have
	\begin{equation}\label{eq:def-quasi-mono-bounded-velo}
		\frac{u_i(\theta\ppi'_i+(1-\theta)\ppi_i,\ppi_{-i})-u_i(\ppi)}{u_i(\ppi'_i,\ppi_{-i})-u_i(\ppi)} \in\left[(1-\gamma) \theta, \left(\frac{1}{1-\gamma}\right) \theta\right],~\text{for any}~\ppi_{i}, \ppi'_{i}\in\Delta^\calA_{i} ,~\theta\in[0,1].
	\end{equation}
\end{corollary}
\begin{proof}
	Consider the original instance in $\onestatessg$.
	Fix a player $i$, the unique state that it controls (denoted by $s_i$), and other players' policies $\ppi_{-i}$.
	Corresponding to two candidate policies of player $i$, $\ppi_i$ and $\ppi_i'$, we construct the following two-action MDP instance (Definition~\ref{def:twoactionmdp}) for player $i$ with the following specifications:
	\begin{itemize}
		\item The state space $\calS$ is the same as the state space $\calS$ of original instance $\calG$. 
		\item 	Player $i$ has two actions: $a$ (which corresponds to $\ppi_i$), and $a'$ (which corresponds to $\ppi'_i$), at state $s=s_i$. 
		For all other states $s'\neq s_i$, player $i$ has only one (degenerate) action.
		\item For actions $a,a'$ taken at state $s$, the respective transition probabilities are given by 
		\[
		\pp_{s,a} = \PP^{(\ppi_i,{\ppi_{-i}})}(s,\cdot), ~~\pp_{s,a'} = \PP^{(\ppi'_i,{\ppi_{-i}})}(s,\cdot).
		\]
		For all other states $s'\neq s$, the transition probability is given by  $\pp_{s'} = \PP^{\ppi}(s',\cdot)$. Similarly, the reward at state $s$ is given by $r_{s,a} = r^{(\ppi_i,{\ppi_{-i}})}_{i}(s)$ when taking action $a$ and $r_{s,a'} = r^{(\ppi'_i,{\ppi_{-i}})}_{i}(s)$ when taking action $a'$. 
		For $s'\neq s$ the reward is given by $r_{s'} = r^{\ppi}_{i}(s')$.
	\end{itemize}
	
	It is then immediate to see that $u_i(\ppi) = v_i^{\ppi,\qq}$ for any strategy $\ppi$. Consequently, we conclude from~\Cref{lem:quasi-mono} that $u_i(\theta\ppi'_i+(1-\theta)\ppi_i,\ppi_{-i})$ is monotonic in $\theta$, which then gives
	\[
	\min\{u_i(\ppi_i,\ppi_{-i}),u_i(\ppi_i',\ppi_{-i})\}\le u_i(\theta\ppi_i+(1-\theta)\ppi_i',\ppi_{-i})\le \max\{u_i(\ppi_i,\ppi_{-i}),u_i(\ppi_i',\ppi_{-i})\}.
	\]
	Similarly, the second claim (Equation~\eqref{eq:def-quasi-mono-bounded-velo}) directly follows from Equation~\eqref{eq:bounded-velo}.
	This completes the proof of the corollary.
\end{proof}

We note that the pseudo-linearity property crucially relies on the fact that one player only controls a single state, and is not true for general $\ssg$s. Appendix~\ref{app:qmcounterexample} provides an explicit counterexample in the case where one player can control multiple states.
For completeness, we specialize the statement to $\onestate$, which is just a special case of~\Cref{coro:quasi-mono-ssg}.

\begin{corollary}[Psuedo-linear utility under~$\onestate$]
	\label{lem:quasi-mono-tbsg}
	Given any $\onestate$ $\mathcal{G} = (n,\calS = \cup_{i\in[n]}\mathcal{S}_i,\mathcal{A},\pp,\r,\gamma)$ and some initial distribution 
	$\qq\in \Delta^{\calS}$. Given anystrategy $\ppi$, define the utility function $u_i(\ppi) = v_i^{\ppi,\qq}$ under initial distribution $\qq$. Then we have  $u_i(\ppi_i,\ppi_{-i})$ is \emph{quasi-monotonic} in $\ppi_i$, as defined in~\eqref{eq:def-quasi-mono}. Further~\Cref{eq:def-quasi-mono-bounded-velo} also holds true.
\end{corollary}

\subsection{Reductions to $\onestatessg$}\label{ssec:redx-m2o}

In this section, we show that one can reduce finding an $\epsilon$-approximate (exact) mixed NE of general $\tbsg$ ($\ssg$) to finding an $\Theta(\epsilon)$-approximate (exact) mixed NE of $\onestate$ ($\onestatessg$). We state and prove the reduction for the most general case of $\ssg$, and the reduction for $\tbsg$ immediately applies as a special case.

We provide a useful way to create a $\onestatessg$ from any $\ssg$ instance below.
\begin{definition}\label{def:mspp_ospp_ssg}
Consider an instance of $\ssg$ given by $\mathcal{G} = (n, \calS ,\mathcal{A},\pp,\r,\gamma)$.
We create a corresponding instance in $\onestatessg$ given by  $\mathcal{G}' = (n|\calS|, \calS, \mathcal{A}', \pp',\r',\gamma)$ with the following set of properties:
\begin{itemize}
\item The number of players in the $\onestatessg$ instance is equal to $n |\calS|$, and we index the players by $(i,s) \in [n] \times \calS$.
That is, a copy of each player is created at different states.
\item For every player $(i,s)$, we have $\calA_{(i,s),s}' = \calA_{i,s}$, and $\calA_{(i,s),s'}' = \emptyset$ for all $s' \neq s$. 
\item For every player $(i,s)$, and all $s' \in \calS$, we have $\pp'_{s',\aa} = \pp_{s',\aa}$  and $r'_{(i,s),s',\aa} = r_{i,s',\aa}$, for any $\aa\in\calA_{s'}$.
\end{itemize}
Further, we write a stationary strategy for the $\onestatessg$ instance $\mathcal{G}'$ as $\ppi' = (\ppi'_{(i,s)})_{(i,s) \in [n] \times \calS}$, where $\ppi'_{(i,s)} = \ppi_{i,s}$, and define the value functions as specified in~\Cref{sec:prelim}. 
\end{definition}

In essence, Definition~\ref{def:mspp_ospp_ssg} simply adds player copies to each state with the same reward functions as the corresponding player in the original $\ssg$ instance.
(Note that this is the reason for the $\onestatessg$ moniker: by definition, player $(i,s)$ only takes an action at state $s$.)
Also note that the corresponding strategy $\ppi'$ is identical in representational size to the original strategy $\ppi$; therefore, we will overload notation and write $\ppi' = \ppi$ for the rest of this section.
As a consequence of this property, we can compare equilibrium conditions directly, which is precisely what we do in the following lemma.

\begin{theorem}[Approximate-NE equivalences for $\ssg$s and $\onestatessg$s]\label{lem:redx-MtO}
		   Fix $\epsilon\ge 0, \gamma\in[0,1)$ and a strategy $\ppi$ in the original $\ssg$.
	   Then, the strategy $\ppi$ is an $\epsilon$-approximate mixed NE of the original $\ssg$ if its induced strategy $\ppi'$ is a $((1-\gamma)\epsilon/|\calS|)$-approximate mixed NE in the corresponding $\onestatessg$ of the $\ssg$ (Definition~\ref{def:mspp_ospp_ssg}). 
\end{theorem}

Our proof is built upon the following Lemma~\ref{lem:mdp-policy-improve} on structural properties of single-player MDP. The proof of this lemma follows from a property of \emph{policy improvement} of the coordinate-wise (also called asynchronous) exact policy iteration algorithm~\cite{Bertsekas19}. We provide the detailed proof for completeness (as it is a slight generalization from exact to approximate policy iteration).

\begin{lemma}[Policy improvement for single-player MDP]\label{lem:mdp-policy-improve}
In a single-player MDP $\calM = (\calS,\calA,\pp,\r,\gamma)$, for given $\epsilon>0$, suppose $\ppi$ is not an $\epsilon$-approximate optimal policy under initial uniform distribution $\qq = \frac{1}{|\calS|}\ee_{\calS}$, then there must exist a state $s\in\calS$ and a single policy $\ppi_s'\neq \ppi_s$ that is varied only at $s$ such that $V^{(\ppi_s',\ppi_{-s})}(s)> V^{\ppi}(s)+(1-\gamma)\epsilon$.
\end{lemma}

\begin{proof}
Let $V^*$ denote the optimal value of the single-player MDP.
Given a policy $\ppi$ and the value vectors $V^{\ppi}$, we have
\[
V^{\ppi} = \r^{\ppi}+\gamma \PP^{\ppi} V^{\ppi}.
\]
Suppose $\ppi$ is not an $\epsilon$-approximate optimal policy under initial distribution $\qq = \frac{1}{|\calS|}\ee_{\calS}$, we first claim there must exist some state $s\in\calS$ such that 
\begin{align}\label{eq:claimedpolicyopt}
\max_{a\in\calA_s} \left(r_{s,a}+\gamma \pp_{s,a} V^{\ppi}\right)> V^{\ppi}(s)+(1-\gamma)\epsilon.
\end{align}
We prove the claim by contradiction.
First, we define $\calT,\calT^{\ppi}:\R^\calS\rightarrow\R^\calS$ to be the optimal and on-policy Bellman operators~\cite{Bertsekas19,puterman2014markov}, defined as
\begin{gather*}
\calT V (s) = \max_{a\in\calA_s} \left[r_{s,a}+\gamma \pp_{s,a}V\right] \text{ and } \\
\calT^{\ppi} V (s) =  \left[r^{\ppi}(s)+\gamma \PP^{\ppi}(s,\cdot)V\right].
\end{gather*} 
Then, the contradiction of Equation~\eqref{eq:claimedpolicyopt} gives us $\calT V^{\ppi}\le V^{\ppi}+(1-\gamma)\epsilon$.
On the other hand, by the definition of the optimal Bellman operator we have $\calT V^* = V^*$.
Then, the property of $\gamma$-contractivity of the optimal Bellman operator $\calT$ yields
\begin{align*}
\gamma (V^* - V^{\ppi}) \geq \calT V^* - \calT V^{\ppi} \geq V^* - V^{\ppi} - (1 - \gamma) \epsilon,
\end{align*}
and rearranging terms gives $V^* \le V^{\ppi}+\epsilon$ which is the desired contradiction. Consequently, we conclude that the claim is true and we denote $s\in\calS$ to be the state such that $\max_{a\in\calA_s} \left(r_{s,a}+\gamma \pp_{s,a}V^{\ppi}\right)> V^{\ppi}(s)+(1-\gamma)\epsilon$.

Accordingly, we consider the alternative policy $\ppi'_s  = \ee_{a^*}$ where $a^*\in\arg\max_{a\in\calA_s} \left(r_{s,a}+\gamma \pp_{s,a}V^{\ppi}\right)$.
The definition of the on-policy Bellman operator and Equation~\eqref{eq:claimedpolicyopt} give us $\calT^{(\ppi_s',\ppi_{-s})}V^{\ppi}\ge V^{\ppi}+(1-\gamma)\epsilon\cdot\ee_s\ge V^{\ppi}$.
Then, applying $\calT^{(\ppi'_s,\ppi_{-s})}$ recursively we get
\begin{align*}
V^{(\ppi'_s,\ppi_{-s})} &= \lim_{n\rightarrow\infty}\Par{\calT^{(\ppi_s',\ppi_{-s})}}^nV^{\ppi} \\
&> \calT^{(\ppi_s',\ppi_{-s})} V^{\ppi} \\
&\geq V^{\ppi}+(1-\gamma)\epsilon\cdot\ee_s.
\end{align*}
Above, the first equality follows because the value function $V^{(\ppi'_s,\ppi_{-s})}$ is the unique fixed point of its on-policy Bellman operator and the first strict inequality uses the well-known monotonicity property of the on-policy Bellman operator~\cite{Bertsekas19,puterman2014markov}.
Consequently we get $V^{(\ppi'_s,\ppi_{-s})} > V^{\ppi}+(1-\gamma)\epsilon\cdot\ee_s$, and restricting to state $s$ of this inequality completes the proof. 
\end{proof}

\begin{proof}[Proof of~\Cref{lem:redx-MtO}]

Consider the original instance in $\ssg$, $\calG$, and fix a player $i$ and the other players' policies $\ppi_{-i}$.
This induces a single-player MDP $\calM := \calM(\calG,i,\ppi_{-i})$ with the following specifications:
\begin{itemize}
\item The state space $\calS$ is the same as the state space $\calS$ of the original instance $\calG$.
\item The action space is given by $\calA_s = \calA_{i,s}$ for all $s \in \calS$.
\item For action $a_{i,s} \in \calA_{i,s}$ taken at state $s$, the transition probability is given by
\begin{align*}
\pp_{s,a_{i,s}} = \PP^{(\ee_{a_{i,s}},\ppi_{-i})}(s,\cdot)
\end{align*}
and correspondingly, the probability transition kernel under strategy $(\ppi_i,\ppi_{-i})$ is given by 
\begin{align*}
\PP^{\ppi_i}(s,\cdot) = \PP^{\ppi}(s,\cdot).
\end{align*}
\item For action $a_{i,s}\in \calA_{i,s}$ taken at state $s$, the instantaneous reward is given by
\begin{align*}
r_{s,a_{i,s}} = r^{(\ee_{a_{i,s}},\ppi_{-i})}_i(s).
\end{align*}
\end{itemize}
As a consequence of these definitions, it is clear that for any policy $\ppi_i$ and any state $s \in \calS$, we have $V^{\ppi_i}(s) = V_i^{(\ppi_i,\ppi_{-i})}(s)$.
Therefore, the best response policy $\ppi^*_i$ to $\ppi_{-i}$ must be the optimal policy for the induced MDP $\calM$. 
We also observe the following claim due to definitions of their utilities:
\begin{claim}\label{claim:MSPPSSG-MDP-approx-NE}
	A strategy $\ppi$ is an $\epsilon$-approximate mixed NE of $\ssg$ \emph{iff} for all players $i$, $\ppi_i$ is an $\epsilon$-approximate optimal policy for the induced MDP $\calM(\calG, i, \ppi_{-i})$, under initial uniform distribution $\qq = \frac{1}{|\calS|}\ee_{\calS}$.
\end{claim}

Next, we consider the ``copied'' instance in $\onestatessg$, which we denoted by $\calG'$.
Fix a ``player'' in this game $(i,s)$ (corresponding to a player $i$ and state $s$ in $\calG$), and the other players' policies $\ppi_{-(i,s)}$.
This induces a single-player MDP by $\calM' := \calM(\calG',(i,s),\ppi_{-(i,s)})$ with the following specifications:
\begin{itemize}
\item The state space $\calS$ is the same as the state space $\calS$ of the original instance.
\item The action space is given by $\calA'_s = \calA_{i,s}$, and $\calA'_{s'} = \emptyset$ for all $s' \neq s$.
\item For action $a_{(i,s),s} \in \calA_{i,s}$ taken at state $s$, the transition probability is given by
\begin{align*}
\pp'_{s,a_{(i,s),s}} = \PP^{(\ee_{a_{(i,s),s}},\ppi_{-(i,s)})}(s,\cdot)
\end{align*}
and correspondingly, the probability transition kernel under strategy $(\ppi_{(i,s)},\ppi_{-{(i,s)}})$ is given by 
\begin{align*}
(\PP')^{\ppi_{(i,s)}}(s,\cdot) = \PP^{\ppi}(s,\cdot).
\end{align*}
The other transitions (starting from states $s' \neq s$) are defined trivially depending on the transition kernel of the fixed policy $\ppi_{-(i,s)}$ of other players'.
\item For action $a_{(i,s),s} \in \calA_{i,s}$ taken at state $s$, the instantaneous reward is given by
\begin{align*}
r'_{s,a_{(i,s),s}} = \r^{(\ee_{a_{(i,s),s}},\ppi_{-(i,s)})}(i,s).
\end{align*}
\end{itemize}
As a consequence of these definitions, it is clear that for any policy $\ppi_{(i,s)}$ and any state $s \in \calS$, we have $(V')^{\ppi_{(i,s)}}(s) = V_i^{(\ppi_{(i,s)},\ppi_{-(i,s)})}$.
Therefore, the best response policy $\ppi^*_{(i,s)}$ to $\ppi_{-(i,s)}$ must be the optimal policy for the induced MDP $\calM'$. We also observe the following claim due to definitions of their utilities:
\begin{claim}\label{claim:OSPPSSG-MDP-approx-NE}
	A strategy $\ppi$ is an $\epsilon$-approximate mixed NE of $\onestatessg$ \emph{iff} for all players $(i,s)$, and fixing $\ppi_{-(i,s)}$, $\ppi_{(i,s)}$ is an $\epsilon$-approximate optimal policy for the induced MDP $\calM' = \calM(\calG', (i,s),\ppi_{-(i,s)})$ under initial uniform distribution $\qq = \frac{1}{|\calS|}\ee_{\calS}$.
\end{claim}

Now, consider a policy $\ppi$ that is \emph{not} an $\epsilon$-approximate mixed NE of the original game $\ssg$ of $\calG$, then by~\Cref{claim:MSPPSSG-MDP-approx-NE} we have there exists a player $i$, such that $\ppi_i$ is not an $\epsilon$-approximate optimal policy for the induced MDP $\calM(\calG, i,\ppi_{-i})$. 
Applying~\Cref{lem:mdp-policy-improve} to MDP $\calM(\calG, i,\ppi_{-i})$ with uniform initial distribution $\qq = \frac{1}{|\calS|}\ee_{\calS}$ we thus have there must exist a state $s\in\calS$ and a policy $\ppi_{i,s}'\neq\ppi_{i,s}$ so that $V_i^{(\ppi_{i,s}',\ppi_{-(i,s)})}(s)> V_i^{\ppi}(s)+(1-\gamma)\epsilon$. 

Note this corresponds to the utilities of the induced single-player MDP $\calM'=\calM(\calG',(i,s),\ppi_{-(i,s)})$ under initial distribution $\qq = \frac{1}{|\calS|}\ee_\calS$, corresponding to the ``copied'' $\onestatessg$ instance $\calG'$, and implies $\ppi_{i,s}$ is not an $((1-\gamma)\epsilon/|\calS|)$-approximate policy for $\calM'$ by the necessary condition of NE in terms of Bellman equations (see~\Cref{lem:bellman-nece} in~\Cref{ssec:bellman}).  Now by~\Cref{claim:OSPPSSG-MDP-approx-NE} we have consequently $\ppi$ is not an $((1-\gamma)\epsilon/|\calS|)$-approximate mixed NE for the constructed $\onestatessg$ of $\calG'$, concluding the proof.

\end{proof}

\section{Membership of $\ssg$s in $\ppad$}\label{sec:ppad-membership}

In this section, we show the membership of infinite-horizon $\ssg$ in $\ppad$. 
We begin with a brief description of the $\ppad$ complexity class, which is defined with respect to a long-standing computational problem in circuit/graph theory, the $\lineofaend$ problem (first defined in~\cite{papadimitriou1994complexity}, see also~\cite{daskalakis2021complexity} for a  detailed illustration).
 
\begin{definition}\label{def:endofline}
The $\lineofaend$ problem takes as input two binary circuits, each having $n$ inputs and outputs: $\calC_S: \{0,1\}^n \to \{0,1\}^n$ (for successor) and $\calC_P: \{0,1\}^n \to \{0,1\}^n$ (for predecessor).
It returns as output one of the following:
\begin{enumerate}
\item $\zero$ if either (a) $\calC_P(\calC_S(\zero)) = \calC_S(\calC_P(\zero)) = \zero$ or (b) $\calC_P(\calC_S(\zero)) \neq \zero$ \emph{and} $\calC_S(\calC_P(\zero)) \neq \zero$.
\item A binary string $\pp \in \{0,1\}^n$ such that $\pp \neq \zero$ and that either $\calC_P(\calC_S(\pp)) \neq \pp$ or $\calC_S(\calC_P(\pp)) \neq \pp$.
\end{enumerate}
\end{definition}

It is well-known that a solution of the type $1$ or $2$ always exists for any input to $\lineofaend$.
This puts $\lineofaend$ in the class of total search problems~\cite{megiddo1991total}.
There is a more intuitive graph-theoretic interpretation of the $\lineofaend$ problem that helps the reader see this more clearly.
In particular, let the circuits $\calC_S$ and $\calC_P$ implicitly define a directed graph with the set of vertices given by $\{0,1\}^n$ such that the directed edge $(\pp,\qq) \in \{0,1\}^n \times \{0,1\}^n$ belongs to the graph if and only if $\calC_S(\pp) = \qq$ and $\calC_P(\qq) = \pp$.
As a consequence of this definition, (a) all vertices of this graph have both in-degree and out-degree at most $1$, (b) $\pp$ has out-degree equal to $1$ if and only if $\calC_P(\calC_S(\pp)) = \pp$, and (similarly) (c) $\pp$ has in-degree equal to $1$ if and only if $\calC_S(\calC_P(\pp)) = \pp$.
In this graph-theoretic interpretation, $\lineofaend$ is equivalent to finding one of the following outputs:
\begin{enumerate}
\item The $\zero$ vertex if it has equal in-degree and out-degree on this graph (either $1$ or $0$).
\item A vertex $\pp\neq \zero$ that has \emph{either} in-degree or out-degree equal to $0$.
\end{enumerate}
The parity argument on directed graphs, i.e. that the sum of in-degrees on all vertices is equal to the sum of out-degrees, implies that a solution of either type 1 or 2 always exists. To see this, note that \emph{if} the first condition does not hold, then the in-degree of vertex $\zero$ is not equal to its out-degree and at least one other vertex has this property.

The complexity class $\ppad$ (short form for ``polynomial parity arguments on directed graphs''), introduced by~\cite{papadimitriou1994complexity}, is defined with respect to the $\lineofaend$ problem below.
\begin{definition}[$\ppad$ complexity class]\label{def:ppad}
The complexity class $\ppad$ consists of all search problems that are polynomially-time reducible to the $\lineofaend$ problem.
\end{definition}

In this section, we prove the following theorem and we show $\ppad$-membership of $\ssg$.

\begin{theorem}[$\ppad$-membership of $\ssg$.]\label{thm:ppad-membership}
The problem of computing an $\eps$-approximate NE  in $\ssg$ where $\eps = \Omega(1/\poly(\Atot))$ and $1/(1-\gamma) = \poly(\Atot)$ is in $\ppad$.
\end{theorem}

To prove~\Cref{thm:ppad-membership}, we follow the template of the $\ppad$-membership proof outlined by Theorem 3.1,~\cite{daskalakis2009complexity}, which proves $\ppad$-membership.
Formally, we reduce the problem of $\eps$-approximate NE computation for any $\ssg$ instance to an instance of the $\brouwer$ problem, defined below.

\begin{definition}\label{def:brouwer}
Consider dimension $d \geq 1$ and domain $\calX \subseteq [0,1]^d$.
The $\brouwer$-$\delta$ problem takes as input a $L$-Lipschitz function (with respect to the $\ell_{\infty}$-norm) $f: \calX \to \calX$, referred to as the \emph{Brower function},  that can be exactly evaluated in time polynomial in $d$. The problem asks to compute a point $\pp \in \calX$ such that $\|\pp - f(\pp) \|_{\infty} \leq \delta$.
\end{definition}
The following lemma from \cite{daskalakis2009complexity} shows that $\brouwer$-$\delta$ is in $\ppad$. 
\begin{lemma}[cf. Theorem 3.1,~\cite{daskalakis2009complexity}]\label{lem:brouwerppad}
$\brouwer$-$\delta$ is in $\ppad$ for $\delta = \poly(1/d)$, $L = \poly(d)$.
\end{lemma}

To prove \Cref{thm:ppad-membership} we apply Lemma~\ref{lem:brouwerppad} on the strategy $\ppi$ where $d = \sum_{i \in [n]} \sum_{s \in \calS_i} |\calA_{i,s}| = \Atot$.
Accordingly, our proof of $\ppad$-membership construct a Brouwer function $f(\ppi)$ corresponding to every $\ssg$ instance.
We actually provide two proofs for $\ppad$-membership that use two different Brouwer functions.
The first, $f_\Vsf(\ppi)$ (\eqref{eq:brouwerfunction},~\Cref{ssec:ppad-membership}), directly uses the values; the second, $f_\osBsf(\ppi)$ (\eqref{eq:brouwer-alt},~\Cref{ssec:ppad-membership-alt}), uses Bellman errors.
The key technical steps in the proof are to show two essential properties of \emph{both} of these Brouwer functions (abbreviated to $f(\cdot)$ here and in corresponding sections for brevity):
\begin{enumerate}
\item \textbf{Property 1:} \emph{$f(\cdot)$ is $L$-Lipschitz with respect to the $\ell_{\infty}$-norm, where $L = \poly(\Atot, \frac{1}{1 - \gamma})$}:
This allows the applicability of Lemma~\ref{lem:brouwerppad} and the proof that we can find a fixed point $\|\ppi - f(\ppi)\|_{\infty} \leq \delta$ for $\delta = 1/\poly(\Atot,\frac{1}{1 - \gamma})$ in $\ppad$. The proof of this property for both $f_\Vsf(\cdot)$ and $f_\osBsf(\cdot)$ will use properties of value functions under strategy evaluation for single player in stochastic games.

\item \textbf{Property 2:} \emph{a solution to $\brouwer$-$\delta$, i.e. $\|\ppi - f(\ppi)\|_{\infty} \leq \delta$, is an $\eps$-approximate NE of the original $\ssg$ instance for $\eps = \poly(\delta)$}:
This ensures that we can find an $\eps = 1/\poly(\Atot,\frac{1}{1 - \gamma})$-approximate NE in $\ppad$.
The proof of this property for $f_\Vsf(\cdot)$ crucially uses technical steps from the the pseudo-linearity structure of utility functions $u_i$ provided in~\Cref{ssec:quasi-mono}, and for $f_\osBsf(\cdot)$ will use the Bellman definitions of NE provided in~\Cref{ssec:bellman}. 
Our first membership proof for $f_\Vsf(\cdot)$ in particular handles pseudo-linear (which may be non-convex) utilities, and is a useful generalization of the ideas in the original $\ppad$-membership proofs~\cite{daskalakis2009complexity} that rely on exactly linear utilities.
We believe this generalization will be of independent interest. 
\end{enumerate}

We also show that one of our approaches (See \Cref{ssec:ppad-membership}) generalizes to showing that, under mild conditions, approximating $\eps$-NE for any multi-agent games with pseudo-linear utilities is in $\ppad$ (See \Cref{thm:ppad-membership-general}). This is an interesting generalization of prior work (which asssumed linearity or piecewise linearity on either the utility functions themselves, or sufficient conditions for NE, to establish $\ppad$-membership~\cite{chen2009settling,daskalakis2009complexity,vazirani2011market,chen2017complexity,garg2017settling,filos2021complexity}) to incorporate \emph{pseudo-linear} structure.
We believe this may find further utility, by adding to our toolbox of key structural properties of nonlinear general-sum games that allow for membership in the $\ppad$ complexity class.

\subsection{Brouwer functions in terms of direct utilities}\label{ssec:ppad-membership}

By the reduction provided in~\Cref{ssec:redx-m2o}, it suffices to show that $\onestatessg$ is in $\ppad$. To do so, we will first provide a general result regarding the $\ppad$-membership for a class of particular games satisfying pseudo-linear utility structure.

\begin{theorem}\label{thm:ppad-membership-general}
Let $\calG$ be a game with utility functions $\{u_i(\ppi)\}_{i\in[n]}$. Suppose each player $i\in[n]$ has actions space $\calA_i$ with size bounded by $A$, each utility function is $L$-Lipschitz in $\|\cdot\|_\infty$, bounded by $U$, and is $\rho$-pseudo-linear for some $\rho\ge1$, i.e.\ for any $i\in[n]$, when $u_i(\ppi_i,\ppi_{-i})\neq u_i(\ppi_i',\ppi_{-i}) $, we have
	\begin{equation}\label{eq:pseudo-linear-general}
    \frac{u_i(\theta\ppi'_i+(1-\theta)\ppi_i,\ppi_{-i})-u_i(\ppi)}{u_i(\ppi'_i,\ppi_{-i})-u_i(\ppi)} \in\left[\frac{1}{\rho}\cdot\theta, \rho\cdot \theta\right],~\text{for any}~\ppi_{i}, \ppi'_{i},~\theta\in[0,1].
	\end{equation}
	The problem of computing an $\epsilon$-approximate NE of such games for any $\epsilon>0$, when $A,L,U,\rho, 1/\epsilon$ are all polynomially-bounded by the problem representation size, is in $\ppad$.
\end{theorem}

To prove the theorem let's consider Brouwer function $f(\ppi)$ as follows: for any $i\in[n]$, $a\in\calA_i$, 
\begin{align}\label{eq:brouwerfunction}
y_{i,a}(\ppi) = [f(\ppi)]_{i}(a) = \frac{\pi_{i}(a) + \Upsilon_{i}^{\ppi}(a)}{1 + \sum_{a' \in \calA_{i}} \Upsilon_{i}^{\ppi}(a')}
~\text{ where }~\Upsilon_{i}^{\ppi}(a) \defeq \max(u_i(\ee_{a},\ppi_{-i}) - u_i(\ppi),0)\,.
\end{align}

Note that this is the identical Brouwer function to the one that is typically defined for normal-form games~\cite{daskalakis2009complexity}.
However, the original proofs of Properties 1 and 2 heavily use the linear structure in the utilities.
We leverage the powerful pseudo-linear structure of game $\calG$ to show that this Brouwer function also satisfies properties 1 and 2. 

The following lemma establishes \textbf{Property 1}.
In other words, it shows that for any $L$-Lipschitz utility function class $\{u_i\}_{i\in[n]}$, the Brouwer function $f$ defined as in~\eqref{eq:brouwerfunction} is also Lipschitz.
 
\begin{lemma}[\textbf{Property 1}]\label{lem:lipschitzbrouwer}
Given a game $\calG$ where the number of actions for each player is bounded by $A$, and further all utility functions $\{u_i\}_{i\in[n]}$ are $L$-Lipschitz with respect to $\|\cdot\|_\infty$, i.e. $u_i(\ppi)-u_i(\ppi')\le L\|\ppi-\ppi'\|_\infty$ for any $i\in[n], \ppi, \ppi'$. Then for any $\|\ppi-\ppi'\|_\infty\le \delta$, for $f$ defined as in~\eqref{eq:brouwerfunction},
we have
\begin{align}\label{eq:brouwerlipschitz}
    \|f(\ppi) - f(\ppi')\|_{\infty} &\leq  \left(1+4AL\right)\delta.
\end{align}
Consequently, the Brouwer function $f(\cdot)$ is $\poly(A,L)$-Lipschitz with respect to the $\ell_{\infty}$-norm.
\end{lemma}

\begin{proof}[Proof of~\Cref{lem:lipschitzbrouwer}]
For any strategy $\ppi$, we recall the definition of $\Upsilon_i^{a}(\ppi) := \max(0,u_i(\ee_{a},\ppi_{-i}) - u_i(\ppi))$ as shorthand.
Applying Lemma 3.6,~\cite{daskalakis2009complexity} we have 
\begin{equation}\label{eq:error-membership}
\begin{aligned}
    & |y_{i,a}(\ppi) - y_{i,a}(\ppi')| \leq |\pi_{i}(a) - \pi'_{i}(a)| + |\Upsilon_i^{a}(\ppi) - \Upsilon_i^{a}(\ppi')| + \left|\sum_{a' \in \calA_i} (\Upsilon_i^{a'}(\ppi) - \Upsilon_i^{a'}(\ppi'))\right| \\
    &\hspace{1em}\leq \delta + |u_i(\ee_{a},\ppi_{-i}) - u_i(\ee_{a},\ppi'_{-i})| + \sum_{a' \in \calA_i} \left|u_i(\ee_{a'},\ppi_{-i}) - u_i(\ee_{a'},\ppi'_{-i})\right| +\left(|\calA_i|+1\right)\left|u_i(\ppi)-u_i(\ppi')\right|,
\end{aligned}
\end{equation}
where for the last inequality we used the  assumption on $\ppi$, $\ppi'$ and $|\max\{a - b, 0\} - \max\{c - d, 0\}| \leq |a - c| + |b - d|$.
Finally, we use that all $u_i$s are $L$-Lipschitz together with $\norm{\ppi-\ppi'}_\infty\le \delta$ and $\norm{(\ee_a,\ppi_{-i})-(\ee_a,\ppi'_{-i})}_\infty\le \delta$ for any $a\in\calA_i$ to bound each absolute difference term in~\Cref{eq:error-membership} respectively and obtain the final bound. 
\end{proof}

Next, we prove \textbf{Property 2}.
The following lemma, which uses both the quasi-monotonicity and bounded-slope properties of the pseudo-linearity of utility functions established in~\Cref{ssec:quasi-mono}, shows that an approximate fixed point of the Brouwer function $f(\cdot)$ is also an approximate NE. 
\begin{lemma}[\textbf{Property 2}]\label{lem:brouwerapproxNE}
Consider any game $\calG$ such that each player $i\in[n]$ has action space $\calA_i$ with size bounded by $A$, and all utility functions $\{u_i\}_{i\in[n]}$ bounded by $U$ and are $\rho$-pseudo-linear for some $\rho\ge 1$ such that for any $i\in[n]$, when $u_i(\ppi_i,\ppi_{-i})\neq u_i(\ppi_i',\ppi_{-i}) $, we have
	\begin{equation*}
    \frac{u_i(\theta\ppi'_i+(1-\theta)\ppi_i,\ppi_{-i})-u_i(\ppi)}{u_i(\ppi'_i,\ppi_{-i})-u_i(\ppi)} \in\left[\frac{1}{\rho}\cdot\theta, \rho\cdot \theta\right],~\text{for any}~\ppi_{i}, \ppi'_{i},~\theta\in[0,1].
	\end{equation*}
Suppose for $f$ defined as in~\eqref{eq:brouwerfunction},
\[
\|f(\ppi) - \ppi\|_{\infty} \leq \epsilon'
~\text{ for }~
\epsilon'\le \min\left(\frac{\rho^2U^2}{4A^2(1+AU)},\frac{1}{1+AU}\right) ~.
\]
Then $\ppi$ is an $\epsilon$-Nash equilibrium for $\epsilon = \left(8\rho^2A^2+\rho A U\right) \cdot \sqrt{\epsilon'\left( 1+AU\right)}$.
\end{lemma}

To prove~\Cref{lem:brouwerapproxNE} for our general non-linear utility functions, we will need the following~\Cref{lem:quasilowerbound}, showing that one can still express $u_i(\ppi_i,\ppi_{-i}) = \sum_{j\in[k]}\widetilde{\pi}_i(a^j)u_i(\ee_{a^j},\ppi_{-i})$ with properly-behaved $\widetilde{\ppi}\in\Delta^{\calA_i}$, building on the pseudo-linear structure of utilities as shown in~\Cref{eq:pseudo-linear-general}.

\begin{lemma}[Relating $\tildeppi$ to $\ppi$]\label{lem:quasilowerbound}
Given game $\calG$ under the same assumptions as in~\Cref{lem:brouwerapproxNE}. Fix player $i$, given any policy $\ppi_i$ and some threshold $\ell$ we define  $\pi_\tail \defeq \sum_{j\ge \ell+1}\ppi_i(a^j)\le \alpha$, there exists a choice $\tildeppi_i\in\Delta^{\calA_i}$ such that $u_i(\ppi_i,\ppi_{-i}) = \sum_{j \in [|\calA_i|]} \tildepi_{i}(a^j) u_i(\ee_{a^j},\ppi_{-i})$ and the following two conditions are satisfied: 
\begin{subequations}
\begin{align}
\sum_{j \in [k]} \pi_{i}(a^j) \tildepi_{i}(a^j) &\geq \frac{(1-\alpha)^2}{\rho^2|\calA_i|^2}  \text{ and } \label{eq:quasilowerbound}\\
\sum_{j = \ell+1}^k \tildepi_{i}(a^j) &\leq \rho\pi_\tail\label{eq:pipitilde}.
\end{align}
\end{subequations}
\end{lemma}

Building on this lemma, we first give the formal proof of~\Cref{lem:brouwerapproxNE}. We will provide the proof for~\Cref{lem:quasilowerbound} after finishing the proof of~\Cref{lem:brouwerapproxNE}.

\begin{proof}[Proof of~\Cref{lem:brouwerapproxNE}]
The proof of this lemma is an extension to the proof of Lemma 3.8,~\cite{daskalakis2009complexity}, which is tailored to \emph{linear} structure in the utility function.
To extend this proof to more general pseudo-linear utilities, we critically uses the bounded-slope structure that was identified in \Cref{coro:quasi-mono-ssg}.
We fix a player $i$ and for convenience, let $k = |\calA_i|$ and denote the actions $a^1,\ldots,a^k$ (dropping the index $i$ for this proof).
Without loss of generality, we order the actions such that
\begin{align*}
u_i(\ee_{a^1},\ppi_{-i}) \geq  \ldots u_i(\ee_{a^\ell},\ppi_{-i}) \geq u_i(\ppi) \geq u_i(\ee_{a^{\ell+1}},\ppi_{-i}) \geq \ldots \geq u_i(\ee_{a^k},\ppi_{-i}).
\end{align*}
Note that, according to this ordering, the value under strategy evaluation $u_i(\ppi)$ lies between the evaluations for the pure actions $\ell$ and $\ell+1$.

For any $j \in [k]$, recall we defined $\Upsilon_i^{a^j_i} := \max(0,u_i(\ee_{a^j},\ppi_{-i}) - u_i(\ppi))$. 
We also recall $U$ as the maximal attainable utility.
By this ordering, it suffices to show that $\Upsilon_i^{1}(\ppi) \leq \epsilon$ for a suitable value of $\epsilon$ given the definition of NE and monotonicity of utilities as shown in~\Cref{coro:quasi-mono-ssg}.
We will upper bound the quantity $\sum_{j \in [k]} \Upsilon_i^j = \sum_{j\in[\ell]}\Upsilon_i^j$.
First, just as in the proof of Lemma 3.8,~\cite{daskalakis2009complexity}, observe that $\|f(\ppi) - \ppi\|_{\infty} \leq \epsilon'$ implies 
\begin{align}\label{eq:intermediatebrouwer}
\pi_{i}(a^{j}) \left(\sum_{j' \in [k]} \Upsilon_i^{j'} \right) \leq \Upsilon_i^j + \epsilon'', ~\text{ where }~
\epsilon'' \defeq \epsilon'\left(1 + k U\right).
\end{align}
We define as shorthand $\pi_{\tail} := \sum_{j' = \ell+1}^k \pi_{i}(a^{j'})$, and distinguish two cases.

\paragraph{Case 1: $\pi_{\tail} \geq \frac{\sqrt{\epsilon''}}{\rho U}$}: This case does not require any special structure on the utilities.
Here, we sum \eqref{eq:intermediatebrouwer} over $j = \ell + 1,\ldots,k$.
Noting (by the definition of $\ell$) that $\Upsilon_i^j = 0$ for $j\ge \ell+1$, yields
\begin{align*}
\pi_{\tail} \sum_{j' \in [k]} \Upsilon_i^{j'} \leq (k - \ell) \epsilon'' 
~~\implies~~ 
\sum_{j' \in [k]} \Upsilon_i^{j'} \leq \frac{(k-\ell)\epsilon''}{\pi_{\tail}} \leq k \sqrt{\epsilon''} \rho U\,.
\end{align*}
Since $\Upsilon_i^{j'}$ is non-negative for all $j' \in [k]$, we have $\Upsilon_i^1 \leq k \sqrt{\epsilon''} \rho U$, which completes the proof in this case.

\paragraph{Case 2: $\pi_{\tail} \leq \frac{\sqrt{\epsilon''}}{\rho U}$}: This is the more subtle case where we need to use special properties of the utility functions.
The original proof provided in~\cite{daskalakis2009complexity} multiplies both sides of Equation~\eqref{eq:intermediatebrouwer} by $\pi_{i}(a^j)$ and critically uses the \emph{linearity} of utilities, i.e. that $u_i(\ppi_i,\ppi_{-i}) = \sum_{a\in\calA_i} \pi_{i}(a) u_i(\ee_{a}, \ppi_{-i})$.
While the nonlinearity in utilities precludes this proof technique from working for us, we find an elegant fix.
We will instead consider a specific choice $\tildeppi_i\in\Delta^{\calA_i}$ such that $u_i(\ppi_i,\ppi_{-i}) = \sum_{a\in\calA_i} \tildepi_{i}(a) u_i(\ee_{a},\ppi_{-i})$.
We show in~\Cref{lem:quasilowerbound} that pseudo-linearity implies the existence of such $\tildeppi_i$, which also enjoys some nice structural proximity to the original $\ppi_i$. Given the $\tildeppi_i$ from applying~\Cref{lem:quasilowerbound} with $\alpha = 1/(2k)\ge \pi_\tail$, we multiply both sides of Equation~\eqref{eq:intermediatebrouwer} by $\tildepi_{i}(a^j)$ and sum over $j \in [k]$ to get
\begin{align*}
\left(\sum_{j \in [k]} \pi_{i}(a^j) \tildepi_{i}(a^j)\right) \left(\sum_{j' \in [k]} \Upsilon_i^{j'}\right) \leq \sum_{j \in [k]} \tildepi_{i}(a^j) \Upsilon_i^j + \epsilon''.
\end{align*}
Recalling the definition of $\tildeppi_i$, noting that $\Upsilon_i^j=0$ for $j\ge \ell+1$ and replacing $u_i(\ppi_i,\ppi_{-i})$ gives us
\begin{align*}
& \sum_{j \in [k]} \tildepi_{i}(a^j) (u_i(\ee_{a^j},\ppi_{-i}) - u_i(\ppi_i,\ppi_{-i}))= 0 \\
\implies & \sum_{j \in [k]} \tildepi_{i}(a^j) \Upsilon_i^j + \sum_{j = \ell+1}^k \tildepi_{i}(a^j) (u_i(\ee_{a^j},\ppi_{-i}) - u_i(\ppi_i,\ppi_{-i})) = 0,
\end{align*}
and substituting this equality gives 
\begin{align*}
\left(\sum_{j \in [k]} \pi_{i}(a^j) \tildepi_{i}(a^j)\right) \left(\sum_{j' \in [k]} \Upsilon_i^{j'}\right) \leq -\sum_{j = \ell+1}^k \tildepi_{i}(a^j) (u_i(\ee_{a^j},\ppi_{-i}) - u_i(\ppi_i,\ppi_{-i})) + \epsilon''.
\end{align*}
On the other hand, substituting Equation~\eqref{eq:quasilowerbound} gives us
\begin{align*}
\frac{1}{4\rho^2k^2} \left(\sum_{j' \in [k]} \Upsilon_i^{j'}\right) &\leq \sum_{j = \ell+1}^k \tildepi_{i}(a^j) (u_i(\ppi_i,\ppi_{-i})-u_i(\ee_{a^j},\ppi_{-i})) + \epsilon'' \\
&\leq \rho\pi_{\tail} U + \epsilon''
\leq 2\sqrt{\epsilon''}.
\end{align*}
where the second inequality uses Equation~\eqref{eq:pipitilde} along with the provided upper-bound on $\sum_{j = \ell+1}^k \pi_{i}(a^j)$, and the last inequality uses an upper-bound on $\pi_\tail$.
Putting everything together, we get
\begin{align*}
\Upsilon_i^{a^1} =  \Upsilon_i^{1} \leq 8\rho^2k^2\sqrt{\epsilon''},
\end{align*}
and so we have proved our desired approximation bound.
\end{proof}

In the end of this section, we provide the complete proof of Lemma~\ref{lem:quasilowerbound} used in the above argument.
\begin{proof}[Proof of~\Cref{lem:quasilowerbound}]
This proof builds on the pseudo-linearity~\eqref{eq:pseudo-linear-general}.
We fix the strategy $\ppi = (\ppi_i,\ppi_{-i})$ and specify the choice of $\tildeppi_i$ satisfying the required conditions as below. 
First, we specify the choice of $\tildeppi_i$ on the ``tail'' coordinates $j' = \ell+ 1,\ldots,k$.
Recall that we defined $\pi_\tail := \sum_{j = \ell+1}^k \pi_i(a^j)$.
We note that we can write 
\begin{align*}
\ppi_i & = \pi_{\tail}\cdot\ppi_\tail+\left(1-\pi_\tail\right)\cdot\ppi_{\head};\\
\text{where}~~\ppi_\tail & \defeq\begin{bmatrix}
\mathbf{0}_\ell & \frac{\pi_i(a^{\ell+1})}{\pi_\tail}	 & \frac{\pi_i(a^{\ell+2})}{\pi_\tail} & \cdots & \frac{\pi_i(a^{k})}{\pi_\tail}\end{bmatrix}, \\
\ppi_\head & \defeq\begin{bmatrix}
\frac{\pi_i(a^{1})}{1-\pi_\tail}	 & \frac{\pi_i(a^{2})}{1-\pi_\tail} & \cdots & \frac{\pi_i(a^{\ell})}{1-\pi_\tail} & \mathbf{0}_{k-\ell}
\end{bmatrix},
\end{align*}
where we use $\mathbf{0}_n$ to denote the $n$-dimensional vector of zeros.
Now, we note that by bounded slope in pseudo-linearity, there exists a $\tildepi_{\tail} \in [0,1]$ such that
\begin{align}\label{eq:pitailsplit}
u_i(\ppi_i,\ppi_{-i}) &= (1 - \tildepi_\tail) u_i(\ppi_\head,\ppi_{-i}) + \tildepi_\tail u_i(\ppi_\tail,\ppi_{-i}).
\end{align}
It is then easy to see that
\begin{align*}
\tildepi_\tail &= \frac{u_i(\ppi_i,\ppi_{-i})-u_i(\ppi_\head, \ppi_{-i})}{u_i(\ppi_\tail,\ppi_{-i})-u_i(\ppi_\head,\ppi_{-i})} \text{ and } \\
1 - \tildepi_\tail &= \frac{u_i(\ppi_i,\ppi_{-i})-u_i(\ppi_\tail, \ppi_{-i})}{u_i(\ppi_\head,\ppi_{-i})-u_i(\ppi_\tail,\ppi_{-i})}
\end{align*}
Next, we apply the pseudo-linear property~\eqref{eq:pseudo-linear-general} for the two policies $\ppi_{\tail}$ and $\ppi_\head$, and the choices $\theta = \pi_\tail$ or $1-\pi_\tail$ respectively.
Since the actions are in decreasing order of their utilities, we have $u_i(\ppi_\tail,\ppi_{-i})\neq u_i(\ppi_\head,\ppi_{-i})$.
We can then choose $\theta = 1 - \pi_\tail$ and $\theta = \pi_\tail$ respectively, and apply~\Cref{coro:quasi-mono-ssg} to get
\begin{equation}\label{eq:tail-needed-bounds}
\begin{aligned}
1 - \tildepi_\tail = \frac{u_i(\ppi_i,\ppi_{-i})-u_i(\ppi_\tail, \ppi_{-i})}{u_i(\ppi_\head,\ppi_{-i})-u_i(\ppi_\tail,\ppi_{-i})} & \in \left[ \frac{1}{\rho}\cdot (1 - \pi_\tail),  \rho\cdot(1-\pi_\tail)\right],~\text{here we choose}~\theta = 1-\pi_\tail,\\
\tildepi_\tail = \frac{u_i(\ppi_i,\ppi_{-i})-u_i(\ppi_\head, \ppi_{-i})}{u_i(\ppi_\tail,\ppi_{-i})-u_i(\ppi_\head,\ppi_{-i})} &\in \left[ \frac{1}{\rho}\cdot\pi_\tail, \rho\cdot\pi_\tail\right],~\text{here we choose}~\theta = \pi_\tail.
\end{aligned}
\end{equation}
Now, we will construct our strategy $\tildeppi$ satisfying the requisite conditions.
We start by specifying the tail indices.
For the tail indices, we again use pseudo-linearity and note there exists a probability vector $\tildeppi'_\tail$ such that
\begin{align}\label{eq:tail-split}
u_i(\ppi_\tail,\ppi_{-i}) = \sum_{j' = \ell+1}^k \tildepi'_\tail(a^{j'}) u_i(\ee_{a^{j'}},\ppi_{-i}).
\end{align}
We will show that \emph{any} such choice of tail indices will satisfy~\Cref{eq:pipitilde}.
In particular, combining this with~\Cref{eq:pitailsplit} gives us $\tildepi(a^{j'}) = \tildepi_\tail \tildepi'_\tail(a^{j'})$ for all $j' \geq \ell + 1$.
Consequently, applying bounds in~\eqref{eq:tail-needed-bounds} we obtain
\begin{align*}
\sum_{j' = \ell + 1}^k \tildepi(a^{j'}) = \sum_{j' = \ell + 1}^k \tildepi_\tail \tildepi'_\tail(a^{j'}) 
= \tildepi_\tail \leq \rho\cdot\pi_\tail.
\end{align*}
This is the desired~\Cref{eq:pipitilde}.

Next, we prove~\Cref{eq:quasilowerbound}.
By the condition $\pi_\tail \leq \alpha$, there must exist an index $j \in [\ell]$ such that $\pi_i(a^j) \geq \frac{1-\alpha}{k}$.
We now specify the coordinate of the strategy $\ppi$ on $a^j$, i.e. $\tildepi(a^j)$.
In particular we consider the policies 
\begin{align*}
\ppi_{\head}^{(-j)} &= \begin{bmatrix}\frac{\pi_i(a^{1})}{1-\pi_\tail - \pi_i(a^j)}  & \cdots & \frac{\pi_i(a^{j-1})}{1-\pi_\tail - \pi_i(a^j)} & 0 & \frac{\pi_i(a^{j+1})}{1-\pi_\tail - \pi_i(a^j)} & \cdots & \frac{\pi_i(a^{\ell})}{1-\pi_\tail - \pi_i(a^j)} & \mathbf{0}_{k-\ell} \end{bmatrix}\\
\ppi_{\head}^{(j)} &= \ee_{a^j}
\end{align*}
and note that $\ppi_\head$ can be written as a linear combination of the policies $\ppi_{\head}^{(-j)}$ and $\ppi_{\head}^{(j)}$ in the following manner:
\begin{align*}
\ppi_\head = \ppi_{\head}^{(-j)} + \frac{\pi_i(a^j)}{1-\pi_\tail} (\ppi_{\head}^{(j)} - \ppi_{\head}^{(-j)}).
\end{align*}
Similarly, by assumption~\Cref{eq:pseudo-linear-general}, there exists a $\theta \in \left[\frac{1}{\rho}\cdot\frac{\pi_i(a^j)}{1-\pi_\tail},\rho\cdot \frac{\pi_i(a^j)}{1-\pi_\tail}\right]$ such that
\begin{align}\label{eq:thetaequation}
u_i(\ppi_\head,\ppi_{-i}) &= (1-\theta) u_i(\ppi_{\head}^{(-j)},\ppi_{-i}) + \theta u_i(\ee_{a^j},\ppi_{-i}).
\end{align}
Combining this with~\Cref{eq:pitailsplit}, and~\eqref{eq:tail-split} we have
\begin{align*}
	u_i(\ppi_i) & = (1 - \tildepi_\tail) u_i(\ppi_\head,\ppi_{-i}) + \tildepi_\tail u_i(\ppi_\tail,\ppi_{-i})\\
	& = (1-\tildepi_\tail)(1-\theta)u_i(\ppi_{\head}^{(-j)},\ppi_{-i})+(1-\tildepi_\tail)\theta u_i(\ee_{a^j},\ppi_{-i}) + \sum_{j' = \ell + 1}^k \tildepi_\tail \tildepi'_\tail(a^{j'}) u_i(\ee_{a^{j'}},\ppi_{-i})\\
	& \stackrel{(\star)}{=} \sum_{j' \neq j, j \in [\ell]} (1-\tildepi_\tail)(1-\theta)\tildepi'_\head(a^{j'}) u_i(\ee_{a^{j'}},\ppi_{-i})+(1-\tildepi_\tail)\theta u_i(\ee_{a^j},\ppi_{-i})\\
	& \hspace{3em} + \sum_{j' = \ell + 1}^k \tildepi_\tail \tildepi'_\tail(a^{j'}) u_i(\ee_{a^{j'}},\ppi_{-i})\\
	& = \sum_{j' \neq j, j \in [\ell]} \tildepi(a^{j'}) u_i(\ee_{a^{j'}},\ppi_{-i})+\tildepi(a^{j'})  u_i(\ee_{a^j},\ppi_{-i})+\sum_{j' = \ell + 1}^k \tildepi(a^{j'})u_i(\ee_{a^{j'}},\ppi_{-i}),\\
\end{align*}
 where we recall that $\tildepi(a^{j'}) = \tildepi_\tail \tildepi'_\tail(a^{j'})~\text{for}~ \ell+1\le j'\le k$ as specified above, and further we let $\tildepi(a^{j'}) = (1-\tildepi_\tail)(1-\theta)\tildepi'_\head(a^{j'})$ for $j'\in[\ell], j'\neq j$.
We also note that $\tildepi(a^{j}) = (1-\tildepi_\tail)\theta$.
To verify equality $(\star)$, we note there exists a probability vector $ \tildeppi'_\head$ on $a^{j'}$ for all $j'\in[l], j'\neq j$ such that $u_i(\tildeppi_{\head}^{(-j)},\ppi_{-i}) = \sum_{j' \neq j, j \in [\ell]} \tildepi'_\head(a^{j'}) u_i(\ee_{a^{j'}},\ppi_{-i})$.
These steps follow again due to pseudo-linearity and a similar recursive argument on $j'\in[\ell], j'\neq j$.

Finally, the bounds that we have established on $\theta$ and $(1 - \pi_\tail)$ give us
\begin{align*}
\sum_{j' \in [k]} \pi_{i}(a^{j'}) \tildepi_{i}(a^{j'}) \geq (1 - \tildepi_\tail) \theta \pi_i (a^j)\ge \frac{(1 - \tildepi_\tail)}{(1 - \pi_\tail)} \frac{1}{\rho}\cdot \pi_i (a^j)^2.
\end{align*}
Noting that $1 - \tildepi_\tail \geq \frac{1}{\rho} (1 - \pi_\tail)$ (Equation~\eqref{eq:tail-needed-bounds}), along with the fact that $\pi_i(a^j)\ge \frac{1-\alpha}{k}$, gives us
\begin{align*}
\sum_{j' \in [k]} \pi_{i}(a^{j'}) \tildepi_{i}(a^{j'}) \geq \frac{(1-\alpha)^2}{\rho^2k^2}.
\end{align*}
This completes the proof of Equation~\eqref{eq:quasilowerbound} and the proof of the lemma.
\end{proof}

We now combine these results to prove~\Cref{thm:ppad-membership-general} and with it, providing a way to formally prove~\Cref{thm:ppad-membership}.
\begin{proof}[Proof of~\Cref{thm:ppad-membership-general}]
Combining~both~\Cref{lem:lipschitzbrouwer} and~\Cref{lem:brouwerapproxNE} formally shows the any games satisfying assumptions in~\Cref{thm:ppad-membership-general} will satisfy both \textbf{property 1} and \textbf{property 2} and thus we can apply~\Cref{lem:brouwerppad} to conclude the $\ppad$-membership of the stated games with pseudo-linear utilities.
\end{proof}

In particular, we could instantiate this general membership result for $\onestatessg$. This gives our first proof of the $\ppad$-membership of $\ssg$s.

\begin{proof}[Proof of~\Cref{thm:ppad-membership} using pseudo-linear utilities]
	Without loss of generality, it suffices to show membership of $\onestatessg$ instances with utility functions defined as $u_i(\ppi) = \outility(\ppi) = V_i^{\ppi}(s_i)$ and the Brouwer function defined following~\eqref{eq:brouwerfunction}.
	
	It is immediate to show all utilities are bounded by $|u_i(\ppi)|\le \frac{1}{1-\gamma}$, and the number of actions is bounded by $\Atot$.
	
	Further, we claim the utilities are $(\max_{i\in[n]}|\calA_i|(1-\gamma)^{-2})$-Lipschitz. To see this, by definition we have
\begin{align}
u_i(\ppi) - u_i(\ppi') &= 
\ee_{s_i}^\top \left[ \left(\II - \gamma \PP^{\ppi} \right)^{-1} \r^{\ppi} -  \left(\II - \gamma \PP^{\ppi'}\right)^{-1} \r^{\ppi'}  \right] \nonumber\\
&=
\ee_{s_i}^\top \left[ \left(\II - \gamma \PP^{\ppi} \right)^{-1} -  \left(\II - \gamma \PP^{\ppi'}\right)^{-1}  \right]  \r^{\ppi} 
+
\ee_{s_i}^\top \left(\II - \gamma \PP^{\ppi'}\right)^{-1} (\r^{\ppi} - \r^{\ppi'} ).\label{eq:utility-lip}
\end{align}
We first note that
\begin{align}
\left(\II - \gamma \PP^{\ppi} \right)^{-1} -  \left(\II - \gamma \PP^{\ppi'}\right)^{-1}  
&=
\left(\II - \gamma \PP^{\ppi} \right)^{-1} \left[  \left(\II - \gamma \PP^{\ppi'} \right)- \left(\II - \gamma \PP^{\ppi} \right) \right]   \left(\II - \gamma \PP^{\ppi'}\right)^{-1}  \nonumber\\
&=
\gamma 
\left(\II - \gamma \PP^{\ppi} \right)^{-1} \left[ \PP^{\ppi} -  \PP^{\ppi'} \right]   \left(\II - \gamma \PP^{\ppi'}\right)^{-1}.\label{eq:utility-lip-1}
\end{align}
Since we have $\norm{\ee_{s_i}}_1 = 1$ and $\norm{\r^{\ppi}}_\infty \leq 1$ as well as
\begin{equation}\label{eq:utility-lip-2}
\begin{aligned}
\norm{\r^{\ppi} - \r^{\ppi'}}_\infty & \leq \max_{i\in[n]}\sum_{a\in\calA_{i}}|\pi_{i}(a)-\pi'_{i}(a)|\cdot|r_{i,s_i,a}|\le \max_{i\in[n]}|\calA_{i}|\cdot \delta,\\
\norm{ \PP^{\ppi'} -  \PP^{\ppi}}_\infty & \le \max_{i\in[n],}\sum_{a\in\calA_{i}}|\pi_{i}(a)-\pi'_{i}(a)|\cdot\norm{\pp_{s_i,a}}_1 \leq \max_{i\in[n]}|\calA_{i}|\cdot\delta,\\
~~\text{and}~~\norm{\left(\II - \gamma \PP^{\ppi} \right)^{-1}}_\infty & = \norm{\sum_{t \geq 0} \gamma^t [\PP^{\ppi} ]^t }_\infty \leq \frac{1}{1 - \gamma},
\end{aligned}
\end{equation}
plugging~\eqref{eq:utility-lip-1}~and~\eqref{eq:utility-lip-2} back in~\eqref{eq:utility-lip} we obtain that 
\begin{align*}
	\left| u_i(\ppi) - u_i(\ppi') \right| & 
	\leq \max_{i\in[n]}|\calA_{i}|\cdot\left(\frac{\gamma \delta}{(1 - \gamma)^2} + \frac{\delta}{1 - \gamma}\right)
	= \max_{i\in[n]}|\calA_{i}|\cdot\frac{\delta}{(1 - \gamma)^2}\,.
\end{align*}

Finally, we have the utility functions of $\onestatessg$ $V_i^{\ppi}(s_i)$ are $\rho$-pseudo-linear with $\rho = 1/(1-\gamma)$, following~\Cref{coro:quasi-mono-ssg}.

 Thus, applying~\Cref{thm:ppad-membership-general} with $U =\rho = 1/(1-\gamma)$, $A = \Atot$, $L = A/(1-\gamma)^2$, which are all polynomial in terms of $\Atot$ by assumption, we thus have computing $\eps$-approximate NE in $\onestatessg$ is in $\ppad$, and also computing $\eps$-approximate NE in $\ssg$ is in $\ppad$ due to the reduction in~\Cref{ssec:redx-m2o}.
\end{proof}

\subsection{Brouwer function using one-step Bellman equations}\label{ssec:ppad-membership-alt}

In this section, we provide an alternative argument that $\ssg$ is in $\ppad$. Rather than working with $\onestatessg$s as in \Cref{ssec:ppad-membership} and utilizing the general pseudo-linear utility structure, here we work directly with $\ssg$s. We consider the following Brouwer function that uses Bellman errors:
\begin{align}\label{eq:brouwer-alt}
y_{i,s,a}(\ppi) & = [f(\ppi)]_{i,s}(a) = \frac{\pi_{i,s}(a) + \Upsilon^{\ppi}_{i,s}(a)}{1+\sum_{a' \in \calA_{i,s}} \Upsilon^{\ppi}_{i,s}(a')},\\
\text{where we define}~\Psi_{i,s}^{\ppi}(a) & :=  [r_{i,s,a} + \gamma \pp_{s,a} \VV_i^{\ppi}]-V_i^{\ppi}(s)~~\text{and}~~\Upsilon_{i,s}^{\ppi}(a) := \max(\Psi_{i,s}^{\ppi}(a),0).\nonumber
\end{align}
The intuition for this choice comes from the Bellman-optimalty conditions for NE defined in Equation~\eqref{eq:NE-bellman}.
We now show \textbf{Property 1} and \textbf{Property 2} below.

\begin{lemma}[\textbf{Property 1}]
\label{lem:lipschitzbrouwer-alt}
For any strategies $\ppi,\ppi'$ such that $\|\ppi - \ppi'\|_{\infty} \leq \delta$ we have 
\begin{align*}
	\|f(\ppi) - f(\ppi')\|_{\infty} &\leq  \left(1+\frac{4}{(1-\gamma)^2}\max_{i,s}|\calA_{i,s}|^2\right)\delta.
\end{align*}
Consequently, the Brouwer function $f(\cdot)$ is $\poly(\Atot, \frac{1}{1 - \gamma})$-Lipschitz with respect to the $\ell_{\infty}$-norm.
\end{lemma}

\begin{proof}
Consider $\ppi,\ppi'$ such that $\|\ppi - \ppi'\|_{\infty} \leq \delta$.
Applying Lemma 3.6,~\cite{daskalakis2009complexity} we have for any $i\in[n]$, $s\in\calS_i$, $a\in\calA_{i,s}$, 
\begin{equation}\label{eq:error-membership-alt}
\begin{aligned}
    \left|[f(\ppi)]_{i,s}(a) - [f(\ppi')]_{i,s}(a)\right| & \leq |\pi_{i,s}(a) - \pi'_{i,s}(a)| + \left|\Upsilon_{i,s}^{\ppi}(a) - \Upsilon_{i,s}^{\ppi'}(a)\right| + \left|\sum_{a \in \calA_{i,s}} (\Upsilon_{i,s}^{\ppi}(a) - \Upsilon_{i,s}^{\ppi'}(a))\right| \\
    & \leq \delta + \left|V^{\ppi}_i(s) - V_i^{\ppi'}(s)\right| + \left|\gamma \pp_{s,a}\VV_i^{\ppi} - \gamma \pp_{s,a} \VV_i^{\ppi'}\right| + \\
&\quad\quad\sum_{a' \in \calA_{i,s}}\left( \left|V^{\ppi}_i(s) - V_i^{\ppi'}(s)\right| +\left| \gamma \pp_{s,a'}\VV_i^{\ppi} - \gamma \pp_{s,a'} \VV_i^{\ppi'}\right|\right) \\
&\leq \delta + 2|\calA_{i,s}|(1+\gamma) \norm{\VV_i^{\ppi} - V_i^{\ppi'}(s)}_\infty,
\end{aligned}
\end{equation}
where we again use assumption on $\ppi$, $\ppi'$ and $|\max\{a - b, 0\} - \max\{c - d, 0\}| \leq |a - c| + |b - d|$ for the last two inequalities. 
Finally we apply \Cref{lem:utility-lipschitz} with $\norm{\ppi-\ppi'}_\infty\le \delta$ for any $i\in[n]$, $s\in\calS_i$, $a\in\calA_i$ to  obtain the  final bound.
\end{proof}

\begin{lemma}[\textbf{Property 2}]\label{lem:brouwerapprox-alt}
Given $\epsilon'\le \frac{1}{2\max_{i,s}|\calA_{i,s}|}$, and any joint strategy $\ppi$ for which $\|f(\ppi) - \ppi\|_{\infty} \leq \epsilon'$, then it is also an $\epsilon$-approximate Nash equilibrium, where $\epsilon := \frac{8\max_{i,s}|\calA_{i,s}|^2}{(1 - \gamma)^2}\sqrt{\epsilon'\max_{i,s}|\calA_{i,s}|}$.
\end{lemma}

\begin{proof}
The proof of this lemma follows similarly as the proof of Lemma 3.8,~\cite{daskalakis2009complexity}, using the linearity of the Bellman operator for strategy evaluation.
Fix a player $i$ and state $s \in \calS_i$, and denote $k := |\calA_{i,s}|$ as shorthand.
Then, without loss of generality, we order the actions such that
\begin{align*}
\Psi^{\ppi}_{i,s}(a^1) \geq \Psi^{\ppi}_{i,s}(a^2) \geq \ldots \geq \Psi^{\ppi}_{i,s}(a^\ell) \geq 0 \geq \Psi^{\ppi}_{i,s}(a^{\ell+1}) \geq \ldots \ge \Psi^{\ppi}_{i,s}(a^k).
\end{align*}
for some index $1 \leq \ell \leq k$.

First, just as in the proof of Lemma 3.8,~\cite{daskalakis2009complexity}, observe that the inequality $\|f(\ppi) - \ppi\|_{\infty} \leq \epsilon'$ implies that 
\begin{align}\label{eq:intermediatebrouwer-alt}
\pi_{i}(a^{j}) \left(\sum_{j' \in [k]} \Upsilon_{i,s}^{\ppi}(a^{j'}) \right) &\leq \Upsilon_{i,s}^{\ppi}(a^{j}) + \epsilon'',\\
 \text{ where }\epsilon'' &:= \epsilon'\left(1 + k U_{\max}\right), ~U_{\max} = \frac{3M}{1-\gamma} \nonumber.
\end{align}
We define as shorthand $\pi_{\tail} := \sum_{j' = \ell+1}^k \pi_{i}(a^{j'})$, and distinguish two cases.
\paragraph{Case 1: $\pi_{\tail} > \sqrt{\frac{\epsilon''}{U_{\max}}}$:} This case does not require any special structure on the utilities.
Here, we sum Equation~\eqref{eq:intermediatebrouwer-alt} over $j = \ell + 1,\ldots,k$ and note $\Upsilon_{i,s}^{\ppi}(a^{j}) = 0$ for $j\ge \ell+1$ by assumption to get
\begin{align*}
\pi_{\tail} \sum_{j' \in [k]} \Upsilon_{i,s}^{\ppi}(a^{j'}) &\leq (k - \ell) \epsilon'' \\
\implies \sum_{j' \in [k]} \Upsilon_{i,s}^{\ppi}(a^{j'}) &\leq \frac{(k-\ell)\epsilon''}{\pi_{\tail}} \leq k \sqrt{\epsilon'' U_{\max}}.
\end{align*}

Since $\Upsilon_{i,s}^{\ppi}(a^{j'})$ is non-negative for all $j' \in [k]$, we have $\Upsilon_{i,s}^{\ppi}(a^{1}) \leq k \sqrt{\epsilon''} U_{\max}$, which completes the proof in this case.
\paragraph{Case 2: $\pi_\tail \leq \sqrt{\frac{\epsilon''}{U_{\max}}}$:} Here we use the linearity of the Bellman errors $\Psi_{i,s}^{\ppi}$ for all $j \in [k]$.
In other words, we have $\sum_{j \in [k]} \pi_{i,s}(a^j) \Psi_{i,s}^{\ppi}(a^j) = 0$ by the definition of policy evaluation of Bellman operator.
First, multiplying both sides of Equation~\eqref{eq:intermediatebrouwer-alt} by $\pi_{i,s}(a^j)$ and summing over all $j \in [k]$, we get
\begin{align*}
\left(\sum_{j \in [k]} \pi_{i,s}(a^j)^2\right) \left(\sum_{j' \in [k]} \Upsilon^{\ppi}_{i,s}(a^j)\right) \leq \sum_{j \in [k]} \pi_{i,s}(a^j) \Upsilon^{\ppi}_{i,s}(a^j) + \epsilon''.
\end{align*}
Now, on one hand, we have $\sum_{j \in [k]} \pi_{i,s}(a^j)^2 \geq \frac{1}{k^2}$ since $\ppi_{i,s}\in\Delta^{\calA_{i,s}}$ is a probability distribution.
On the other hand, we have
\begin{align*}
\sum_{j \in [k]} \pi_{i,s}(a_j) \Psi^{\ppi}_{i,s}(a^j) &= 0 \\
\implies \sum_{j \in [k]} \pi_{i,s}(a^j) \Upsilon^{\ppi}_{i,s}(a^j) + \sum_{j = \ell+1}^k \pi_{i,s}(a^j) \Psi^{\ppi}_{i,s}(a^j)  & = 0,
\end{align*}
and putting these together then gives us
\begin{align*}
\frac{1}{k^2} \sum_{j' \in [k]} \Upsilon^{\ppi}_{i,s}(a^j) &\leq \sum_{j = \ell+1}^k \pi_{i,s}(a^j) | \Psi^{\ppi}_{i,s}(a^j)| + \epsilon'' \\
&\leq \pi_{\tail} U_{\max} + \epsilon'' \\
&\leq 2\sqrt{U_{\max} \epsilon''},
\end{align*}
which proves our required approximation bound.
Combining this with the sufficient conditions for Bellman NE(\Cref{ssec:bellman},~\Cref{eq:NE-bellman-suff}) completes the proof.
\end{proof}

This alternative proof of membership makes use of an alternative nonlinear objective (other than the direct values as utility function) in the construction of the fixed-point function $f$~\eqref{eq:brouwer-alt}. 
Certain key steps in the proof preserve linearity in the policy $\ppi_i$ keeping other policies $\ppi_{-i}$ fixed, which yields a comparatively straightforward proof of membership.
However, in the proof argument we still crucially utilize the structure of the $\ssg$s to go from the optimality of $\max_{a^j\in\calA_{i,s}}\Upsilon_{i,s}(a^j)\le \epsilon$ to the approximate Nash equilibrium for player $i$ of the given structure.
In particular we use the sufficient Bellman conditions of approximate NE --- these are discussed in more detail in~\Cref{lem:bellman-suff} of~\Cref{ssec:bellman}.

\section{Complexity results for $\tbsg$s}\label{sec:tbsg}

In this section, we switch our focus to $\tbsg$'s. $\tbsg$'s are a special case of $\ssg$s which remove the obvious indicator of hardness --- which is simultaneity of the game played at each state.
The extra stucture of $\tbsg$'s  allows the design of polynomial-time algorithms for $\tbsg$ in certain simplified settings in contrast to $\ssg$ (for which the corresponding problems remain $\ppad$-hard).
Specifically, in~\Cref{sec:nonstat} we provide a polynomial-time algorithm for computing non-stationary NEs for $\tbsg$s and in~\Cref{ssec:lpalg}, we provide an algorithm which computes approximate NE of $\tbsg$s in polynomial time whenever the number of states is held to a constant. However, ultimately, we show in~\Cref{ssec:ppad-hard} that computing approximate NE of $\tbsg$ is $\ppad$-hard.

\subsection{Non-stationary NE computation for $\tbsg$}\label{sec:nonstat}

The computational complexity of $\ssg$s and $\tbsg$s are fundamentally different when we consider non-stationary strategies, i.e.\ a stationary strategy at every time step when $\gamma=0$.
In particular, $\tbsg$s reduce to multi-agent MDPs with horizon length $H$, where the one-step optimal strategy can be approximated efficiently in polynomial time (assuming polynomially-bounded effective horizon $1/(1-\gamma)$. The latter actually implies we can approximate the non-stationary NE efficiently for $\tbsg$s. Below we formalize this statement for $\tbsg$s to highlight its difference with $\ssg$s.

\begin{definition}\label{def:nonstationary-NE}
	Given a $\tbsg$ instance $\mathcal{G} = (n,\calS=\cup_{i\in[n]}\calS_i,\mathcal{A},\pp,\r,\gamma)$, we consider time-dependent non-stationary policies in form $\ppi=(\ppi_{i}^h)_{i\in[n],h\ge 0}$ where $\ppi_{i,s}^h\in\Delta^{\calA_{i,s}}$. The utility function for each player is $v_i^{\ppi,\qq}$ under the fixed initial distribution $\qq = \frac{1}{|\calS|}\ee_{\calS}$. We say a non-stationary  strategy $\ppi$ is an $\eps$-approximate NE if and only if it holds for all player $i\in[n]$ that 
\begin{align*}
v_i^{\ppi,\qq}\ge v_i^{(\ppi_i',\ppi_{-i}),\qq} -\eps,
~\text{for any}~ [\ppi']_{i,s}^h\in \Delta^{\calA_{i,s}},~s\in\calS_i,~h\ge 1.
\end{align*}

\end{definition}

We provide the model-based backward induction method on computing an approximate NE in~\Cref{alg:nonstationary-BI}.
\begin{algorithm2e}[h]
	\caption{Backward induction for finding approximate non-stationary NE of $\tbsg$}
	\label{alg:nonstationary-BI}
	\DontPrintSemicolon
	\codeInput TBSG instance $\calG  = (n, \calS,\calA,\pp,\r,\gamma)$, accuracy $\epsilon$\;
	Set horizon length $H = \lceil \frac{1}{1-\gamma}\log((1-\gamma)^{-1}\epsilon^{-1})\rceil\in\N_{+}$\;
	Initialize $V_i^{(0)}(s)=0$, for each $i\in[n]$, $s\in\calS$\;
	\For{$h=1$ {\bfseries{\textup{to}}} $H$}
	{
	$Q^{(h)}_i(s, a) \gets r_i(s,a)+\gamma \sum_{s'\in\calS}p_{s,a}(s')V^{(h-1)}_{i}(s')$
	for all $i \in [n]$, $s \in \calS_i$, and $a \in \calA_s$ \label{line:update-Q}\;
	 $a^{h}_s \in 
	\argmax_{a} Q^{(h)}_i(s,a)$
	for all $i \in [n]$ and $s \in \calS_i$ \tcp*{best-response action (break ties arbitrarily)}\label{line:br-a}
	$\ppi_{i}^{h}(s) \gets \ee_{a^{h}_s}$
	for all $i \in [n]$ and $s \in \calS_i$ \tcp*{best-response policy}
	\label{line:br-pi}
	$V^{(h)}_i{(s)} \gets Q_i^{(h)}(s,a^{h}_s)$ for each $i\in[n]$, $s\in\calS$\label{line:update-V}
	}
	\codeReturn $\bar{\ppi}$ such that $\bar{\ppi}^{h} = \ppi^{H-h+1}$, $\forall h\in[H]$, and arbitrary for $h\ge H+1$
\end{algorithm2e}	

We show that~\Cref{alg:nonstationary-BI} computes an $\epsilon$-approximate NE satisfying~\Cref{def:nonstationary-NE} in polynomial time, formally through the following proposition.
We note that variance reduction and sampling techniques could further improve this complexity in terms of the dependence on the various parameters~\cite{sidford2020solving}, but is outside the scope of this paper.

\begin{proposition}[Guarantees of~\Cref{alg:nonstationary-BI}]
	Given a $\tbsg$ instance $\mathcal{G} = (n,\calS=\cup_{i\in[n]}\calS_i,\mathcal{A},\pp,\r,\gamma)$ and some accuracy $\epsilon>0$,~\Cref{alg:nonstationary-BI} finds a non-stationary strategy which is an  $\eps$-approximate NE satisfying~\Cref{def:nonstationary-NE} using a total computational cost of 
\[
O\Par{\frac{|\calS| \Atot}{1-\gamma}\log\Par{\frac{1}{(1-\gamma)\epsilon}}}.
\]
\end{proposition}

\begin{proof}

\textbf{Correctness.} By the turn-based structure of the game, it is not hard to see that by definition the backward induction gives a NE for the $H$-horizon general-sum $\tbsg$, i.e. where the utilities are defined using the same initial distribution, but the cumulative rewards for each player are computed restricted to the first $H$ time steps, i.e.,
\[
V_i^{(H),\ppi}(s) = \E\left[\sum_{h \in [H]}\gamma^h r_i(s^h,a^h)|s_1 = s, a^h\sim \ppi^h_j(s^h)~\text{for some}~j~\text{such that}~s^h\in\calS_j\right].
\]
By backward induction, the resulting strategy $\ppi = (\ppi_{i}^h)_{i\in[n], h\in[H]}$ satisfies for any $i\in[n]$, $s\in\calS$,
\begin{align}\label{eq:nonstationary-finite}
V_i^{(H),\ppi}(s)\ge V_i^{(H),(\ppi_i',\ppi_{-i})}(s), \quad \text{for any}~~  [\ppi']_{i,s}^h\in \Delta^{\calA_{i,s}}, ~s\in\calS_i,~h\in[H].
\end{align}

Now, we consider the general discounted infinite-horizon case, we argue that for given $H$, the output strategy of algorithm that follows strategies $(\ppi^{H},\ppi^{H-1},\cdots,\ppi^{1})$ and plays arbitrary strategy afterwards would give a good approximate NE satisfying~\Cref{def:nonstationary-NE}.

Consider the strategy we play as $\bar{\ppi}$ such that $\bar{\ppi}^{h} = \ppi^{H+1-h}$ for all $h\in[H]$, and arbitrary for $h>H$. For any player $i\in[n]$, we let $V_i^{H,\ppi}\in\R^{\calS}$ to be the value of the game for playing strategy $\ppi$ for the first $H$ finite-horizon starting from initial value $0$ so that\footnote{Note this is a definition different from $V_i^{(H),\ppi}$ as it reverses the use of strategy in its order.} $V_i^{H,\ppi} = V_i^{(H),\bar{\ppi}}$, and $V_i^{> H,\ppi}$ as the value of the state following strategy $\ppi$ starting from time step $h=H+1$. 
Thus, by the recursive formulation of $V_i$ we have 
\begin{align*}
V_i^{\bar{\ppi}}(s) & = V_i^{H,\bar{\ppi}}(s)+\gamma^H \ee_s^\top  \Par{\PP^{\bar{\ppi}^1} \PP^{\bar{\ppi}^2} \ldots \PP^{\bar{\ppi}^H}}V_i^{> H,\bar{\ppi}}\\
& = V_i^{(H),\ppi}(s) + \gamma^H \ee_s^\top  \left(\PP^{\ppi^H} \PP^{\ppi^{H-1}} \ldots \PP^{\ppi^1} \right)V_i^{> H,\bar{\ppi}}.
\end{align*}

Now given the assumption that all rewards $|r_i(s,a)|\le 1$, for any $i\in[n]$, $s\in\calS$ and $a\in\calA_s$, we know by definition of $V_i$ that $\|V_i^{> H,\ppi}\|_\infty\le \frac{1}{1-\gamma}$ for arbitrary strategy $\ppi$. Also note $\norm{\PP^{\ppi}}_\infty\le 1$ for any strategy $\ppi$, consequently one has for any player $i\in[n]$, $s\in\calS$, 
\begin{align*}
V_i^{\bar{\ppi}}(s) & \stackrel{(i)}{\ge} V_i^{H,\bar{\ppi}}(s)-\frac{\gamma^H}{1-\gamma}\stackrel{(ii)}{\ge} V_i^{H,(\ppi_i',\bar{\ppi}_{-i})}(s) -\frac{\gamma^H}{1-\gamma}\\
& \stackrel{(i)}{\ge}  V_i^{(\ppi_i',\bar{\ppi}_{-i})}(s)-\frac{2\gamma^H}{1-\gamma} \stackrel{(iii)}{\ge}  V_i^{(\ppi_i',\bar{\ppi}_{-i})}(s)-\epsilon,~~\text{for any}~~  [\ppi']_{i,s}^h\in \Delta^{\calA_{i,s}}, ~s\in\calS_i,~h\ge 1.,
\end{align*}
where we use $(i)$ the $\ell_\infty$ bound of $V_i^{> H,\bar{\ppi}}$, $(ii)$ the optimality of $\ppi$ for finite-horizon in~\eqref{eq:nonstationary-finite}, and $(iii)$ the choice of $H\ge \frac{\log((1-\gamma)^{-1}\epsilon^{-1})}{1-\gamma}$ so that $\gamma^H/(1-\gamma)\le \epsilon$.

Now by definition we have $u_i(\ppi) = \frac{1}{|\calS|}\sum_{s\in\calS}V_i^{\ppi}(s)$ and consequently taking average over $s\in\calS$ for each $i\in[n]$ we conclude that~\Cref{alg:nonstationary-BI} outputs an $\eps$-approximate NE satisfying~\Cref{def:nonstationary-NE}.

\textbf{Computational cost.} 
For the computational cost, each step of~\ref{line:update-Q},~\ref{line:br-a},~\ref{line:br-pi},~\ref{line:update-V} costs $O(|\calS| \Atot)$, $O(|\calS|)$, $O(|\calS|)$, and $O(n|\calS|)$ respectively. 
Since $n = O(|\calS|)$ and $|\calS| = O(\Atot)$, the total computational cost is bounded by 
\begin{align*}
O\Par{H\cdot\Par{|\calS| \Atot+|\calS| + n|\calS|}} = O\Par{Hn\Atot} 
= O\Par{\frac{|\calS|\Atot}{1-\gamma}\log\Par{\frac{1}{(1-\gamma)\epsilon}}}.
\end{align*}	
\end{proof}

\subsection{Polynomial-time approximate NE for constant state TBSGs}\label{ssec:lpalg}

We now turn our attention to stationary NE computation for $\tbsg$s, and provide an algorithm that is polynomial-time when the number of states $|\calS|$ is held to a constant.
The central idea is to write the NE computation problem as the following joint feasibility problem over values $\Vbold := \{V_i(s)\}_{s \in \calS, i \in [n]}$ and joint strategy $\ppi$:
\begin{equation}\label{eq:NE-bellman-sufficient}
V_i(s)
 	\begin{cases}
 	\ge \max_{a\in\calA_s} \Brack{r_{i,s,a}+\gamma \pp_{s,a}^\top   \Vbold_i} -(1-\gamma)\epsilon,~~\text{for any}~i\in[n],~s\in\calS_i,\\
 	= r^{\ppi}_i(s)+\gamma \PP^{\ppi}(s,\cdot) \Vbold_i,~~\text{for any}~i\in[n],~s\notin\calS_i
 	\end{cases}
\end{equation}
Observe that these sufficient conditions are simpler than the ones for $\ssg$ (see~\eqref{eq:NE-bellman-suff} in~\Cref{ssec:bellman}); specifically, if we fix a collection of value vectors $\Vbold \in \R^{\calS \times [n]}$, then solving Equation~\eqref{eq:NE-bellman-sufficient} in strategy-space reduces to solving a \emph{linear program} (LP) feasibility problem: any feasible solution is an $\epsilon$-approximate NE.
Note that fixing the value vector $\Vbold$ does not yield a corresponding simplification for $\ssg$, as Equation~\eqref{eq:NE-bellman-sufficient} would then reduce only to finding a feasible NE of a multiplayer simultaneous normal-form game.

Some natural consequences follow from this observation.
First, if we were given oracle access to the value vector $\Vbold^*$ of \emph{any} NE, then we could compute a corresponding $\epsilon$-approximate NE in time that is polynomial in $\Atot$ and $\frac{1}{1 - \gamma}$ and logarithmic in $\frac{1}{\epsilon}$.
Consequently, when the number of states $|\calS|$ is held to a constant we can perform a brute-force search over values rather than policies, and search for a feasible strategy for each of the corresponding LPs.
It is important to note that when we do this, several of the candidate LPs will \emph{not} be feasible --- but this can be easily checked and such values can be ruled out. The associated algorithm is described in~\Cref{alg:LP-subroutine}.

Putting these pieces together, we obtain which computes an $\epsilon$-approximate NE for $\tbsg$ in polynomial-time whenever the number of states, $|\calS|$, is constants.
We provide its pseudocode in~\Cref{alg:stationary-bruteforce-LP}.

\begin{algorithm2e}[h]
	\caption{Linear program on policies for fixed values}
	\label{alg:LP-subroutine}
	\DontPrintSemicolon
	\codeInput $\tbsg$ instance $\calG  = (\calS = \cup_{i\in[n]}\mathcal{S}_i,\mathcal{A},\pp,\r,\gamma)$,  candidate values $\Vbold := \{V_i(s)\}_{i \in [n], s \in \calS}$, accuracy $\epsilon''$\;
	Define the $\epsilon$-approximate linear feasibility program over $\ppi$
	\begin{align}\label{eq:LP-fixedvalue}
	V_i(s) &\geq r_{i,s,a} + \gamma \pp_{s,a}^\top   \Vbold_i - \epsilon'' \text{ for all } s \in \calS_i, a \in \calA_s \\
	V_i(s) &\geq \r^{\ppi}_i(s) + \gamma \PP^{\ppi}(s,\cdot) \Vbold_i - \epsilon'' \text{ for all } s \in \calS \nonumber \\
	V_i(s) &\leq \r^{\ppi}_i(s) + \gamma \PP^{\ppi}(s,\cdot) \Vbold_i + \epsilon'' \text{ for all } s \in \calS \nonumber \\
	\pi(s,a) &\geq 0 \text{ for all } s \in \calS, a \in \calA \nonumber \\
	\sum_{a \in \calA_s} \pi(s,a) &= 1 \text{ for all } s \in \calS \nonumber.
	\end{align}\;
	\codeReturn either a) $\bar{\ppi}$ that satisfies Equations~\eqref{eq:LP-fixedvalue} or b) ``Infeasible''.
\end{algorithm2e}

\begin{algorithm2e}[h]
	\caption{Brute-force over states + Linear program over actions to solve NE}
	\label{alg:stationary-bruteforce-LP}
	\DontPrintSemicolon
	\codeInput $\tbsg$ instance $\calG  = (\calS = \cup_{i\in[n]}\mathcal{S}_i,\mathcal{A},\pp,\r,\gamma)$, accuracy $\epsilon$\;
	Define $\mathcal{N} := \{\VV \in \R^{\calS \times [n]}: V_i(s) \in \left\{0, \epsilon', 2\epsilon', \ldots, \lceil \frac{1}{1 - \gamma} \rceil \right\} \text{ for all } i \in [n], s \in \calS\}$ \text{ where } $\epsilon' = \frac{\epsilon(1 - \gamma)^2}{1 + \gamma}$.\;
	\For{$\Vbold \in \mathcal{N}$}
	{
	Run~\Cref{alg:LP-subroutine} for $\Vbold,\epsilon'' = \epsilon'(1 + \gamma)$\;
	\lIf{Output of~\Cref{alg:LP-subroutine} is a feasible $\bar{\ppi}$}{
	\codeReturn $\bar{\ppi}$
	}
	}
\end{algorithm2e}

\begin{proposition}\label{prop:TBSG_LP}
Consider a $\tbsg$ instance $\calG = (n, \calS = \cup_{i \in [n]} \calS_i, \calA, \pp, \r, \gamma)$.
Then we have the following results:
\begin{enumerate}
\item Algorithm~\ref{alg:stationary-bruteforce-LP} returns an $\epsilon$-approximate NE in time \[O\left(\exp\left(n |\calS| \log \left(\frac{2(1 + \gamma)}{\epsilon(1 - \gamma)^3}\right) \right)\cdot \poly(\Atot, (1-\gamma)^{-1}, \epsilon^{-1})\right).\]
\item Moreover, Algorithm~\ref{alg:LP-subroutine} initialized with oracle access to an exact equilibrium value $\Vbold^*$ returns an $\epsilon$-approximate NE in time that is polynomial in $\Atot,\frac{1}{1 - \gamma}$ and logarithmic in $\frac{1}{\eps}$.
\end{enumerate}
\end{proposition}
Proposition~\ref{prop:TBSG_LP} provides an algorithm for computing an $\eps$-approximate NE of a general-sum $\tbsg$ that is a) polynomial in the effective action size $\Atot$ when the number of states (and, therefore, number of players) is held to a constant, b) exponential in the number of players $n$ and the number of states $|\calS|$, and c) pseudo-polynomial in the error tolerance $\epsilon$ and effective horizon $\frac{1}{1 - \gamma}$. 
Proposition~\ref{prop:TBSG_LP} shows that oracle access to the value of the game significantly simplifies the problem of equilibrium computation over policies.
(Note that this cannot be the case for a normal-form bimatrix game: if it were, an algorithm in the spirit of~\Cref{alg:stationary-bruteforce-LP} could be run and the brute-force search over all candidate equilibrium values can be done in polynomial time when the number of players is a constant --- which would lead to a contradiction with $\ppad$-hardness.)
It also demonstrates that to derive a hardness result for $\tbsg$, we need to increase either the number of states $|\calS|$ or the number of players $n$.

We conclude this subsection with the proof of Proposition~\ref{prop:TBSG_LP}.
\begin{proof}

\textbf{Correctness.} To show correctness, it suffices to show two lemmas listed below. 

\begin{lemma}\label{lem:feasibleexists}
There exists a feasible $\bar{\ppi}$ to the $\epsilon''$-approximate LP defined in Equation~\eqref{eq:LP-fixedvalue} for $\Vbold := \Vbold^*$, where $\Vbold^*$ is an exact equilibrium value.
For some $\Vbold \in \Nor$, there exists a feasible $\bar{\ppi}$ to the $\epsilon''$-approximate LP defined in Equation~\eqref{eq:LP-fixedvalue}.
\end{lemma}

\begin{lemma}\label{lem:feasibleimpliesopt}
For any candidate value vector $\Vbold$, an $\epsilon''$-feasible solution $\bar{\ppi}$ of~\eqref{eq:LP-fixedvalue}, when it exists, is an $\epsilon$-approximate NE.
\end{lemma}

Lemmas~\ref{lem:feasibleexists} and~\ref{lem:feasibleimpliesopt} clearly directly imply the proof of Proposition~\ref{prop:TBSG_LP}.
To see this, we note that:
\begin{itemize}
\item Lemma~\ref{lem:feasibleexists} implies that Algorithm~\ref{alg:LP-subroutine} run with the equilibrium value $\Vbold^*$ will find an $\epsilon''$-approximate feasible solution, and Lemma~\ref{lem:feasibleimpliesopt} in turn implies that this is an $\epsilon$-approximate NE.
\item Lemma~\ref{lem:feasibleexists} ensures that Algorithm~\ref{alg:stationary-bruteforce-LP} will terminate at some $\widetilde{\Vbold} \in \Nor$, and Lemma~\ref{lem:feasibleimpliesopt} in turn implies that the corresponding returned strategy $\bar{\ppi}$ will be an $\epsilon$-approximate NE.
\end{itemize}
To complete our proof of correctness, we prove these lemmas below.

\begin{proof}(of~\Cref{lem:feasibleexists}.)
Consider any exact NE $\ppi^*$ and corresponding value $\Vbold^*$.
First, recall the Bellman definition of sufficient condition for $\epsilon$-approximate NE given in Equation~\eqref{eq:NE-bellman-sufficient}.
Since $\ppi^*$ is an exact NE, it will satisfy
\begin{align*}
V^*_i(s) &\geq r_{i,s,a} + \gamma \pp_{s,a}^\top   \Vbold^*_i \text{ for all } a \in \calA_s, s \in \calS_i \text{ and } \\
V^*_i(s) &= \r^{\ppi}_i(s) + \gamma \PP^{\ppi}(s,\cdot) \Vbold^*_i \text{ for all } s \in \calS,
\end{align*}
which clearly satisfies the feasibility conditions in Equation~\eqref{eq:LP-fixedvalue}.
Thus, $\ppi^*$ is a feasible solution for the LP in Equation~\eqref{eq:LP-fixedvalue} initialized with $\Vbold := \Vbold^*$.
We will now show that there exists an index $\widetilde{\Vbold} \in \Nor$ such that $\ppi^*$ is also a feasible solution for the LP in Equation~\eqref{eq:LP-fixedvalue} intitialized with $\Vbold$.
We consider the index 
\begin{align*}
\widetilde{\Vbold} := {\arg \min}_{\Vbold \in \Nor} \max_{i \in [n], s \in \calS} |V_i(s) - V^*_i(s)|.
\end{align*}
Since $\Nor$ is an $\epsilon'$-net in the infinity norm, we have $\|\widetilde{\Vbold}_i - \Vbold^*_i\|_{\infty} \leq \epsilon'$ for all $i \in [n]$.
Now, we need to check the feasibility of $\ppi^*$ in the LP defined for $\widetilde{\Vbold}$.
The simplex constraints are trivially satisfied.
To show the best-response constraint, we have for any $i \in [n]$, $a \in \calA_s$ and $s \in \calS_i$,
\begin{align*}
\widetilde{V}_i(s) &\geq V^*_i(s) - \epsilon' \\
&\geq r_i(s,a) + \gamma \pp_{s,a}^\top  \Vbold^*_i - \epsilon' \\
&\geq r_i(s,a) + \gamma \pp_{s,a}^\top  \widetilde{\Vbold}_i - (\gamma + 1) \epsilon'
\end{align*}
where the last inequality uses Holder's inequality.
To show the strategy evaluation constraint, we similarly have for any $i \in [n]$ and $s \in \calS$,
\begin{align*}
&\widetilde{V}_i(s) - \r^{\ppi^*}_i(s) - \gamma \PP^{\ppi^*}(s,\cdot) \widetilde{\Vbold}_i \\
&\leq V_i^*(s) + \epsilon' - \r^{\ppi^*}_i(s) - \gamma \PP^{\ppi^*}(s,\cdot) \Vbold^*_i + \gamma \epsilon' \\
&\leq (\gamma + 1) \epsilon'
\end{align*}
where the last step follows via the strategy evaluation equality for the equilibrium value $\Vbold^*$ and corresponding strategy $\ppi^*$.
An identical argument works for the corresponding lower bound on the strategy evaluation error.

In sum, we have shown that $\ppi^*$ satisfies the feasibility conditions of Equation~\eqref{eq:LP-fixedvalue} for $\widetilde{\Vbold}$ and this completes the proof of the lemma.
\end{proof}

\begin{proof}(of~\Cref{lem:feasibleimpliesopt}.)
Consider a value $\Vbold$ for which we can find an $\epsilon''$-approximate feasible solution $\bar{\ppi}$ to the LP initialized with $\Vbold$.
We show that $\bar{\ppi}$ is an $\epsilon$-approximate NE.
To do this, it suffices to show that this strategy $\bar{\ppi}$ and the actual values of the strategy $\{V^{\bar{\ppi}}_i(s)\}_{i \in [n], s \in \calS}$ satisfy the sufficient condition~\eqref{eq:NE-bellman-sufficient} with $\epsilon := \frac{2\epsilon''}{(1 - \gamma)^2}$. This is because if this sufficient condition holds, the choice of $\epsilon'' = \epsilon'(1 + \gamma)$ immediately completes the proof of the lemma.

To show that condition~\eqref{eq:NE-bellman-sufficient} holds, it suffices to prove only the inequality condition
\[
V^{\bar{\ppi}}_i(s)\ge \max_{a\in\calA_s} \Brack{r_{i,s,a}+\gamma\pp_{s,a}^\top  \Vbold^{\bar{\ppi}}_i} -(1-\gamma)\epsilon,~~\text{for any}~i\in[n],~s\in\calS_i,
\] 
since the equality holds immediately by definition.

Our first step is to show that the actual values are close to the candidate values.
Since $\bar{\ppi}$ is feasible in the LP~\eqref{eq:LP-fixedvalue} defined by $\Vbold$, we have for any $s\in\calS$,
\begin{align*}
V_i(s) - \r^{\bar{\ppi}}_i(s) - \gamma \PP^{\bar{\ppi}}(s,\cdot) \Vbold_i \leq \epsilon''
\end{align*}
On the other hand, the strategy evaluation Bellman operator gives us
\begin{align*}
V^{\bar{\ppi}}_i(s) - \r^{\bar{\ppi}}_i(s) - \gamma \PP^{\bar{\ppi}}(s,\cdot) \Vbold^{\bar{\ppi}}_i =0.
\end{align*}
Subtracting the second equation from the first and taking an maximum over $s\in\calS\setminus\calS_i$ gives us
\begin{align*}
\max_{s\in\calS} \left[V_i(s)- V^{\bar{\ppi}}_i(s)\right]\le \gamma \max_{s\in\calS}  \left[V_i(s)- V^{\bar{\ppi}}_i(s)\right]+\epsilon''.
\end{align*}
Similarly we can also obtain 
\begin{align*}
\min_{s\in\calS} \left[V_i(s)- V^{\bar{\ppi}}_i(s)\right]\ge \gamma \min_{s\in\calS}  \left[V_i(s)- V^{\bar{\ppi}}_i(s)\right]-\epsilon''.
\end{align*}
Combining the two inequalities and by rearranging terms this implies that $\|\Vbold_i - \Vbold^{\bar{\ppi}}_i\|_{\infty} \leq \frac{\epsilon''}{1 - \gamma} \text{ for all } i \in [n]$.

It only remains to show the best-response property for strategy $\bar{\ppi}$ on its corresponding value vector $\Vbold^{\bar{\ppi}}$.
This follows via similar arguments to the proof of Lemma~\ref{lem:feasibleexists} together with the above shown closeness of the strategy evaluation $\Vbold^{\bar{\ppi}}$ and candidate value $\Vbold$.
For every $i \in [n],s \in \calS_i$ and $a \in \calA_s$, we have
\begin{align*}
V^{\bar{\ppi}}_i(s) &\geq V_i(s) - \frac{\epsilon''}{1 - \gamma} \\
&\geq r_i(s,a) + \gamma \pp_{s,a}^\top  \Vbold_i - \frac{\epsilon''}{1 - \gamma} - \epsilon'' \\
&\geq r_i(s,a) + \gamma \pp_{s,a}^\top  \Vbold^{\bar{\ppi}}_i - \frac{\epsilon''(1 + \gamma)}{(1 - \gamma)} - \epsilon'' \\
&= r_i(s,a) + \gamma \pp_{s,a}^\top  \Vbold^{\bar{\ppi}}_i - \frac{2 \epsilon''}{1 - \gamma}.
\end{align*}
Recalling the sufficient Bellman equations for NE (\Cref{lem:bellman-suff} in~\Cref{ssec:bellman}) yields an $\epsilon := \frac{2\epsilon''}{(1 - \gamma)^2}$-NE, and noting that $\epsilon'' = \epsilon'(1 + \gamma)$ completes the proof of the lemma.
\end{proof}

\textbf{Computational cost.} The computational cost is given by the size of the $\epsilon'$-net $\Nor$ multiplied by the cost of running~\Cref{alg:LP-subroutine} for a fixed $\Vbold$.
The LP feasibility problem defined in~\eqref{eq:LP-fixedvalue} can be verified to have $O(\Atot)$ variables and $O(n\Atot)$ constraints.
Consequently, standard interior point methods give complexity that is polynomial in $n,\Atot$ and $\frac{1}{\epsilon}$ for running~\Cref{alg:LP-subroutine}.
Therefore, the total time complexity is given by $O\left(|\Nor| \cdot \text{poly}\left(\Atot,\frac{1}{1 - \gamma},\frac{1}{\epsilon}\right)\right)$.
Finally, we characterize $|\Nor|$.
Since $\Nor$ is an infinity-norm net, we can set all the values independently between $\left[0, \lceil \frac{1}{1 - \gamma} \rceil\right]$; to $\epsilon'$-precision.
Therefore, we have
\begin{align*}
|\Nor| \leq \left(\lceil \frac{1}{(1 - \gamma) \epsilon'} \rceil \right)^{n |\calS|} = \exp(n |\calS| \log \left(\lceil \frac{1}{\epsilon'(1 - \gamma)} \rceil \right))
\end{align*}
and putting these together yields total time complexity \[O\left(\exp\left(n |\calS| \log \left(\frac{2(1 + \gamma)}{\epsilon(1 - \gamma)^3}\right)\right) \cdot \text{poly}\left(\Atot,\frac{1}{1 - \gamma},\frac{1}{\epsilon}\right)\right).\]
This completes the proof.
\end{proof}

\subsection{Computing mixed NE is $\ppad$-hard for $\onestate$}\label{ssec:ppad-hard}

Here we show that computing approximate NE in $\onestate$s (and therefore in $\tbsg$s) is $\ppad$-hard, by providing a formal reduction from the generalized circuit problem ($\epsilon$-$\gcircuit$), which is known to be $\ppad$-hard for $\eps = \poly(1/n)$~\cite{chen2006settling} and for (sufficiently small) constant $\eps$~\cite{rubinstein2018inapproximability}.
We reproduce the definition of the $\epsilon$-$\gcircuit$ problem below from~\cite{chen2006settling,daskalakis2009complexity,rubinstein2018inapproximability} for completeness.

\begin{definition}[Generalized circuits, reproduced from~\cite{chen2006settling,rubinstein2018inapproximability}]\label{def:gcircuit} A generalized circuit $\calC$ is a pair $(V,\calT)$ where $V$ is a set of nodes (or vertices) and $\calT$ is a collection of gates.
Every gate $T \in \calT$ is a tuple of four possible types (all described in \Cref{tab:gcircuit}):
\begin{itemize}
\item $T = G(s\subino,s\subint|s\subout)$ (the addition, subtraction, comparison).
\item $T = G(s\subin|s\subout)$ (the equal and NOT gates).
\item $T = G(\alpha | s\subin | s\subout)$ (the  multiplication gate). 
\item $T = G(\alpha | s\subout)$ (the  constant gate). 
\end{itemize}
Above, $s\subino,s\subint,s\subin$ represent input nodes, $s\subout$ represents an output node and $\alpha$ represents a parameter.
The collection of gates $\calT$ must satisfy the following important property: For every two gates $T,T'$, the output nodes must be distinct, i.e. $s\subout \neq s\subout'$.
\end{definition}

For completeness,~\Cref{tab:gcircuit} lists the set of gates required to implement an arbitrary instance of $\gcircuit$.

\begin{table}[t]
\begin{center}
\begin{small}
\begin{tabular}{|l|c|c|}
\hline
Gate name &  Input & Output constraint  \\
\hline\hline
$G_{=}$ (Equal gate) & $p\subin$ & $p\subout\in[p\subin- \epsilon, p\subin+\epsilon]$  \\
\hline
$G_{\alpha}$ (Constant gate) & $\alpha$ & $p\subout\in[\med(0,\alpha-\epsilon, 1-\epsilon), \med(1,\alpha +\epsilon, \epsilon)]$   \\
\hline
$G_{\times}$ (Multiplicative gate) & $p\subin,\alpha \in(0,2]$ & $p\subout\in[\med(0,\alpha (p\subin-\epsilon),1-\epsilon), \med(1,\alpha (p\subin+\epsilon),\epsilon)]$ \\
\hline
$G_{+}$ (Sum gate) & $p\subino,p\subint$ & $p\subout\in[\med(0,p\subino+p\subint-\epsilon,1-\epsilon), \min(1, p\subino+p\subint+\epsilon)]$ \\
\hline
$G_{-}$ (Subtraction gate) & $p\subino,p\subint$ & $p\subout\in[\max(0,p\subino-p\subint-\epsilon), \med(1, p\subino-p\subint+\epsilon,\epsilon)]$ \\
\hline
$G_{>}$ (Comparison gate) & $p\subino,p\subint$ & $p\subout\ge 1-\epsilon$ if $p\subino\ge p\subint+\epsilon$, and $p\subout\le \epsilon$ if $p\subino\le p\subint-\epsilon$ \\
\hline
$G_{\land}$ (AND gate) & $p\subino,p\subint$ & $p\subout \ge 1 - \epsilon$ if $p\subino,p\subint \ge 1 - \epsilon$, $p\subout \le \epsilon$ if $p\subino$ or $p\subint \le \epsilon$. \\
\hline
$G_{\lor}$ (OR gate) & $p\subino,p\subint$ & $p\subout \ge 1 - \epsilon$ if $p\subino \ge 1 - \epsilon$ or $p\subint \ge 1 - \epsilon$, $p\subout \le \epsilon$ if $p\subino,p\subint \le  \epsilon$.\\
\hline
$G_{\neg}$ (NOT gate) & $p\subin$ & $p\subout \le \epsilon$ if $p\subin \ge 1 - \epsilon$, $p\subout \ge 1 - \epsilon$ if $p\subin \le \epsilon$\\
\hline
\end{tabular}
\end{small}
\caption{ 
Table of gates required to implement an arbitrary instance of $\epsilon$-$\gcircuit$.
All values of $p\subin,p\subino,p\subint,p\subout$ are probabilities and so in the interval $[0,1]$.
Note that the logic gates can be implemented through the composition of subtraction, sum, multiplication and comparison gates. For instance,  we can implement $G_{\neg}$ with input $p\subin$ by $G_{-}$ with input $1$ and $p\subin$. For any $\eps\le 1/12$, we can implement $G_{\land}$ with the comparison gadget $G_{>}$ with the first input being $0.75$, and the second input being the output of the sum gate applied to the output of $p\subino$ and $p\subint$ with a multiplicative gate of $\alpha = 1/2$, respectively. We can also implement $G_{\lor}$ with the sum gadget applied to inputs $p\subino$ and $p\subint$. 
\label{tab:gcircuit}}
\end{center}
\vskip -0.1in
\end{table}

With these definitions, the $\epsilon$-$\gcircuit$ problem is defined as below, also reproduced from~\cite{rubinstein2018inapproximability}.
\begin{definition}[$\epsilon$-$\gcircuit$ problem, reproduced from~\cite{rubinstein2018inapproximability}]
Given a generalized circuit $\calC = (V,\calT)$ an assignment $\pp: V \to [0,1]$ \emph{$\epsilon$-approximately satisfies $\calC$} if for each of the gates $T \in \calT$, the assignment $\pp$ satisfies the constraints listed in~\Cref{tab:gcircuit}.
The \emph{$\epsilon$-$\gcircuit$ problem} is that of finding an $\epsilon$-approximately satisfying assignment $\pp$.
\end{definition}

In fact, $\epsilon$-$\gcircuit$ is known to be $\ppad$-hard not only for an inverse-polynomial error tolerance, i.e. $\epsilon = \poly(1/|V|)$~\cite{chen2006settling}, but even for a constant value of $\epsilon_0 > 0$ that is sufficiently small~\cite{rubinstein2018inapproximability}.

To translate these results to $\ppad$-hardness of $\onestate$, we need to reduce a solution of $\epsilon$-$\gcircuit$ to a $\delta$-approximate NE of the constructed $\tbsg$ game, where $\delta = O(\eps)$.
We are successful in doing this and obtain the following result.

\begin{theorem}[$\ppad$-hardness of $\onestate$]\label{thm:PPAD-hard}
It is $\ppad$-hard to find an $\epsilon_1$-approximate NE in the class of $\onestate$s with discount factor $\gamma=1/2$, for some sufficiently small constant $\epsilon_1$.
\end{theorem}

Theorem~\ref{thm:PPAD-hard} implies $\ppad$-hardness for any error tolerance smaller than some sufficiently small constant $\epsilon_1$, given a discount factor $\gamma=1/2$. This also implies the $\ppad$-hardness for any accuracy $\epsilon\le \epsilon_1$, and for any discount factor $\gamma\in[1/2,1)$, as the latter two cases are weaker notions of hardness for the class of problems we study. We state this as an immediate corollary below.
\begin{corollary}[$\ppad$-hardness of $\onestate$]
	\label{coro:PPAD-hard}
There exists some constant $\epsilon_1\in(0,1)$, such that for any discount factor $\gamma\in(1/2,1)$ and any accuracy $\epsilon\le \epsilon_1$, it is $\ppad$-hard to find an $\epsilon$-approximate NE in the class of $\gamma$-discounted $\onestate$s.
\end{corollary}
\begin{proof}
This is an immediate consequence of~\Cref{lem:hardness-ineq} in~\Cref{apdx:average}, which shows we can reduce NE computation in  $\gamma$-discounted $\onestate$s to NE computation in $\gamma'$-discounted $\onestate$ with same accuracy, for any $\gamma<\gamma'<1$.
\end{proof}

The particular class of $\tbsg$s that we reduce $\epsilon$-$\gcircuit$ to admits the structural property that the discount factor is some fixed constant $\gamma=1/2$ and each player controls only one state ($\onestate$) with two actions. Throughout we use $p_i$ to parametrize the strategy of player $i$ at state $s_i$: with probability $p_i$ it plays $a^1_i$, and with probability $1-p_i$ it plays $a^2_i$.
We first give the construction of a series of gadgets for $\onestate$, which shows that the satisfying the corresponding gate required in $\gcircuit$ can be reduced to finding a mixed NE of the gadget in $\onestate$.
In this reduction, each state (equivalently, player) of $\onestate$ represents a different node/vertex in the $\eps$-$\gcircuit$ instance.
The transition structure that is used in each of the $\onestate$ gadgets is displayed in Figure~\ref{fig:mixedgamegadgets}.

\begin{figure}[t]
\centering
\begin{minipage}[b]{0.49\linewidth}
  \centering
  \centerline{\includegraphics[width=8cm]{Figures/equal_gadget.pdf}}
  \centerline{The equal and multiplication gadgets.}\medskip
\end{minipage}
\begin{minipage}[b]{0.49\linewidth}
  \centering
  \centerline{\includegraphics[width=8cm]{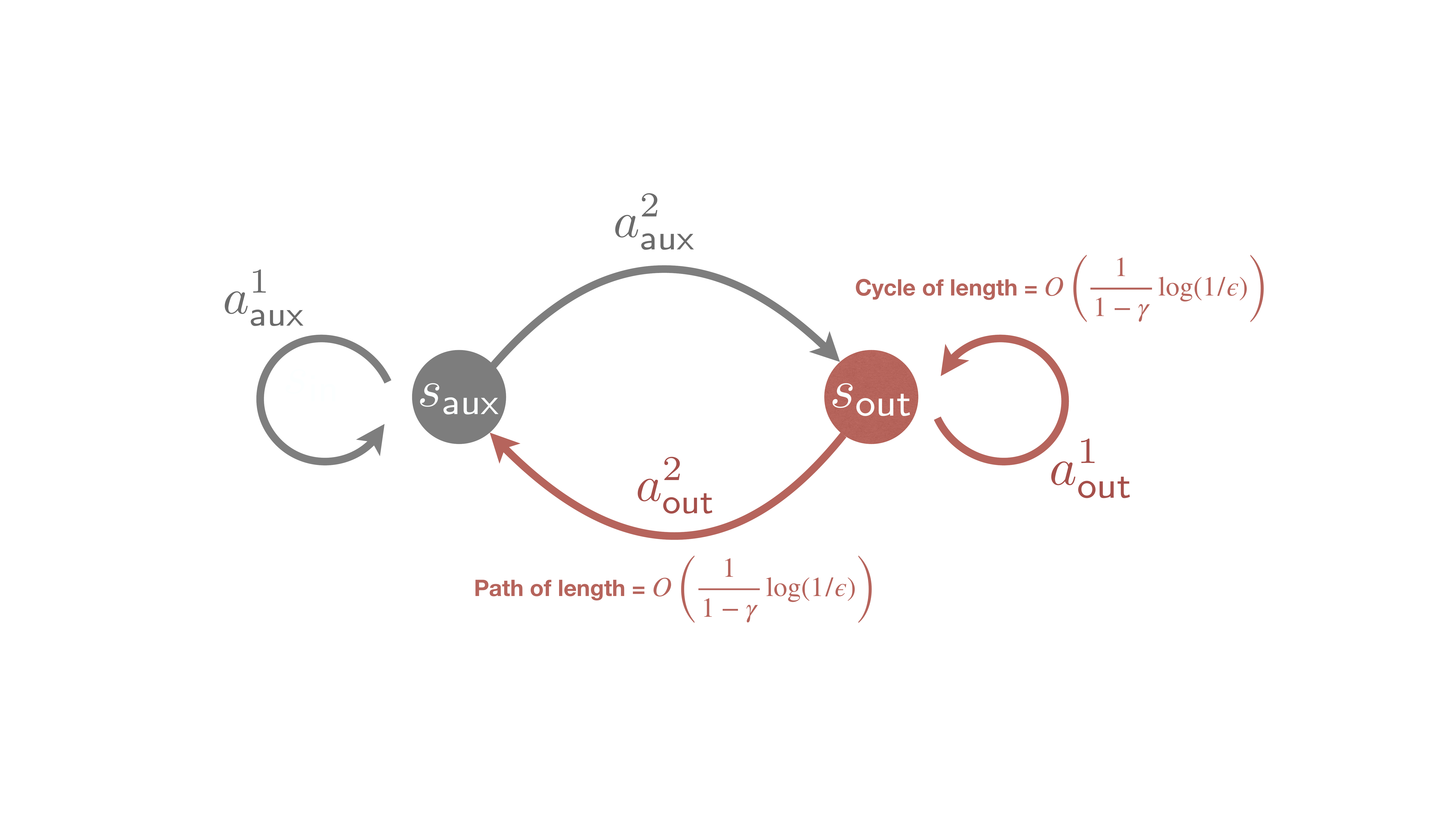}}
  \centerline{The constant gadget.}\medskip
\end{minipage}
\begin{minipage}[b]{0.49\linewidth}
  \centering
  \centerline{\includegraphics[width=8cm]{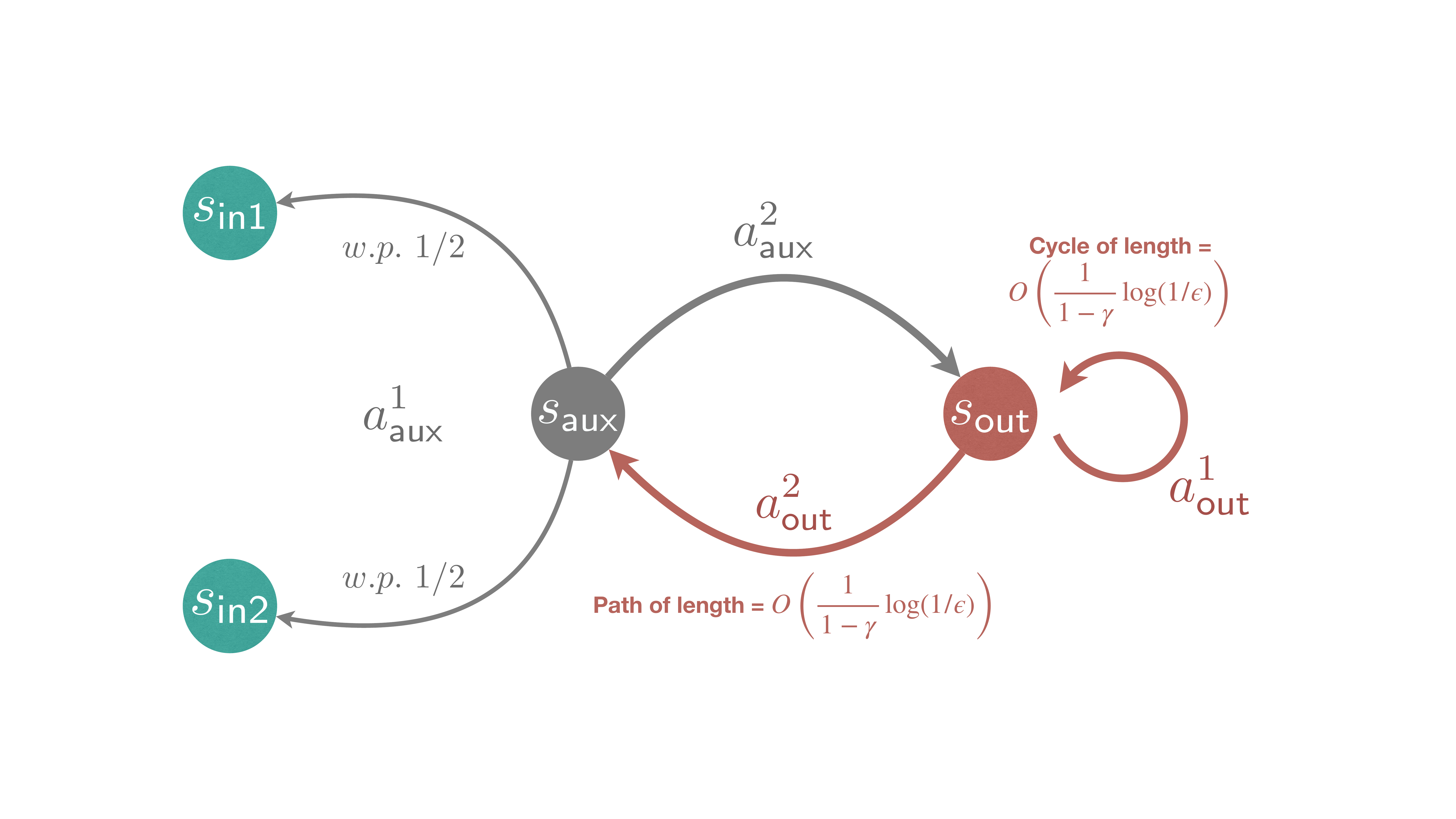}}
  \centerline{The sum and subtraction gadgets.}\medskip
\end{minipage}
\begin{minipage}[b]{0.49\linewidth}
  \centering
  \centerline{\includegraphics[width=8cm]{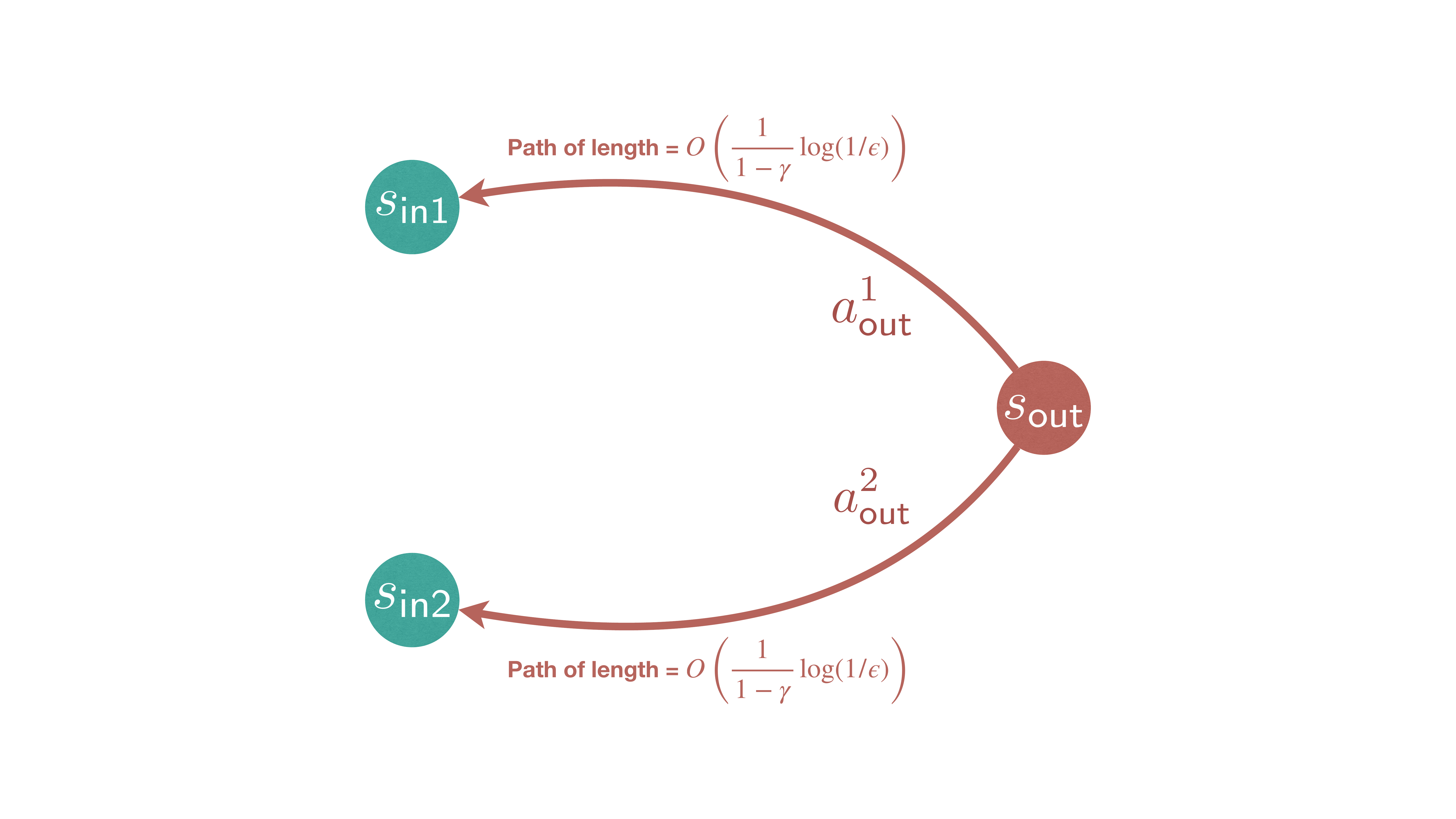}}
  \centerline{The comparison gadget.}\medskip
\end{minipage}

 \caption{Illustration of states and transitions for ``game gadgets'' used to construct an $\gcircuit$ instance using $\onestate$. Each vertex corresponds to a different player. Instantaneous rewards are specified in the tables corresponding to each gadget. All transitions are deterministic unless specified otherwise. The specific transitions marked in red represent a cycle or path of length $L$, where we define $L = \lceil\tfrac{4}{1-\gamma}\log\left(\frac{1}{\eps}\right)\rceil$.
 All the states in the cycle/path, besides $\out$, have a single action (of transiting to the next state in the cycle/path) and zero reward for all players.}

    \label{fig:mixedgamegadgets}
\end{figure}

Note that throughout we will use $\upsilon_i^{\ppi} = V_i^{\ppi}(s_i)$ as the utility function we consider, given the known equivalence in~\Cref{coro:equivalent-notion-NE} of NEs for $\onestate$. While we state the construction of gadgets with general $\gamma$ (as discount factor of $\tbsg$ we consider), and $\epsilon$ (as accuracy of gates in the $\gcircuit$ we implement), we only need it for constant $\gamma = 1/2$, and $\epsilon = \epsilon_0$ for~\Cref{thm:PPAD-hard}, and all the other cases of $1/2\le \gamma< 1$ and $\epsilon\le \epsilon_0$ follows as an immediate corollary (see~\Cref{coro:PPAD-hard}). All our constructions also use long paths / cycles with length $L \defeq \lceil \frac{4}{1-\gamma}\log\left(\frac{1}{\epsilon}\right)\rceil$ by default, which ensures $\gamma^L\le \epsilon^2$.

\paragraph{The equal gadget $G_{=}$:}~
\begin{itemize}
\item $\mathsf{INPUT}$:	Player $i = \mathsf{in}$ which controls state $s\subin$ and plays strategy $p\subin\in[0,1]$, corresponding to choosing action $a^1\subin$ with probability $p\subin$ and action $a^2\subin$ with probability $1-p\subin$.
\item $\mathsf{OUTPUT}$: Player $i= \mathsf{out}$ which controls state $s\subout$ and achieves the same strategy $p\subout\in[p\subin- \epsilon, p\subin+\epsilon]$ at $s\subout$ at a $\delta(\eps)$-approximate NE where $\delta(\eps) \leq \frac{(1-\gamma)\gamma^{L+1}}{8}\epsilon$.
\end{itemize}

We consider an input player $i$ with state $s\subin$, actions $a^1\subin$ and $a^2\subin$. We construct the output player $i= \mathsf{out}$ with state $s\subout$ and auxiliary player $i=\mathsf{aux}$ with state $s\subaux$. When player $\mathsf{out}$ plays action $a^1\subout$, with probability $1$ it transits to itself $s\subout$ through a length-$L$ cycle with dummy states; when it plays action $a^2\subout$, with probability $1$ it transits to $s\subaux$ through a length-$L$ path with dummy states, where we set $ L=\lceil \frac{4}{1-\gamma}\log\left(\frac{1}{\epsilon}\right)\rceil$. When player $\mathsf{aux}$ plays action $a\subaux^1$, with probability $1$ transits to state $s\subin$, when it plays action $a^2\subaux$, with probability $1$ transits to state $s\subout$. 
All states in the long cycle/path, besides $s\subout$ itself, have a single action (of transiting to the next state in the cycle/path) and zero reward for all players.
Moreover, the dummy states have a trivial action that they take in any (approximate) NE.

We will assume that both actions of $\mathsf{in}$ go through a cycle/path of length-$L$ first before going back to any states in the current gadget, throughout the claims we make on all gadget constructions. This assumption is important and we will explain why it is satisfied in our generic $\onestate$ in~\Cref{sssec:combininggadgets} when we discuss how to \textbf{combine gadgets}.

We begin by specifying the rewards that player $\mathsf{aux}$ and $\mathsf{out}$ receives at each state-action pair in~\Cref{tab:equal-reward-aux-out}. 

\begin{table}[htb!]
\centering
\caption{\textbf{Rewards of player $\mathsf{aux}$ and $\mathsf{out}$:}}\label{tab:equal-reward-aux-out}
\begin{tabular}{lllllll}
\multicolumn{1}{l|}{}     & \multicolumn{1}{l|}{$s\subin,a\subin^1$} & \multicolumn{1}{l|}{$s\subin,a\subin^2$} & \multicolumn{1}{l|}{$s\subaux, a^1\subaux$} & \multicolumn{1}{l|}{$s\subaux, a^2\subaux$} & \multicolumn{1}{l|}{$s\subout, a^1\subout$} & \multicolumn{1}{l}{$s\subout, a^2\subout$}  \\ \hline
\multicolumn{1}{l|}{$\aux$}   &
\multicolumn{1}{l|}{$\frac{1}{2}$} & \multicolumn{1}{l|}{$0$}   & 
\multicolumn{1}{l|}{$0$} & 
\multicolumn{1}{l|}{$-\frac{\gamma}{2}$}        & \multicolumn{1}{l|}{$1$}    
& \multicolumn{1}{l}{$\frac{1}{2}$}  
\\ \hline
\multicolumn{1}{l|}{$\mathsf{out}$}  &
\multicolumn{1}{l|}{$-\frac{1}{4}$}   &
\multicolumn{1}{l|}{$-\frac{1}{4}$} & \multicolumn{1}{l|}{$0$}      & \multicolumn{1}{l|}{$\frac{3\gamma}{4}$} & \multicolumn{1}{l|}{$\frac{\gamma^{L+1}}{4}$}        & \multicolumn{1}{l}{$0$}   
\end{tabular}
\end{table}

\begin{claim}
	For any $\gamma\in(0,1)$ and $\epsilon\le \frac{\gamma}{12}$, suppose both actions of $p\subin$ go through path with length at least $L = \lceil\frac{4}{1-\gamma}\log(1/\epsilon)\rceil$ before transiting to any states in the current gadget $\mathsf{in}$, $\mathsf{out}$, $\mathsf{aux}$, the above $\onestate$ gadget implements $p\subout\in[\max(0,p\subin -\epsilon), \min(1,p\subin +\epsilon)]$ for any $\delta$-approximate NE, where $\delta\le \delta \leq \frac{(1-\gamma)\gamma^{L+1}}{8}\epsilon$.
\end{claim}

\begin{proof}
We prove by contradiction. 
Suppose instead that $0 \le p\subout < p\subin - \eps$ at $\delta$-approximate NE.
Under this condition, we denote by $V\subaux^{(p\subaux)}(\mathsf{aux})$ the cumulative reward for player $\mathsf{aux}$ 
starting from initial distribution $\qq = \ee_{\mathsf{aux}}$, assuming it takes action $a^1\subaux$ with probability $p\subaux$ at its own state, and the input and output players play $p\subin$ and $p\subout$ respectively. 
We also drop the initial distribution $(\mathsf{aux})$ in the $\onestate$ setting when clear from context. 
Then, the definition of rewards in~\Cref{tab:equal-reward-aux-out} directly gives us 
\begin{align*}
 \frac{1}{2}\gamma p\subin\le V\subaux^{(1)} & \le \frac{1}{2}\gamma p\subin+\epsilon^2+\epsilon^4+\cdots\le  \frac{1}{2}\gamma p\subin+2\epsilon^2,
 \end{align*}
 where for the upper bound we also uses the fact that both actions in state $\mathsf{in}$ and $\mathsf{out}$ go into length-$L$ paths before generating another non-zero reward and that $\gamma^L\le \epsilon^2$.
We also can write the linear system for $V\subaux^{(0)}$ as
\begin{equation*}
\begin{aligned} 
& \begin{cases}
	V\subaux^{(0)} = -\frac{\gamma}{2}+\gamma V\subaux^{(0)}(\out)\\
	V\subaux^{(0)}(\out) = p\subout\left(1+\gamma ^L V\subaux^{(0)}(\out)\right) +(1-p\subout)\left(\frac{1}{2}+\gamma ^L V\subaux^{(0)}\right)
\end{cases}\\
\implies V\subaux^{(0)} & = \frac{\frac{\gamma}{2}(1+\gamma^L)p\subout}{(1- p\subout \gamma^L - \gamma^{L+1}(1-p\subout))}\in\frac{1}{2}\gamma p\subout\cdot \left[1, 1+4\epsilon^2\right],
\end{aligned}
\end{equation*}
where both ranges are due to the assumed choice of $\gamma^L\le \epsilon^2$ and the assumption of the claim that $\epsilon\le 1$.

Combining these two bounds, we obtain
\begin{equation}\label{eq:V-aux-equal}
\frac{1}{2}\gamma \left(p\subin-p\subout\right) -2\gamma\epsilon^2 \le V\subaux^{(1)} - V\subaux^{(0)}\le \frac{1}{2}\gamma\left(p\subin-p\subout\right)+2\epsilon^2.
\end{equation}

Now, we consider a $\delta$-approximate NE with $\delta \le  \frac{(1-\gamma)\gamma\epsilon}{16}$.
We show by pseudo-linearity of the utilities (\Cref{coro:quasi-mono-ssg}) that this must imply $p\subaux \ge 1-\frac{1}{4}\ge \frac{3}{4}$, i.e. player $\mathsf{aux}$ must play action $a\subaux^1$ with probability at least $3/4$. 
To see this, we apply~\eqref{eq:def-quasi-mono-bounded-velo} in~\Cref{coro:quasi-mono-ssg} with $i=\aux$, $\ppi_i = \ee_{a\subaux^1}$, $\ppi_i' = \ee_{a\subaux^2}$ and $\theta = 1-p\subaux$ to get
\begin{align*}
\frac{V\subaux^{(1)}-V\subaux^{(p\subaux)}}{V\subaux^{(1)} - V\subaux^{(0)}} & \ge (1-\gamma)(1-p\subaux)\\
\implies \frac{(1-\gamma)\gamma}{16}\epsilon\stackrel{(i)}{\ge} V\subaux^{(1)}-V\subaux^{(p\subaux)} & \ge (1-\gamma)(1-p\subaux)\left(V\subaux^{(1)}-V\subaux^{(0)}\right)\\
& \stackrel{(ii)}{\ge} \frac{(1-\gamma)\gamma}{3}(1-p\subaux)\epsilon,\\
\implies p\subaux\ge 1-\frac{3}{16}\ge \frac{3}{4},
\end{align*}
where we use $(i)$ definition of $\delta$-approximate NE and $(ii)$ the inequality~\eqref{eq:V-aux-equal} with choice of $\epsilon$.

Thus, we have $p\subaux>\frac{3}{4}$ at any $\delta$-approximate NE. 
We show that this in turn requires player $\out$ to play $p\subout \geq 1 - \eps \geq p\subin - \eps$, which is the desired contradiction.
We bound the value functions for player $\mathsf{out}$ under such strategy of player $\mathsf{aux}$, starting from the initial distribution $\qq = \ee\subout$ as
\begin{align*}
V\subout^{(1)} & = \frac{\gamma^{L+1}}{4(1-\gamma^L)}\ge \frac{\gamma^{L+1}}{4},\\
V\subout^{(0)} & \le  \gamma^L p\subaux\Par{0+\gamma \Par{-\frac{1}{4}}} + \gamma^L\left(1-p\subaux\right)\left(\frac{3\gamma}{4}+\gamma V\subout^{(0)}\right)\\
~~\implies ~~V\subout^{(0)} & \le \frac{-\frac{\gamma^{L+1}}{4}p\subaux+\frac{3\gamma^{L+1}}{4}(1-p\subaux)}{1-\gamma^{L+1}(1-p\subaux)}\le \frac{-\frac{3}{4}\gamma^{L+1}+\frac{3}{4}\gamma^{L+1}}{1-\gamma^{L+1}(1-p\subaux)} = 0.
\end{align*}
We then apply~\eqref{eq:def-quasi-mono-bounded-velo} in~\Cref{coro:quasi-mono-ssg} with $i=\out$, $\ppi_i = \ee_{a\subout^1}$, $\ppi_i' = \ee_{a\subout^2}$ and $\theta = 1-p\subout$ to get
\begin{align*}
& \frac{V\subout^{(1)}-V\subout^{(p\subout)}}{V\subout^{(1)} - V\subout^{(0)}}\ge (1-\gamma)(1-p\subout),\\
\implies & \frac{(1-\gamma)\gamma^{L+1}}{4}\epsilon \stackrel{(i)}{\ge}  V\subout^{(1)}-V\subout^{(p\subout)}\ge (1-\gamma)(1-p\subout)\frac{\gamma^{L+1}}{4},\\
\implies & p\subout\ge 1-\epsilon,
\end{align*}
where we use again $(i)$ the definition of a $\delta$-approximate NE with $\delta\le \frac{(1-\gamma)\gamma^{L+1}}{4}\epsilon$.  

We now similarly consider the other contradictory case.
Suppose instead that $1\ge p\subout>p\subin+\epsilon$.
In this case, a $\delta$-approximate NE must have player $\mathsf{aux}$ taking strategy $p\subaux\le \frac{1}{4}$ for achieving the corresponding $\delta$-approximate NE, since by a similar application of~\Cref{coro:quasi-mono-ssg},~\Cref{eq:V-aux-equal}, definition of $\delta$-approximate NE, and choice of $\epsilon$ we have
\begin{align*}
& \frac{V\subaux^{(0)}-V\subaux^{(p\subaux)}}{V\subaux^{(0)}-V\subaux^{(1)}}\ge (1-\gamma)p\subaux,\\
\implies & \frac{(1-\gamma)\gamma\epsilon}{16}\ge V\subaux^{(0)}-V\subaux^{(p\subaux)}\ge (1-\gamma)p\subaux\left(V\subaux^{(0)}-V\subaux^{(1)}\right)\ge \frac{(1-\gamma)\gamma\epsilon}{3}p\subaux,\\
\implies & p\subaux\le \frac{3}{16}\le \frac{1}{4}.
\end{align*}

Under such choice of player $\mathsf{aux}$, the value function  of player $\mathsf{out}$ starting from its own state  becomes
\begin{align*}
	V\subout^{(1)} & =  \frac{\gamma^{L+1}}{4(1-\gamma^L)}\le \frac{1}{4}\gamma^{L+1}\left(1+2\epsilon^2\right)\\ 
V\subout^{(0)} & =\frac{-\frac{\gamma^{L+1}}{4}p\subaux+\frac{3\gamma^{L+1}}{4}(1-p\subaux)}{1-\gamma^{L+1}(1-p\subaux)}\ge \frac{\gamma^{L+1}}{2(1-\gamma^{L+1}(1-p\subaux))}\ge \frac{1}{2}\gamma^{L+1}.
\end{align*}

Combining this with a similar application of~\Cref{coro:quasi-mono-ssg} and choice of $\delta$, $\epsilon$, we get 
\begin{align*}
& \frac{V\subout^{(0)}-V\subout^{(p\subout)}}{V\subout^{(0)}-V\subout^{(1)}}\ge (1-\gamma)p\subout\\
\implies & \frac{(1-\gamma)\gamma^{L+1}}{8}\epsilon\ge V\subout^{(0)}-V\subout^{(p\subout)}\ge (1-\gamma)p\subout\left(V\subout^{(0)}-V\subout^{(1)}\right)\ge \frac{(1-\gamma)\gamma^{L+1}}{8}p\subout,\\
\implies p\subout\le \epsilon,
\end{align*}
and thus leads to a contradiction with the assumption. 

Since the existence of an approximate NE is guaranteed for such games, we conclude that at the approximate NE, it must hold that $p\subout \in [\max(0,p\subin -\epsilon), \min(1,p\subin +\epsilon)] $at any $\delta$-approximate NE where $\delta \leq \frac{(1-\gamma)\gamma^{L+1}}{8}\epsilon$.
\end{proof}

\paragraph{The constant gadget $G_{\alpha}$:}~
\begin{itemize}
\item $\mathsf{INPUT}$:	Some scalar $\alpha\in\R$.
\item $\mathsf{OUTPUT}$: Player $i= \mathsf{out}$ which controls state $s\subout$ and plays strategy $p\subout\in[\med(0,\alpha-\epsilon, 1-\epsilon), \med(1,\alpha +\epsilon, \epsilon)]$ at $s\subout$ at a $\delta(\epsilon)$-approximate NE where  $\delta(\epsilon)\leq \frac{(1-\gamma)\gamma^{L+1}}{8}\epsilon$.
(Here $\med(\cdot)$ denotes the median of the tuple of numbers.)
\end{itemize}

We consider a construction of gadget with states $s\subaux$, $s\subout$, and specify the transition probabilities of each player as follows: When auxiliary player $\mathsf{aux}$ plays $a^1\subaux$ at its own state $s\subaux$, with probability $1$ it transits to itself, when it plays $a^2\subaux$, with probability $1$ it transits to $\mathsf{out}$. When player $\mathsf{out}$ plays action $a^1\subout$, with probability $1$ it transits to itself $s\subout$ through length-$L$ cycle; when it plays action $a^2\subout$, with probability $1$ it transits to $s\subaux$ through length-$L$ path. The rewards each player of $\mathsf{aux}$, $\mathsf{out}$ receives at each state-action pair are defined in~\Cref{tab:const-reward-both} -  we omit the all-zero rewards they receive along the length-$L$ cycle / path as previous gadget construction for simplicity.

\begin{table}[htb!]
\centering
\caption{\textbf{Rewards of players $\mathsf{aux}$, $\mathsf{out}$:} $\alpha\in\R$.}\label{tab:const-reward-both}
\begin{tabular}{lllllll}
\multicolumn{1}{l|}{}     & \multicolumn{1}{l|}{$s\subaux,a^1\subaux$} & \multicolumn{1}{l|}{$s\subaux,a^2\subaux$} & \multicolumn{1}{l|}{$s\subout,a^1\subout$} & \multicolumn{1}{l}{$s\subout,a^2\subout$}  \\ \hline
\multicolumn{1}{l|}{$\mathsf{aux}$}     & \multicolumn{1}{l|}{$\frac{\gamma(1-\gamma) \alpha}{2}$} & \multicolumn{1}{l|}{$-\frac{\gamma}{2}$}        & \multicolumn{1}{l|}{$1$}    & \multicolumn{1}{l}{$\frac{1}{2}$}\\
\hline
\multicolumn{1}{l|}{$\mathsf{out}$}    & \multicolumn{1}{l|}{$-\frac{\gamma}{4}$} & \multicolumn{1}{l|}{$\frac{3\gamma}{4}$}        & \multicolumn{1}{l|}{$\frac{\gamma^{L+1}}{4}$}    & \multicolumn{1}{l}{$0$}
\end{tabular}
\end{table}

\begin{claim} For any $\gamma\in(0,1)$ and $\epsilon\le \frac{1}{12}$, the above $\onestate$ gadget implements $p\subout\in [\med(0,\alpha -\epsilon,1-\epsilon), \med(1,\alpha+\epsilon,\epsilon)]$ for any $\delta$-approximate NE, where $\delta\le \frac{(1-\gamma)\gamma^{L+1}}{8}\epsilon$.
\end{claim}

\begin{proof}
We prove the upper and lower bound respectively. Suppose  $0\le p\subout<\alpha-\epsilon$ (as for other cases the lower bound naturally holds). Under such condition, for player $\mathsf{aux}$, we let $V\subaux^{(p\subaux)}(\mathsf{aux})$ be the cumulative reward starting from initial distribution $\qq = \ee_{\mathsf{aux}}$, assuming it takes action $a^1\subaux$ with probability $p\subaux$  at its own state, and input and output players play $p\subin$ and $p\subout$ respectively. We also drop the initial distribution $(\mathsf{aux})$ in the $\onestate$ setting when clear from context. Then the rewards in~\Cref{tab:const-reward-both} and same calculation as in~the equal gadget directly gives us
\begin{align*}
V\subaux^{(1)} & = \frac{\gamma \alpha(1-\gamma)}{2(1-\gamma)} = \frac{\gamma}{2}\alpha,\\
V\subaux^{(0)} & = \frac{\frac{\gamma}{2}(1+\gamma^L)p\subout}{(1- p\subout \gamma^L -\gamma^{L+1}(1-p\subout))}\in\frac{\gamma}{2} p\subout\cdot \left[1, 1+4\epsilon^2\right].
\end{align*}

Consequently, since the current strategy achieves $\delta$-NE with $\delta \le  \frac{(1-\gamma)\gamma}{16}\epsilon$, we show by pseudo-linearity of utility function (\Cref{coro:quasi-mono-ssg}) similar as in equal gadget that  $p\subaux \ge \frac{3}{4}$, i.e. player $\mathsf{aux}$ plays action $a\subaux^1$ with probability at least $3/4$. We show that this in turn implies player $\mathsf{out}$ needs to play $p\subout\ge 1-\epsilon$, which combined with the assumption that $0\le p\subout<\alpha-\epsilon$ shows the lower bound. To see this, we consider the value functions for player $\mathsf{out}$ under such strategy of player $\mathsf{aux}$, following similar calculations as the equal gadget we have
\begin{align*}
V\subout^{(1)} & \ge \frac{\gamma^{L+1}}{4}~~\text{and}~~V\subout^{(0)} \le 0,
\end{align*}
which implies $p\subout \ge 1-\epsilon$ at the $\delta$-approximate NE with $\delta\le \frac{(1-\gamma)\gamma^{L+1}}{4}\epsilon$ given pseudo-linearity in~\Cref{coro:quasi-mono-ssg} and thus leads to the condition that $p\subout\ge \med(0,\alpha-\epsilon,1-\epsilon)$. 

Similarly, for the upper bound, it suffices to consider the case when $1\ge p\subout>\alpha+\epsilon$, we have player $\mathsf{aux}$ must take strategy $p\subaux\le \frac{1}{4}$ for achieving the corresponding $\delta$-approximate NE, i.e. take action $a^2\subaux$ with at least probability $3/4$. Under such choice of player $\mathsf{aux}$, the value functions  of player $\mathsf{out}$ remains the same as the equal gadget case, which implies $p\subout\le \epsilon$ at the $\delta$-approximate NE and thus leads to the condition that $p\subout\le \med(1,\alpha+\epsilon,\epsilon)$. Since the existence of an approximate NE is guaranteed for such games, we conclude that at the corresponding approximate NE, it holds that $p\subout \in[\med(0,\alpha-\epsilon, 1-\epsilon), \med(1,\alpha +\epsilon, \epsilon)]$.
\end{proof}

\paragraph{The multiplicative gadget $G_{\times}$:}~
\begin{itemize}
\item $\mathsf{INPUT}$:	Player $i = \mathsf{in}$ which controls state $s\subin$ and plays strategy $p\subin\in[0,1]$, corresponding to choosing action $a^1\subin$ with probability $p\subin$ and action $a^2\subin$ with probability $1-p\subin$, some scalar $0<\alpha\le 2$.
\item $\mathsf{OUTPUT}$: Player $i= \mathsf{out}$ which controls state $s\subout$ and plays strategy $p\subout\in[\med(0,\alpha p\subin-\epsilon,1-\epsilon), \med(1,\alpha p\subin+\epsilon,\epsilon)]$ at $s\subout$ at a $\delta(\epsilon)$-approximate NE where $\delta(\epsilon)\le \frac{(1-\gamma)\gamma^{L+1}}{8}\epsilon$.
\end{itemize}

We consider a similar construction of gadget with states $s\subin$, $s\subaux$, $s\subout$, and the same transition probabilities under each two actions for the states. The rewards each player of $\mathsf{aux}$, $\mathsf{out}$ receives at each state-action pair is defined in~\Cref{tab:multi-reward-both}.

\begin{table}[htb!]
\centering
\caption{\textbf{Rewards of players $\mathsf{aux}$, $\mathsf{out}$:} $\alpha>0$.}\label{tab:multi-reward-both}
\begin{tabular}{lllllll}
\multicolumn{1}{l|}{}     & \multicolumn{1}{l|}{$s\subin,a^1\subin$} & \multicolumn{1}{l|}{$s\subin,a^2\subin$} & \multicolumn{1}{l|}{$s\subaux,a^1\subaux$} & \multicolumn{1}{l|}{$s\subaux,a^2\subaux$} & \multicolumn{1}{l|}{$s\subout,a^1\subout$} & \multicolumn{1}{l}{$s\subout,a^2\subout$}  \\ \hline
\multicolumn{1}{l|}{$\mathsf{aux}$}   &
\multicolumn{1}{l|}{$\frac{\alpha}{2}$} & \multicolumn{1}{l|}{$0$}      & \multicolumn{1}{l|}{$0$} & \multicolumn{1}{l|}{$-\frac{\gamma}{2}$}        & \multicolumn{1}{l|}{$1$}    & \multicolumn{1}{l}{$\frac{1}{2}$}\\
\hline
\multicolumn{1}{l|}{$\mathsf{out}$}   &
\multicolumn{1}{l|}{$-\frac{1}{4}$} & \multicolumn{1}{l|}{$-\frac{1}{4}$}      & \multicolumn{1}{l|}{$0$} & \multicolumn{1}{l|}{$\frac{3\gamma}{4}$}        & \multicolumn{1}{l|}{$\frac{\gamma^{L+1}}{4}$}    & \multicolumn{1}{l}{$0$}
\end{tabular}
\end{table}

\begin{claim}
	For any $\alpha\in(0,2]$, $\gamma\in(0,1)$ and $\epsilon\le \frac{\gamma}{12}$, suppose both actions of $p\subin$ go through path with length at least $L = \lceil\frac{4}{1-\gamma}\log(1/\epsilon)\rceil$ before transiting to any states in the current gadget $\mathsf{in}$, $\mathsf{out}$, $\mathsf{aux}$, the above $\onestate$ gadget implements $p\subout\in [\med(0,\alpha p\subin-\epsilon,1-\epsilon), \med(1,\alpha p\subin+\epsilon,\epsilon)]$ for any $\delta$-approximate NE, where $\delta \leq \frac{(1-\gamma)\gamma^{L+1}}{8}\epsilon$.
\end{claim}

\begin{proof}
We prove by contradiction. 	Similar to the calculations for the equal gadget we have 
\begin{align*}
V\subaux^{(1)} & \in 
 \left[\frac{\gamma}{2}\alpha p\subin, \frac{\gamma}{2}\alpha p\subin+2\epsilon^2\right],\\
V\subaux^{(0)} & = \frac{\frac{\gamma}{2}(1+\gamma^L)p\subout}{(1- p\subout \gamma^L +\gamma^{L+1}(1-p\subout))}\in\frac{\gamma}{2} p\subout\cdot \left[1, 1+4\epsilon^2\right].
\end{align*}

Now suppose $0\le p\subout<\alpha p\subin-\epsilon$, then at $\delta$-NE for given $\delta$ player $\mathsf{aux}$ will take action $a\subaux^{1}$ with probability no less than $3/4$, which leads to $p\subout\ge 1-\epsilon$ due to rewards of player $\mathsf{out}$, following the same calculations in the equal gadget. The argument applies symmetrically for the upper bound.
\end{proof}

\paragraph{The sum gadget $G_{+}$:}~
\begin{itemize}
\item $\mathsf{INPUT}$:	Two players $i = \mathsf{in}1, \mathsf{in2}$ which control state $s\subino$, $s\subint$ and play strategy $p\subino,p\subint\in[0,1]$,  corresponding to choosing action $a^1_i$ with probability $p_i$ and action $a^2_i$ with probability $1-p_i$, for $i\in\{\mathsf{in}1,\mathsf{in}2\}$, respectively.
\item $\mathsf{OUTPUT}$: Player $i= \mathsf{out}$ which controls state $s\subout$ and plays strategy $p\subout\in[\max(0,p\subino+p\subint-\epsilon), \min(1, p\subino+p\subint+\epsilon)]$ at $s\subout$ at a $\delta(\epsilon)$-approximate NE where $\delta(\eps) \leq \frac{(1-\gamma)\gamma^{L+1}}{8}\epsilon$.
\end{itemize}

We consider a similar construction of gadget with states $s\subino$, $s\subint$, $s\subaux$, $s\subout$, and the same transition probabilities under each two actions for the state $s\subout$. For $s\subaux$, it still has two actions $a\subaux^1$ and $a\subaux^2$. If it takes action $a\subaux^1$, with probability  $1/2$ it transits to state $s\subino$, with probability $1/2$ it transits to state $s\subint$. When taking action $a\subaux^2$, with probability $1$ it transits to state $s\subout$. The rewards each player of $\mathsf{aux}$, $\mathsf{out}$ receives at state-action pairs of states $s\subino$ and $s\subint$ are defined in~\Cref{tab:sum-reward-cumu}. The instant rewards each player of $\mathsf{aux}$, $\mathsf{out}$ receives at state-action pairs of states $s\subaux$ and $s\subout$ are defined in~\Cref{tab:sum-reward-inst}.

\begin{table}[htb!]
\centering
\caption{\textbf{Rewards of players $\mathsf{aux}$, $\mathsf{out}$.}}\label{tab:sum-reward-cumu}
\begin{tabular}{lllll}
\multicolumn{1}{l|}{}     & \multicolumn{1}{l|}{$s\subino,a^1\subino$} & \multicolumn{1}{l|}{$s\subino,a^2\subino$} &
\multicolumn{1}{l|}{$s\subint,a^1\subint$} & \multicolumn{1}{l}{$s\subint,a^2\subint$} \\ \hline
\multicolumn{1}{l|}{$\mathsf{aux}$}   &
\multicolumn{1}{l|}{$1$} & \multicolumn{1}{l|}{$0$}   & 
\multicolumn{1}{l|}{$1$} & \multicolumn{1}{l}{$0$}   \\
\hline
\multicolumn{1}{l|}{$\mathsf{out}$}   &
\multicolumn{1}{l|}{$-\frac{1}{4}$} & \multicolumn{1}{l|}{$-\frac{1}{4}$} &
\multicolumn{1}{l|}{$-\frac{1}{4}$} & \multicolumn{1}{l}{$-\frac{1}{4}$} 
\end{tabular}
\end{table}

\begin{table}[htb!]
\centering
\caption{\textbf{Rewards of players $\mathsf{aux}$, $\mathsf{out}$.} }\label{tab:sum-reward-inst}
\begin{tabular}{lllll}
\multicolumn{1}{l|}{}   &\multicolumn{1}{l|}{$s\subaux,a^1\subaux$} & \multicolumn{1}{l|}{$s\subaux,a^2\subaux$} & \multicolumn{1}{l|}{$s\subout,a^1\subout$} & \multicolumn{1}{l}{$s\subout,a^2\subout$}  \\ \hline
\multicolumn{1}{l|}{$\mathsf{aux}$}   & 
\multicolumn{1}{l|}{$0$} & \multicolumn{1}{l|}{$-\frac{\gamma}{2}$}        & \multicolumn{1}{l|}{$1$}    & \multicolumn{1}{l}{$\frac{1}{2}$}\\
\hline
\multicolumn{1}{l|}{$\mathsf{out}$}   & \multicolumn{1}{l|}{$0$} & \multicolumn{1}{l|}{$\frac{3\gamma}{4}$}        & \multicolumn{1}{l|}{$\frac{\gamma^{L+1}}{4}$}    & \multicolumn{1}{l}{$0$}
\end{tabular}
\end{table}

\begin{claim}
	For any $\gamma\in(0,1)$ and $\epsilon\le \frac{\gamma}{12}$, suppose both actions of $p\subino$, $p\subint$ go through path with length at least $L = \lceil\frac{4}{1-\gamma}\log(1/\epsilon)\rceil$ before transiting to any states in the current gadget $\mathsf{in1}$, $\mathsf{in2}$, $\mathsf{out}$, $\mathsf{aux}$, the above $\onestate$ gadget implements $p\subout\in [\med(0,p\subino+p\subint-\epsilon,1-\epsilon), \min(1, p\subino+p\subint+\epsilon)]  $ for any $\delta$-approximate NE, where $\delta \leq \frac{(1-\gamma)\gamma^{L+1}}{8}\epsilon$.
\end{claim}

\begin{proof}
We prove by contradiction. 	Similar to the calculations for the previous gadgets we have 
\begin{align*}
V\subaux^{(1)} & \in  \left[\frac{\gamma}{2}\left(p\subino + p\subint \right),\frac{\gamma}{2}\left(p\subino + p\subint \right)+2\epsilon^2\right],\\
V\subaux^{(0)} & = \frac{\frac{\gamma}{2}(1+\gamma^L)p\subout}{(1- p\subout \gamma^L - \gamma^{L+1}(1-p\subout))}\in\frac{\gamma}{2} p\subout\cdot \left[1, 1+4\epsilon^2\right].
\end{align*}

Now suppose $0\le p\subout< p\subino+p\subint-\epsilon$, then at $\delta$-NE for given $\delta$ player $\mathsf{aux}$ will take action $a\subaux^{1}$ with probability no less than $3/4$, which leads to $p\subout\ge1-\epsilon$ due to rewards of player $\mathsf{out}$, following the same calculations in previous gadget and proves the lower bound condition on $p\subout$. The argument applies symmetrically to the upper bound.
\end{proof}

\paragraph{The subtraction gadget $G_{-}$:}~
\begin{itemize}
\item $\mathsf{INPUT}$:	Two players $i = \mathsf{in}1, \mathsf{in2}$ which control state $s\subino$, $s\subint$ and play strategy $p\subino,p\subint\in[0,1]$, corresponding to choosing action $a^1_i$ with probability $p_i$ and action $a^2_i$ with probability $1-p_i$, for $i\in\{\mathsf{in}1,\mathsf{in}2\}$, respectively.
\item $\mathsf{OUTPUT}$: Player $i= \mathsf{out}$ which controls state $s\subout$ and plays strategy $p\subout\in[\max(0,p\subino-p\subint-\epsilon), \med(1, p\subino-p\subint+\epsilon,\epsilon)]$ at $s\subout$ at a $\delta(\epsilon)$-approximate NE where $\delta(\eps) \leq \frac{(1-\gamma)\gamma^{L+1}}{8}\epsilon$.
\end{itemize}

We consider a similar construction of gadget with states $s\subino$, $s\subint$, $s\subaux$, $s\subout$, and the same transition probabilities under each two actions for the state $s\subaux$, $s\subout$ as in the construction for the sum gadget. The rewards each player of $\mathsf{aux}$, $\mathsf{out}$ receives at state-action pairs of states $s\subino$ and $s\subint$ are defined in~\Cref{tab:subtract-reward-cumu}. The instant rewards each player of $\mathsf{aux}$, $\mathsf{out}$ receives at state-action pairs of states $s\subaux$ and $s\subout$ are defined in~\Cref{tab:subtract-reward-inst}.

\begin{table}[htb!]
\centering
\caption{\textbf{Rewards of players $\mathsf{aux}$, $\mathsf{out}$.} }\label{tab:subtract-reward-cumu}
\begin{tabular}{lllll}
\multicolumn{1}{l|}{}     & \multicolumn{1}{l|}{$s\subino,a^1\subino$} & \multicolumn{1}{l|}{$s\subino,a^2\subino$} &
\multicolumn{1}{l|}{$s\subint,a^1\subint$} & \multicolumn{1}{l}{$s\subint,a^2\subint$} \\ \hline
\multicolumn{1}{l|}{$\mathsf{aux}$}   &
\multicolumn{1}{l|}{$1$} & \multicolumn{1}{l|}{$0$}   & 
\multicolumn{1}{l|}{$-1$} & \multicolumn{1}{l}{$0$}   \\
\hline
\multicolumn{1}{l|}{$\mathsf{out}$}   &
\multicolumn{1}{l|}{$-\frac{1}{4}$} & \multicolumn{1}{l|}{$-\frac{1}{4}$} &
\multicolumn{1}{l|}{$-\frac{1}{4}$} & \multicolumn{1}{l}{$-\frac{1}{4}$} 
\end{tabular}
\end{table}

\begin{table}[htb!]
\centering
\caption{\textbf{Rewards of players $\mathsf{aux}$, $\mathsf{out}$.} }\label{tab:subtract-reward-inst}
\begin{tabular}{lllll}
\multicolumn{1}{l|}{}   &\multicolumn{1}{l|}{$s\subaux,a^1\subaux$} & \multicolumn{1}{l|}{$s\subaux,a^2\subaux$} & \multicolumn{1}{l|}{$s\subout,a^1\subout$} & \multicolumn{1}{l}{$s\subout,a^2\subout$}  \\ \hline
\multicolumn{1}{l|}{$\mathsf{aux}$}   & 
\multicolumn{1}{l|}{$0$} & \multicolumn{1}{l|}{$-\frac{\gamma}{2}$}        & \multicolumn{1}{l|}{$1$}    & \multicolumn{1}{l}{$\frac{1}{2}$}\\
\hline
\multicolumn{1}{l|}{$\mathsf{out}$}   & \multicolumn{1}{l|}{$0$} & \multicolumn{1}{l|}{$\frac{3\gamma}{4}$}        & \multicolumn{1}{l|}{$\frac{\gamma^{L+1}}{4}$}    & \multicolumn{1}{l}{$0$}
\end{tabular}
\end{table}
\begin{claim}
	For any $\gamma\in(0,1)$ and $\epsilon\le \frac{\gamma}{12}$, suppose both actions of $p\subino$, $p\subint$ go through path with length at least $L = \lceil\frac{4}{1-\gamma}\log(1/\epsilon)\rceil$ before transiting to any states in the current gadget $\mathsf{in1}$, $\mathsf{in2}$, $\mathsf{out}$, $\mathsf{aux}$, the above $\onestate$ gadget implements $p\subout\in [\max(0,p\subino-p\subint-\epsilon), \med(1, p\subino-p\subint+\epsilon,\epsilon)]  $ for any $\delta$-approximate NE, where $\delta \leq \frac{(1-\gamma)\gamma^{L+1}}{8}\epsilon$.
\end{claim}

\begin{proof}
We prove by contradiction. 	Similar to the calculations for the previous gadgets we have 
\begin{align*}
V\subaux^{(1)} & \in 
 	\left[\frac{\gamma}{2}\left(p\subino - p\subint \right)-2\epsilon^2, \frac{\gamma}{2}\left(p\subino - p\subint \right)+2\epsilon^2\right],\\
V\subaux^{(0)} & =\frac{\frac{\gamma}{2}(1+\gamma^L)p\subout}{(1- p\subout \gamma^L - \gamma^{L+1}(1-p\subout))}\in\frac{\gamma}{2} p\subout\cdot \left[1, 1+4\epsilon^2\right].
\end{align*}

Now suppose $0\le p\subout< p\subino-p\subint-\epsilon$, then at $\delta$-NE for given $\delta$ player $\mathsf{aux}$ will take action $a\subaux^{1}$ with probability no less than $3/4$, which leads to contradiction due to rewards of player $\mathsf{out}$, following the same calculations in previous gadget. The argument applies symmetrically for the upper bound.
\end{proof}

\paragraph{The comparison gadget $G_{>}$:}~
\begin{itemize} 
\item $\mathsf{INPUT}$:	Two players $i = \mathsf{in}1, \mathsf{in2}$ which control state $s\subino$, $s\subint$ and play strategy $p\subino,p\subint\in[0,1]$, corresponding to choosing action $a^1_i$ with probability $p_i$ and action $a^2_i$ with probability $1-p_i$, for $i\in\{\mathsf{in}1,\mathsf{in}2\}$, respectively.
\item $\mathsf{OUTPUT}$: Player $i= \mathsf{out}$ which controls state $s\subout$ and plays strategy $p\subout\ge 1-\epsilon$ if $p\subino\ge p\subint+\epsilon$, and $p\subout\le \epsilon$ if $p\subino\le p\subint-\epsilon$ at $s\subout$ at a $\delta(\epsilon)$-approximate NE, where $\delta(\epsilon)\le \frac{(1-\gamma)\gamma^L}{2}\epsilon^2$.
\end{itemize}

We consider a simplified construction of previous gadgets with states $s\subino$, $s\subint$, $s\subout$, and no auxiliary player / state. If playing action $a\subout^1$ at $s\subout$, it transits to $s\subino$ with probability $1$ through a length-$L$ path; if playing action $a\subout^2$ at $s\subout$, it transits to $s\subint$ with probability $1$ through a length-$L$ path. The rewards player $\mathsf{out}$ receives at state-action pairs are defined in~\Cref{tab:comp-reward}.

\begin{table}[htb!]
\centering
\caption{\textbf{Rewards of player $\mathsf{out}$.} }\label{tab:comp-reward}
\begin{tabular}{lllllll}
\multicolumn{1}{l|}{}     & \multicolumn{1}{l|}{$s\subino,a^1\subino$} & \multicolumn{1}{l|}{$s\subino,a^2\subino$} &
\multicolumn{1}{l|}{$s\subint,a^1\subint$} & \multicolumn{1}{l|}{$s\subint,a^2\subint$} & \multicolumn{1}{l|}{$s\subout,a^1\subout$}& \multicolumn{1}{l}{$s\subout,a^2\subout$}\\ \hline
\multicolumn{1}{l|}{$\mathsf{out}$}   &
\multicolumn{1}{l|}{$1$} & \multicolumn{1}{l|}{$0$}   & 
\multicolumn{1}{l|}{$1$} & \multicolumn{1}{l|}{$0$}  & \multicolumn{1}{l|}{$0$} & \multicolumn{1}{l}{$0$}  
\end{tabular}
\end{table}

\begin{claim}
	For any $\gamma\in(0,1)$ and $\epsilon\le \frac{\gamma}{12}$, suppose both actions of $p\subino$, $p\subint$ go through path with length at least $L = \lceil\frac{4}{1-\gamma}\log(1/\epsilon)\rceil$ before transiting to any states in the current gadget $\mathsf{in1}$, $\mathsf{in2}$, $\mathsf{out}$, $\mathsf{aux}$, the above $\onestate$ gadget implements $p\subout\ge 1-\epsilon$ if $p\subino\ge p\subint+\epsilon$, and $p\subout\le \epsilon$ if $p\subino\le p\subint-\epsilon$ for any $\delta$-approximate NE, where $\delta \leq \frac{(1-\gamma)\gamma^L}{2}\epsilon^2$.
\end{claim}

\begin{proof}
Denote $V^{(p\subout)}\subout$ as the value function for player $\mathsf{out}$ starting from its own state $\qq = \ee\subout$, when the input players play strategy $p\subino$ and $p\subint$ respectively, and the output player plays action $a^1\subout, a^2\subout$ with probability $p\subout$, we have 
\begin{align*}
V\subout^{(1)} & \in  \gamma^{L} p\subino\cdot \left[1, 1+\epsilon^2+\epsilon^4\cdots\right]\subseteq \gamma^{L} p\subino\cdot \left[1, 1+2\epsilon^2\right],\\
V\subout^{(0)} & = \gamma^{L}p\subint\cdot \left[1, 1+\epsilon^2+\epsilon^4\cdots\right]\subseteq \gamma^{L} p\subint\cdot \left[1, 1+2\epsilon^2\right].
\end{align*}

Now suppose $p\subino> p\subint+\epsilon$, we show that at $\delta$-NE for given $\delta \le\frac{(1-\gamma)\gamma^L}{2}\epsilon^2 $ player $\mathsf{out}$ will take action $a\subaux^{1}$ with probability no less than $1-\epsilon$. To see this, applying~\Cref{coro:quasi-mono-ssg} with $i=\out$, $\ppi_i = \ee_{a\subout}^1$, $\ppi_i' = \ee_{a\subout}^2$ and $\theta = 1-p\subout$, we have 
\[
\frac{V\subout^{(1)}-V\subout^{(p\subout)}}{V\subout^{(1)} - V\subout^{(0)}}> (1-\gamma)(1-p\subout)\implies V\subout^{(1)}-V\subout^{(p\subout)}\ge (1-\gamma)(1-p\subout)\frac{\gamma^L\epsilon}{2},
\] 
which implies $p\subout \ge 1-\epsilon$ at a $\delta$-approximate NE. The argument applies symmetrically for the other case.
\end{proof}

\subsubsection{Combining gadgets together}\label{sssec:combininggadgets}

The traditional reductions~\cite{daskalakis2009complexity,chen2006settling} of $n$-vertex graphical games to $\gcircuit$ involves combining the individual sub-gadgets described above ``in series''.
What this means is that an output vertex for a certain gate can be used as an input vertex, but \emph{not} an output vertex for another gate.
In other words, for two gates $T,T'$ we can have $s'\subin = s\subout$ but need to have $s'\subout \neq s\subout$.
This is consistent with the definition of $\gcircuit$ provided in~\Cref{def:gcircuit}.

To complete our proof of $\ppad$-hardness we show that a similar idea works for $\onestate$.
In other words, we posit the following structure in the transitions when combining gadgets $G$ and $G'$ with $s\subout = s'\subin$:
\begin{enumerate}
\item Outgoing transitions from $s\subout$ is solely decided by the gadget $G$.
\item Since $s\subout = s'\subin$ in the new gadget $G'$, the outgoing transitions of $s'\subin$ are completely independent of the structure of gadget $G'$, while ingoing transitions from $s'\subaux$ or $s'\subout$ is decided by the structure of gadget $G'$.
\item Since both actions of $s\subout$ of all of the gadgets constructed above go through cycles or paths of length at least $L$ before (if it ever) transits to themselves, the assumptions of the previous claims for properties of constructed gadgets hold. 
\end{enumerate}
Mapping this structure to the original $\gcircuit$ instance, they ensure outputs of arbitrary sub-gadgets (as corresponding to nodes / gates in $\gcircuit$) can be combined to be the new input of the subsequent sub-gadget (as node in $\gcircuit$), without changing the internal game structure that we require of each sub-gadget. This ensures that we can create an arbitrary $\gcircuit$ instance from $\onestate$ by linking $\onestate$ sub-gadgets together and treating each vertex as a separate player.
The rewards of the sub-gadget players are kept unchanged at the original states in the sub-gadget, and the instantaneous rewards of the sub-gadget players when they transit to a state outside their sub-gadget are identically set to zero.
\emph{Note that this also ensures that each player receives reward at no more than $4$ states.}
The final step is to argue that the equilibrium strategies at each vertex remain unchanged as a consequence of this combination.
For this, we critically rely on the $\onestate$ structure and the satisfaction of assumptions made in the claims.
More concretely, we show the property of approximate NE of a designated output state $s\subout$ only relies on the transition probabilities and rewards defined for the additional auxiliary player $\aux$, and output player $\out$, and never depends on the input player $\mathsf{in}$'s actions, probability transitions,  or instantaneous rewards, besides each player's specified instant rewards at input players, and the fact that all actions of input players go through long paths before transitting back to themselves.  
This allows us to combine small gadgets to bigger ones while maintaining the property of approximate NE for each individual sub-gadgets in it. 
This finishes our reduction from $\onestate$ to $\epsilon$-$\gcircuit$.
The later reductions provided in~\cite{daskalakis2009complexity,chen2006settling,rubinstein2018inapproximability} show that $\gcircuit$ is $\ppad$-hard, even for a small enough constant $\epsilon$.
This completes our proof of $\ppad$-hardness of $\onestate$.

\section{Complexity results for $\tbsg$s under $\ronly$}\label{sec:pure}

In this section, we consider specially structured instances of $\onestate$ in which each player only receives a non-zero instantaneous reward at her own (single) state.
We call this class of games $\ronly$ and formally define it below.
\begin{definition}
A $\ronly$ game $\mathcal{G} = (n,\calS = \cup_{i\in[n]}\{s_i\},\mathcal{A},\pp,\r,\gamma)$ is an instance of $\onestate$ that satisfies $r_{i,s,a} = 0$ for any $s \neq s_i$.
\end{definition}

In~\Cref{ssec:reward-only-sign}, we present a few positive results when the all non-zero instantaneous rewards have a fixed sign, and design particular algorithms based on potential games~\cite{monderer1996potential} for randomized transition probabilities and algorithms based on graph structure for deterministic transition probabilities \emph{and} action-independent rewards, respectively. In~\Cref{ssec:mixed-sign}, we show that once we lift the condition that all players receive rewards of the same sign, the existence of pure NE becomes an NP-hard problem. 

\subsection{$\ronly$ $\onestate$s with fixed-sign rewards}\label{ssec:reward-only-sign}

First, we provide polynomial-time algorithms for $\ronly$ $\onestate$s instances with rewards all of the same sign, i.r. $r_{i,s_i,a} \geq 0$ for all $i \in [n], a \in \calA_i$ \emph{or} $r_{i,s_i,a} \leq 0$ for all $i \in [n],a \in \calA_i$.
For an $\ronly$ instance, the utility function of each player simplifies to
\begin{equation}\label{eq:fixed-sign-V}
\begin{aligned}
V_i^{\ppi} = \inprod{\ee_{s_i}}{(\II-\gamma\PP^{\ppi})^{-1}\r^{\ppi}} & = \frac{1}{\det(\ppi)}\left( \sum_{j\in[n]} E_i(\ppi_{-j})r_i^{\ppi_j}(s_j)\right)\\
& = \frac{E_i(\ppi_{-i})}{\det(\ppi)}r_i^{\ppi_i}(s_i).
\end{aligned}
\end{equation}
Above, we defined as shorthand $\det(\ppi) \defeq \det(\II-\gamma\PP^{\ppi})$ and $E_i(\ppi_{-j})\defeq (-1)^{i+j}\det\left(\mathbf{M}_{ji}\right)$, where $\mathbf{M}_{ji}$ denotes the $(j,i)$ minor of matrix $\II-\gamma\PP^{\ppi}$.
Since only the $j^{th}$ row of $\PP^{\ppi}$ depends on $\ppi_j$, we get that $\mathbf{M}_{ji}$ is only a function of $\ppi_{-j}$.
Moreover, as long as $\gamma < 1$ we have $E(\ppi_{-i}),\det(\ppi) > 0$ based on properties of strictly diagonally dominant matrices~\cite{golub2013matrix}.
Consequently, the sign of $V_i^{\ppi}$ will entirely depend on the reward functions $r_i^{\ppi_i}(s_i)$ for each $i\in[n]$. We consider the two cases of non-negative and non-positive rewards below.

\subsubsection{$\ronly$ $\onestate$s with non-negative rewards}\label{ssec:non-negative}

First, we consider the case for which all players receive non-negative rewards.
In other words, we assume that $r_{i,s_i,a} \ge 0$ for all $i \in [n]$ and $a\in\calA_{i}$.
As a warm-up, we show in the following lemma that a pure NE must always exist in this case.
\begin{lemma}\label{lem:nonneg-existence}
Consider a $\ronly$ $\onestate$ instance $\mathcal{G} = (n,\calS = \cup_{i\in[n]}\{s_i\},\mathcal{A},\pp,\r,\gamma)$ for which all rewards are non-negative, i.e., $r_{i,s_i,a}\ge0$	for any $i \in [n]$ and $a\in\calA_i$.
Then, any such game must have a pure NE.
\end{lemma}
\begin{proof}
We prove the existence of pure NE by contradiction. 
Suppose instead that a pure NE does not exist.
Then, for any pure strategy $\ppi$ there exists a player $i$ who strictly wishes to deviate.
More formally, there exists an index $i \in [n]$ and some other strategy $(\ppi_i',\ppi_{-i})$ that changes only at state $s_i$ such that
\begin{align}\label{eq:ronly-relationship}
\frac{r_i^{\ppi_i}(s_i) E_i(\ppi_{-i})}{\det(\ppi)} < \frac{r_i^{\ppi'_i}(s_i)E_i(\ppi_{-i})}{\det(\ppi_i',\ppi_{-i})} 
\implies \det(\ppi) > \frac{r_i^{\ppi_i}(s_i)}{r_i^{\ppi'_i}(s_i)}\det(\ppi_i', \ppi_{-i}). 
\end{align}
By the pseudolinearity property of the utilities in~\Cref{ssec:quasi-mono}, we can take $\ppi'_i$ to be a deterministic policy, and thus $(\ppi_i', \ppi_{-i})$ is a pure strategy as well.

In summary, since we have made the contradictory assumption that there is no pure NE, for any pure $\ppi$ there is a coordinate corresponding to some player's policy $\ppi_i$ that we can change to some other pure strategy $(\ppi'_i,\ppi_{-i})$ and ensure that $\det(\ppi) > \frac{r_i^{\ppi_i}(s_i)}{r_i^{\ppi'_i}(s_i)}\det(\ppi_i', \ppi_{-i})$. 
Also note that the total number of pure strategies is finite.
Therefore, this strategy improvement procedure will lead to a cycle of length $K < \infty$, i.e.
\[\ppi^{(0)}\rightarrow\ppi^{(1)}\rightarrow \ppi^{(2)}\rightarrow\cdots\rightarrow \ppi^{(K)}\rightarrow\ppi^{(K+1)} = \ppi^{(0)}, ~~\text{for}~~K<\infty.\]
For step $k\in[K+1]$, we denote $i_k \in [n]$ as the player who improves their strategy.
Further, we denote the constant coefficient 
\[C^{(k)} = \frac{r_{i_k}^{\ppi_{i_k}}(s_{i_k})}{r_{i_k}^{\ppi'_{i_k}}(s_{i_k})}\] and apply the relationship in \eqref{eq:ronly-relationship} to $\ppi^{(0)}, \ppi^{(1)},\cdots$ to conclude that
\[
\det(\ppi^{(0)}) > C^{(0)}\det(\ppi^{(1)})\ge C^{(0)}C^{(1)}\det(\ppi^{(2)})\ge \cdots \ge \Par{\prod_{k\in[K+1]}C^{(k-1)}}\det(\ppi^{(0)}) = \det(\ppi^{0}).
\]
Note that the last equality above uses the fact that since $\ppi^{(0)}$ and $\ppi^{(K+1)}$ have the same policy for each player, $\prod_{k\in[K+1]}C^{(k)} = 1$ as the policy shifts will cancel in the cycle. 
This ultimately leads to $\det(\ppi^{(0)}) > \det(\ppi^{(0)})$, which is not possible.
This completes the proof by contradiction and proves the claim that a pure NE must exist.
\end{proof}
Lemma~\ref{lem:nonneg-existence} hints at a \emph{potential game} structure for the utilities in the case of $\ronly$ $\onestate$ with non-negative rewards.
This suggests that \emph{best-response-dynamics}~\cite{monderer1996potential} and their variants would successfully converge to a pure NE.
Our first Algorithm~\ref{alg:ronly-non-negative-AltMin} is a variant of approximate-best-response dynamics, and we show that it converges to an $\eps$-approximate pure NE in polynomial-time.
Our analysis of this algorithm makes the potential structure in $\ronly$ $\onestate$ explicit.

\begin{algorithm2e}[h]
	\caption{Strategy iteration for non-negative $\ronly$ $\onestate$}
	\label{alg:ronly-non-negative-AltMin}
	\DontPrintSemicolon
	\codeInput $\ronly$ $\onestate$ instance $\calG  = (n,\calS = \cup_{i\in[n]}\{s_i\},\mathcal{A},\pp,\r,\gamma)$ satisfying $r_{i,s,a}\ge 0$, accuracy $\epsilon$\;
	Preprocess $\tilde{r}_{i,s_i,a} = \max\Par{r_{i,s_i,a}, \frac{(1-\gamma)\eps}{2}}$, for each $i\in[n]$, $a\in\calA_i$\;
	Define game $\calG'= (n,\calS = \cup_{i\in[n]}\{s_i\},\mathcal{A},\pp,\tilde{r},\gamma)$\;
	Initialize two pure strategies $ \ppi\neq\ppi_{+}$\;
	\While{$\ppi\neq \ppi_{+}$}{\label{line:ronly-while}
	$\ppi\gets\ppi_{+}$\;
	Compute $\hat{\upsilon}_i(\ppi) = \hat{V}_i^{\ppi}(s_i)$ for all player $i\in[n]$ for game $\calG'$\;\label{line:ronly-V}
	\For{$i=1$ {\bfseries{\textup{to}}} $n$, $a_i\in\calA_i$}{
	Compute $\hat{\upsilon}_i(\ee_{a_i},\ppi_{-i}) = \hat{V}_i^{(\ee_{a_i},\ppi_{-i})}(s_i)$ for game $\calG'$ using reward $\tilde{r}$\;\label{line:ronly-V-new}
	\If{$\hat{\upsilon}_i(\ee_{a_i},\ppi_{-i})\ge \hat{\upsilon}_i(\ppi)+\frac{\epsilon}{2}$}{
	$\ppi_{+} \gets (\ee_{a_i},\ppi_{-i})$ and \textbf{break}\;
	}
	}
	}
	\codeReturn $\ppi$ 
\end{algorithm2e}	

\begin{proposition}\label{lem:pureinP_potential}
Consider a $\ronly$ $\onestate$ instance $\calG  = (n,\calS = \cup_{i\in[n]}\{s_i\},\mathcal{A},\pp,\r,\gamma)$ where all rewards are non-negative, i.e., $r_{i,s_i,a}\ge0$ for any $i\in[n]$, $a\in\calA_i$.
Then, given some desired accuracy $\epsilon$,~\Cref{alg:ronly-non-negative-AltMin} finds an $\epsilon$-approximate pure NE in time \[O\Par{\frac{n^4\Atot}{(1-\gamma)\epsilon}\log \Par{\frac{1}{(1-\gamma)^2\epsilon}}}\]

\end{proposition}

\begin{proof}

\textbf{Correctness.}
We use $\upsilon_i(\ppi)$ and $\hat{\upsilon}_i(\ppi)$ to denote the value function of game $\calG$ and $\calG'$ (where everything remains unchanged except that we pad the instantaneous rewards to be $\tilde{r}_{i,s_i,a} = \max\Par{r_{i,s_i,a}, \frac{(1 - \gamma)\eps}{2}}$) for player $i$ at the state it controls. We show that~\Cref{alg:ronly-non-negative-AltMin} finds an $\eps/2$-approximate pure NE of $\calG'$ under the utility function defined by $\{\upsilon_i\}_{i\in[n]}$. To see this, note that when the algorithm terminates with strategy $\ppi$, we must have $\hat{\upsilon}_i(\ee_{a_i},\ppi_{-i})\le \hat{\upsilon}_i(\ppi)+\frac{\epsilon}{2}$ for all $i\in[n]$ and $a_i\in\calA_i$.
As a consequence of this and the pseudolinearity of utility functions in $\onestate$ (\Cref{coro:quasi-mono-ssg}), we have
\begin{equation}\label{eq:single-reward-positive-eq-1}
	\hat{\upsilon}_i(\ppi_i',\ppi_{-i})\le \hat{\upsilon}_i(\ppi)+\frac{\epsilon}{2},
\end{equation}
for any $i \in [n]$ and any alternative policy $\ppi'_i$.
Consequently, the output of the algorithm, $\ppi$, is an $\eps/2$-approximate pure NE for game $\calG'$. 
Next, we need to show that $\ppi$ is also an $\eps$-approximate pure NE for the original game $\calG$.
By the expression of value functions in~\eqref{eq:fixed-sign-V}, we observe that $\norm{\left(\II-\gamma\PP^{\ppi}\right)^{-1}}_\infty\le \sum_{t\ge 0}\gamma^t\norm{\left(\PP^{\ppi}\right)^t}_\infty\le \frac{1}{1-\gamma}$ and all entries of $\left(\II-\gamma\PP^{\ppi}\right)^{-1}$ are non-negative. Consequently, noting that $|\tilde{r}_{i,s_i,a} - r_{i,s_i,a}| \leq \frac{(1 - \gamma)\eps}{2}$ gives us
\[\upsilon_i(\ppi)\le \hat{\upsilon}_i(\ppi)\le \upsilon_i(\ppi)+\frac{1}{1-\gamma}\frac{(1-\gamma)\eps}{2}=\upsilon_i(\ppi)+\frac{\eps}{2}~~\text{for any}~i\in[n], \ppi.\]
Plugging this back into~\eqref{eq:single-reward-positive-eq-1},  we have
\[
\upsilon_i(\ppi_i',\ppi_{-i})\le \upsilon_i(\ppi)+\epsilon
\]
for any $i \in [n]$ and alternative policy $\ppi'_i$ for player $i$.
This implies that $\ppi$ is also an $\eps$-approximate pure NE for the original game $\calG$, proving the correctness of the algorithm.

\textbf{Runtime.} 
We first show that~\Cref{alg:ronly-non-negative-AltMin} terminates in polynomial number of steps. To see this, we define the logarithmic reward function of each player to be 
\[\tilde{V}_i^{\ppi} \defeq \log \hat{\upsilon}_i(\ppi) = \log \frac{E_i(\ppi_{-i})\tilde{r}_i^{\ppi_i}(s_i)}{\det(\ppi)},\]
and a universal potential function to be 
\[\Phi(\ppi) \defeq \log\left(\frac{\prod_{i\in[n]}\tilde{r}_i^{\ppi_i}(s_i)}{\det(\ppi)}\right).\]
The key observation is that this potential function satisfies the following property: For any $i\in[n]$, and strategy $\ppi$ and alternative policy $\ppi_i'$ for player $i$, we have
\begin{align*}
\Phi(\ppi_i',\ppi_{-i})-\Phi(\ppi) & = \log\left(\frac{\tilde{r}_i^{\ppi'_i}(s_i)}{\det(\ppi'_i,\ppi_{-i})}\right)-\log \left(\frac{\tilde{r}_i^{\ppi_i}(s_i)}{\det(\ppi)}\right)\\
& = \log\left(\frac{E_i(\ppi_{-i})\tilde{r}_i^{\ppi'_i}(s_i)}{\det(\ppi'_i,\ppi_{-i})}\right)-\log \left(\frac{E_i(\ppi_{-i})\tilde{r}_i^{\ppi_i}(s_i)}{\det(\ppi)}\right)\\
& = \tilde{V}_i^{(\ppi'_i,\ppi_{-i})}-\tilde{V}_i^{\ppi}.
\end{align*}

By the properties of strictly diagonally dominant matrices $\II-\gamma \PP^{\ppi}$, we know all its eigenvalues are bounded, i.e.\  $\lambda_i\in[1-\gamma,1+\gamma]$, for each $i\in[n]$. Thus, we have $\det(\ppi) = \prod_{i\in[n]}\lambda_i\in[(1-\gamma)^n, (1+\gamma)^n]$.
Moreover, noting that $\tilde{r}_{i,s_i,a} \geq \frac{(1 - \gamma) \eps}{2}$ for all $i \in [n], a \in \calA_i$ implies that
\begin{equation}\label{eq:ronly-phi-bound-1-pos}
	n\log\Par{\frac{(1-\gamma)\eps}{2(1+\gamma)}}\le\Phi(\ppi)\le n\log\Par{\frac{1}{1-\gamma}}.
\end{equation}
Now, whenever we restart on~\Cref{line:ronly-while}, it means that we have found a player $i\in[n]$ and a new strategy $(\ppi_i',\ppi)$ such that  $\hat{\upsilon}_i(\ppi_i',\ppi)\ge \hat{\upsilon}_i(\ppi)+\frac{\epsilon}{2}$. Since $\hat{\upsilon}_i(\ppi)\le \frac{1}{1-\gamma}$ for any $i,\ppi$, we have $\hat{\upsilon}_i(\ppi_i',\ppi_{-i})\ge \hat{\upsilon}_i(\ppi)+\frac{(1-\gamma)\epsilon}{2}\hat{\upsilon}_i(\ppi)$.
As a consequence, every ``restart'' must increase the potential by at least 
\begin{equation}\label{eq:ronly-phi-bound-2-pos}\Delta \Phi\defeq \log \Par{\frac{\hat{\upsilon}_i(\ppi_i',\ppi_{-i})}{ \hat{\upsilon}_i(\ppi)}}\ge \log\Par{1+\frac{(1-\gamma)\epsilon}{2}}.	
\end{equation}

Combining~\eqref{eq:ronly-phi-bound-1-pos} and~\eqref{eq:ronly-phi-bound-2-pos} yields that the while loop of~\Cref{line:ronly-while} restarts at most 
\[
O\Par{\frac{n}{(1-\gamma)\epsilon}\log \Par{\frac{1}{(1-\gamma)^2\epsilon}}}
\]
iterations.
Next, we characterize the per-iteration runtime.
Within each iteration of the while loop, we need to first compute all $\hat{\upsilon}_i(\ppi)$ for each $i\in[n]$ under the given strategy $\ppi$ in~\Cref{line:ronly-V}, which takes $O(\sum_{i\in[n]}|\calA_i|+n^3)$. 
We also need to compute each value function when switching the policy at one single state for each state and action pair in~\Cref{line:ronly-V-new}, which takes $O(n^3(\sum_{i\in[n]}|\calA_i|))$ in total. 
Summing these together and multiplying by the total number of iterations yields overall runtime
\[
O\Par{\frac{n^4\Atot}{(1-\gamma)\epsilon}\log \Par{\frac{1}{(1-\gamma)^2\epsilon}}},
\]
which is polynomial in problem parameters $\Atot$, $1/(1-\gamma)$ and $1/\epsilon$.
\end{proof}
Algorithm~\ref{alg:ronly-non-negative-AltMin} is an appealing algorithm to compute $\eps$-approximate NE for $\ronly$ $\onestate$ instances with non-negative rewards, and its proof uncovers a Markov potential game structure.
In~\Cref{sssec:nonpositive} we show that a similarly defined algorithm (with negative clippings for the rewards) also succeeds in computing $\eps$-approximate NE for $\ronly$ $\onestate$ instances with non-positive rewards.
Algorithm~\ref{alg:ronly-non-negative-AltMin} is a centralized algorithm as all players do not update their strategies in the same iteration.
However, based on the identified potential game structure we believe that there is also potential to study the convergence of independent policy gradient methods~\cite{fox2021independent} for $\ronly$ $\onestate$.

The essence of the potential function structure simply lies in the fact that the rewards attained by each player at their own state are all of the same sign.
While at a high level this provides a qualitative incentive-symmetry, the quantitative incentives of each player can be quite different.
In particular each player wants to maximize some combination of her instantaneous reward at her own state and her own state's visitation probability; but the relative combination of these factors can completely vary across players.
We next further specialize the non-negative $\ronly$ $\onestate$ problem to completely symmetrize the players' incentives, and show that \emph{exact} NE can be computed efficiently under this further simplification.
\begin{definition}\label{def:ronly-further}
	We say an $\ronly$ $\onestate$ has \emph{deterministic transitions} if  all transitions are deterministic, i.e. $\pp_{s,a}(s') \in \{0,1\}$ for any $s\in\calS, a\in\calA_s$, $s'\in\calS$. 
	Further, we say it has \emph{action-independent rewards} if the instantaneous  rewards only depend on state, i.e. $r_{i,s_i,a} = r_{i,s_i,a'}$ for any $a \neq a'$.
\end{definition}
The independence of instantaneous rewards on action implies that every player $i$ simply wants to visit her own state as frequently as possible to maximize her utility.
Further, the deterministic nature of the transitions allows us to make a natural mapping to graph problems.
In particular, for every $\ronly$ $\onestate$ instance with the further structure of~\Cref{def:ronly-further} we can define a corresponding directed graph with vertices representing states and directed edges representing deterministic transitions between states.
Through this graph interpretation, we see that maximizing a player's own-state visitation probability corresponds to finding short cycles for the vertex that player controls.
This graph-theoretic interpretation yields the following polynomial-time algorithm for finding \emph{exact} pure NEs, described in~\Cref{alg:ronly-non-negative-graph}. 

\begin{algorithm2e}[h]
	\caption{Shortest cycles for non-negative $\ronly$ $\onestate$, deterministic transitions, action-independent rewards}
	\label{alg:ronly-non-negative-graph}
	\DontPrintSemicolon
	\codeInput $\ronly$ $\onestate$ $\calG  = (n,\calS = \cup_{i\in[n]}\{s_i\},\mathcal{A},\pp,\r,\gamma)$ with deterministic transitions, action independent rewards so that $r_{i,s,a}\ge 0$\;
	Construct a graph $G=(V,E)$: $V = \calS$, each action $a\in\calA_i$ corresponds to a distinct edge $(s_i,s_j)\in E$ when $\pp_{s_i,a}(s_j)=1$. \;
	\While{there exist cycles in $G(E)$}{
	Find shortest cycle on $G$, denote its vertex set as $V'$\tcp*{break ties arbitrarily}
	Update $\ppi$ according to the cycle\;
	Update set $E$ by removing $(s_i, s_j)\in E$ for all $s_i\in V'$, $s_j\in V$\;
	}
	Pick arbitrary action for the remaining vertices with edges on $G(E)$ and update $\ppi$\;\label{line:ronly-graph-remain}
	\codeReturn $\ppi$ 
\end{algorithm2e}	

\begin{proposition}
	Given a $\ronly$ $\onestate$ instance $\calG  = (n,\calS = \cup_{i\in[n]}\{s_i\},\mathcal{A},\pp,\r,\gamma)$ where all rewards are non-negative, all transitions are deterministic and rewards are action-independent (see~\Cref{def:ronly-further}),~\Cref{alg:ronly-non-negative-graph} finds an exact pure NE in time $O(n\Atot^2\log n)$
\end{proposition}

\begin{proof}
\textbf{Correctness.} We first show that the players whose corresponding actions are defined only in~\Cref{line:ronly-graph-remain} cannot switch to another action to generate strictly better utility.
Note that, given the policies chosen by the other players, no cycles pass through the corresponding vertices for these players in the induced graph.
Consequently, any action yields the same reward for these players and the chosen actions are optimal.

Next, we consider any player $i$ whose corresponding action $\ppi_i$ is chosen by picking a cycle at some step.
Let $\calS_{<i}$ denote the set of all states (equivalently players) that have chosen an action \emph{before} player $i$. 
If player $i$ takes any action that transits to any $s \in \calS_{<i}$, she will never visit herself again, simply because player $s$ has already found a shortest cycle that repeatedly transits back to itself.
Therefore, any such action are clearly suboptimal for player $i$.
Moreover, by the definition of a shortest cycle, player $i$ also cannot generate a strictly better utility by picking an action that transits to some other state than the current shortest cycle chosen. 
Putting these cases together, player $i$'s current action, which keeps her on the shortest-cycle, is optimal given all other players' actions.
Consequently, the strategy $\ppi$ is a pure NE.

\textbf{Runtime}. We know that each iteration at least removes one state from $V$, so the total while loop must terminate within $n$ iterations. Also, within each iteration, we compute a shortest cycle which takes time $O(|\calA|(|\calA|+n)\log n)$~\cite{dijkstra1959note}. Thus, the total runtime is bounded by $O(n\Atot^2\log n)$.
\end{proof}

\subsubsection{$\ronly$ $\onestate$s with non-positive rewards}\label{sssec:nonpositive}

In this section, we consider the alternative case of $\ronly$ for which all players receive non-positive rewards.
In other words, we assume that $r_{i,s_i,a} \le 0$ for all $i \in [n]$ and $a\in\calA_i$.
All proof arguments follow very similarly as in~the non-negative-reward case (except for a slightly different exact NE algorithm for $\ronly$ $\onestate$s when transitions are deterministic and rewards are action-independent, see~\Cref{alg:ronly-non-positive-graph}). 
We provide the results here mainly for completeness. 
We first show that, similar to the case of non-negative rewards, a pure NE must always exist.

\begin{lemma}\label{lem:nonpos-existence}
Consider a $\ronly$ $\onestate$ instance $\mathcal{G} = (n,\calS = \cup_{i\in[n]}\{s_i\},\mathcal{A},\pp,\r,\gamma)$ for which all rewards-at-own-state are non-positive, i.e., $r_{i,s_i,a}\le0$	for any $i \in [n]$ and $a\in\calA_i$.
Then, any such game must have a pure NE.
\end{lemma}
\begin{proof}
Similar to the corresponding proof of~\Cref{lem:nonneg-existence}, we prove the existence of pure NE by contradiction.
Suppose instead that a pure NE does not exist.
Then, for any pure strategy $\ppi$ there exists a player $i$ who strictly wishes to deviate.
More formally, there exists an index $i \in [n]$ and some other strategy $(\ppi_i', \ppi_{-i})$ that changes at one state $s_i$ such that 
\begin{align}\label{eq:ronly-relationship-negative}
0\ge \frac{r_i^{\ppi'_i}(s_i)}{\det(\ppi_i', \ppi_{-i})}> \frac{r_i^{\ppi_i}(s_i)}{\det(\ppi)} ~~\implies~~ \det(\ppi_i', \ppi_{-i})>\frac{r_i^{\ppi'_i}(s_i)}{r_i^{\ppi_i}(s_i)}\det(\ppi).
\end{align}
Again, by the pseudo-linearity property of the utilities in~\Cref{ssec:quasi-mono} we can take $\ppi'_i$ to be a deterministic policy, and thus $(\ppi'_i,\ppi_{-i})$ is a pure strategy as well.

In summary, since we have made the contradictory assumption that there is no pure NE, for any $\ppi$ there is a coordinate corresponding to some player's policy $\ppi_i$ that we can change to some other pure strategy and ensure that $\det(\ppi_i', \ppi_{-i})>\frac{r_i^{\ppi'_i}(s_i)}{r_i^{\ppi_i}(s_i)}\det(\ppi)$. 
Also note that the total number of pure strategies is finite.
Therefore, this strategy improvement procedure will lead to a cycle of length $K < \infty$,~i.e.
\[\ppi^{(0)}\rightarrow\ppi^{(1)}\rightarrow \ppi^{(2)}\rightarrow\cdots\rightarrow \ppi^{(K)}\rightarrow\ppi^{(K+1)}=\ppi^{(0)}, ~~\text{for}~~K<\infty.\]
Similar to the non-negative case, for every step $k \in [K+1]$ we denote $i_k \in [n]$ as the player who improves their strategy.
Further, we denote the contsant coefficient
\[C^{(k)} = \frac{r_{i_k}^{\ppi'_{i_k}}(s_{i_k})}{r_{i_k}^{\ppi_{i_k}}(s_{i_k})}\]
 and apply the relationship in \eqref{eq:ronly-relationship-negative} to $\ppi^{(0)}, \ppi^{(1)},\cdots$, we can conclude that
\[
\det(\ppi^{(0)}) > C^{(0)}\det(\ppi^{(1)})\ge  \cdots \ge \Par{\prod_{k\in[K+1]}C^{(k-1)}}\det(\ppi^{(0)}) = \det(\ppi^{0}).
\] 
Note that the last equality above uses the fact that since $\ppi^{(0)}$ and $\ppi^{(K+1)}$ have the same policy for each player, $\prod_{k \in [K+1]} C^{(k)} = 1$ as the policy shifts will cancel in the cycle.
This ultimately leads to $\det(\ppi^{(0)}) > \det(\ppi^{(0)})$, which is not possible.
This completes the proof by contradiction and proves the claim that a pure NE must exist.
\end{proof}

Now we provide two different algorithms for finding pure NEs of non-positive $\ronly$ $\onestate$s. First, we can adapt a variant of~\Cref{alg:ronly-non-negative-AltMin}, i.e.,~\Cref{alg:ronly-non-positive-AltMin} to all such TBSG instances to find approximate pure NEs.
This proof shows similar potential structure with the negative logarithm of utilities considered instead (owing to the non-positivity of rewards). 
The only change in the algorithm is a slightly different way to truncate the rewards $r$ in~\Cref{line:ronly-neg-truncate}, also based on the fact that they are non-positive.

\begin{algorithm2e}[h]
	\caption{Strategy iteration for non-positive $\ronly$ $\onestate$}
	\label{alg:ronly-non-positive-AltMin} 
	\DontPrintSemicolon
	\codeInput $\ronly$ $\onestate$ instance $\mathcal{G} = (n,\calS = \cup_{i\in[n]}\{s_i\},\mathcal{A},\pp,\r,\gamma)$ so that $r_i(s,a)\le 0$, accuracy $\epsilon$\;
	Preprocess $\tilde{r}_{i,s_i,a} = \min\Par{r_{i,s_i,a}, -\frac{(1-\gamma)\eps}{2}}$, for each $i\in[n]$, $a\in\calA_i$\;\label{line:ronly-neg-truncate}
	Define game $\calG'=(n,\calS = \cup_{i\in[n]}\{s_i\},\mathcal{A},\pp,\tilde{\r},\gamma)$\;
	Initialize two pure strategies $ \ppi\neq\ppi_{+}$\;
	\While{$\ppi\neq \ppi_{+}$}{
	$\ppi_{+}\gets\ppi$\;
	Compute $\hat{\upsilon}_i(\ppi)=\hat{V}_i^{\ppi}(s_i)$ for all player $i\in[n]$ for game $\calG'$\;
	\For{$i=1$ {\bfseries{\textup{to}}} $n$, $a_i\in\calA_i$}{
	Compute $\hat{\upsilon}_i(\ee_{a_i},\ppi_{-i})=\hat{V}_i^{(\ee_{a_i},\ppi_{-i})}(s_i)$ for game $\calG'$ using reward $\tilde{r}$\;
	\If{$\hat{\upsilon}_i(\ee_{a_i},\ppi_{-i})\ge \hat{\upsilon}_i(\ppi)+\frac{\epsilon}{2}$}{
	$\ppi_{+} \gets (\ee_{a_i},\ppi_{-i})$ and \textbf{break}\;
	}
	}
	}
	\codeReturn $\ppi$ 
\end{algorithm2e}

\begin{proposition}\label{lem:pureinP_potential-nonpos}
Consider a $\ronly$ $\onestate$  instance $\calG = (n,\calS = \cup_{i\in[n]}\{s_i\},\mathcal{A},\pp,\tilde{\r},\gamma)$ where all rewards are non-positive, i.e., $ r_{i,s_i,a}\le0$	for any $i\in[n]$ , $a\in\calA_i$.
Then, given some desired accuracy $\epsilon$,~\Cref{alg:ronly-non-positive-AltMin} finds an $\epsilon$-approximate pure NE in time \[O\Par{\frac{n^4\Atot}{(1-\gamma)\epsilon}\log \Par{\frac{1}{(1-\gamma)^2\epsilon}}}.\]
\end{proposition}
\begin{proof}

\textbf{Correctness.}
We use $\upsilon_i(\ppi)$ and $\hat{\upsilon}_i(\ppi)$ to denote the value function of $\calG$ and $\calG'$ respective for player $i$ at state $s_i$. 
Following an identical argument to the non-negative case,~\Cref{alg:ronly-non-positive-AltMin} finds an $\eps/2$-approximate pure NE of $\calG'$ satisfying
\begin{equation}\label{eq:single-reward-negative-eq-1}
	\hat{\upsilon}_i(\ppi_i',\ppi_{-i})\le \hat{\upsilon}_i(\ppi)+\frac{\epsilon}{2}
\end{equation}
for any $i \in [n]$ and any alternative policy $\ppi'_i$.
By the definition of the padded rewards $\tilde{r}$ and the expression of value functions in~\eqref{eq:fixed-sign-V}, we then get
\[\upsilon_i(\ppi)-\frac{\eps}{2}=\upsilon_i(\ppi)-\frac{1}{1-\gamma}\frac{(1-\gamma)\eps}{2}\le \hat{\upsilon}_i(\ppi)\le \upsilon_i(\ppi)~~\text{for any}~i\in[n], \ppi.\]
Plugging this back into~\eqref{eq:single-reward-negative-eq-1},  we have 
\[
\upsilon_i(\ppi_i',\ppi_{-i})\le \upsilon_i(\ppi)+\epsilon,~~\text{for any}~i\in[n], \ppi'_{i},
\]
for any $i \in [n]$ and any alternative policy $\ppi'_i$ for player $i$.
This implies that $\ppi$ is also an $\eps$-approximate pure NE for the original game $\calG$, proving the correctness of the algorithm.

\textbf{Runtime.} Similar to the non-negative case, we first show that~\Cref{alg:ronly-non-positive-AltMin} terminates in polynomial number of steps. To see this, we define the logarithmic reward function of each player to be 
\[\tilde{V}_i^{\ppi} \defeq \log \hat{\upsilon}_i(\ppi) = \log \frac{E_i(\ppi_{-i})(-\tilde{r}_i^{\ppi_i}(s_i))}{\det(\ppi)},\]
and a universal ``negative'' version of potential function to be 
\[\Phi(\ppi) \defeq \log\left((-1)^n\frac{\prod_{i\in[n]}\tilde{r}_i^{\ppi_i}(s_i)}{\det(\ppi)}\right).\]

The key observation remains that this potential function satisfies the following property: For any $i\in[n]$, and strategy $\ppi$, $\ppi_i'$,
\begin{align*}
\Phi(\ppi_i',\ppi_{-i})-\Phi(\ppi) & = \log\left(\frac{\tilde{r}_i^{\ppi'_i}(s_i)}{\det(\ppi'_i,\ppi_{-i})}\right)-\log \left(\frac{\tilde{r}_i^{\ppi_i}(s_i)}{\det(\ppi)}\right)\\
& = \log\left(\frac{E_i(\ppi_{-i})\tilde{r}_i^{\ppi'_i}(s_i)}{\det(\ppi'_i,\ppi_{-i})}\right)-\log \left(\frac{E_i(\ppi_{-i})\tilde{r}_i^{\ppi_i}(s_i)}{\det(\ppi)}\right)\\
& = \tilde{V}_i^{(\ppi'_i,\ppi_{-i})}-\tilde{V}_i^{\ppi}.
\end{align*}
Using again a similar argument to the non-negative and the new padding definition for non-positive rewards provided in~\Cref{line:ronly-neg-truncate}, we get
\begin{equation}\label{eq:ronly-phi-bound-1}
	n\log\Par{\frac{(1-\gamma)\eps}{2(1+\gamma)}}\le\Phi(\ppi)\le n\log\Par{\frac{1}{1-\gamma}}.
\end{equation}
Everytime we restart the while loop we also must increase the potential by at least 
\begin{equation}\label{eq:ronly-phi-bound-2}\Delta \Phi\defeq \log \Par{\frac{\hat{\upsilon}_i(\ppi_i',\ppi_{-i})}{ \hat{\upsilon}_i(\ppi)}}\ge \log\Par{1+\frac{(1-\gamma)\epsilon}{2}}.	
\end{equation}
Combining~\eqref{eq:ronly-phi-bound-1} and~\eqref{eq:ronly-phi-bound-2} we know the while loop of~\Cref{line:ronly-while} restarts for iterations at most 
\[
O\Par{\frac{n}{(1-\gamma)\epsilon}\log \Par{\frac{1}{(1-\gamma)^2\epsilon}}}.
\]
Thus the total runtime remains unchanged in comparison with the non-negative $\ronly$ case and is bounded by
\[
O\Par{\frac{n^4\Atot}{(1-\gamma)\epsilon}\log \Par{\frac{1}{(1-\gamma)^2\epsilon}}},
\]
which is polynomial in problem parameters $\Atot$, $1/(1-\gamma)$ and $1/\epsilon$.
\end{proof}
We now turn to the problem of exact NE computation for $\ronly$ $\onestate$ with non-positive rewards and the corresponding special case defined in~\Cref{def:ronly-further},~i.e. deterministic transitions and action-independent instantaneous rewards.
Since the rewards are non-positive, every player $i$ now wants to visit her own state as \emph{infrequently} as possible to maximize their utility, which corresponds to finding paths or longest cycles for each vertex in the graph-theoretic interpretation of the problem. 
We provide the following algorithm for finding exact pure NEs of this  of this deterministic transition and action-independent reward setting in~\Cref{alg:ronly-non-positive-graph}. 

\begin{algorithm2e}[h]
	\caption{Longest cycles for non-positive $\ronly$ $\onestate$, deterministic transitions, action-independent rewards}
	\label{alg:ronly-non-positive-graph}
	\DontPrintSemicolon
	\codeInput $\ronly$ $\onestate$ TBSG $\calG  = (\{0\}\cup[n],\calS = \calS_0\cup(\cup_{i\in[n]}\mathcal{S}_i),\mathcal{A}=\calA_0\cup(\cup_{i\in[n]}\mathcal{A}_i),\PP,\RR,\gamma)$ with deterministic transitions, reward bound $M$ so that $r_i(s,a)\le 0$\;
Randomly initialize some action $\ppi$\;
	\Repeat{$\ppi_{+}= \ppi$}{
	$\ppi_{+}\gets \ppi$\;
	\For{$s\in \calS$}{
	\If{can switch action to a longer cycle or acyclic path for $s$}{
	Update $\ppi_{+}$ and \textbf{break}\;
	}
	}
	}
	\codeReturn $\ppi$ 
\end{algorithm2e}	

\begin{proposition}
	Given a $\ronly$ $\onestate$ instance $\calG  = (n,\calS = \cup_{i\in[n]}\{s_i\},\mathcal{A},\pp,\r,\gamma)$ where all rewards are non-positive, all transitions are deterministic and rewards are action-independent (see~\Cref{def:ronly-further}), ~\Cref{alg:ronly-non-positive-graph} finds an exact pure NE in time $O(n^2\Atot)$.
\end{proposition}

\begin{proof}
\textbf{Correctness.} We first show that when the algorithm converges, i.e.,  $\ppi_{+}=\ppi$, the corresponding strategy $\ppi$ must be a pure NE. 
We prove this by contradiction. 
Suppose, instead, that $\ppi$ is not a pure NE.
Then there must exist some state $s_i$ of player $i$ such that one can switch from current action $a$ to some $a'\in\calA_i$ to lead to a longer cycle.
This contradicts with the fact that $\ppi_{+}=\ppi$. 
Consequently, we have $\ppi$ is a pure NE for the given instance.

\textbf{Runtime}.  We first note the random initialization step is $O(n)$. Now within each step, enumerating over all vertices $s\in C$ and their actions takes at most $O(\Atot)$ complexity. 
Note if we have updated $\ppi_{+}$ from $\ppi$, then either the total number of cycles has gone down by $1$, or the total number of vertices in all cycles has gone up by at least $1$. 
This then implies $\ppi_{+}$ could only change at most $n^2$ times, and the algorithm must terminate within $n^2$ iterations. 
Putting it all together, this shows the algorithm's runtime is $O(n^2\Atot)$.

\end{proof}

\subsection{$\ronly$ $\onestate$s with mixed-sign rewards}\label{ssec:mixed-sign}

The preceding sections have shown that it is possible to compute either approximate or exact NE of $\ronly$ $\onestate$ in polynomial-time, provided that the rewards received by players at their respective controlling states are of the same sign.
We saw in~\Cref{ssec:ppad-hard} that if we remove the $\ronly$ assumption, the problem of approximate NE computation becomes $\ppad$-hard.
In this section, we show that the further assumption of requiring the rewards to be of the same sign is essential to our positive results and the existence of pure NE.
In particular, we consider instances of $\ronly$ $\onestate$ in which each player receives a reward that can be either positive or negative only at the single state she controls.
In other words, this is $\ronly$ $\onestate$ without the requirement of rewards being of the same sign.
We show that even when the transitions are deterministic and rewards are action-independent (\Cref{def:ronly-further}), we can construct an instance $\calG  = (n,\calS = \cup_{i\in[n]}\{s_i\},\mathcal{A},\pp,\r,\gamma)$ such that deciding whether a pure NE exists or not is NP-hard.
This in turn implies that computing a pure NE for $\ronly$ $\onestate$ games is NP-hard.
This is in contrast to the case of zero-sum $\tbsg$ for which a pure NE always exists and is computable in polynomial time~\cite{shapley1953stochastic,hansen2013strategy}. 

The main idea is to reduce finding a Hamiltonian path of any directed graph to finding a pure NE of a correspondingly constructed $\ronly$ $\onestate$ instance. To show why the problem of finding Hamiltonian cycle is closely related to $\onestate$, we first provide a warm-up argument that proves that the \emph{decision problem} of pure NE for $\onestate$ is NP-hard. 
We formally define a pure-NE decision problem for $\onestate$ below.
\begin{definition}\label{def:decisionproblemotbsg}
For any $\onestate$ instance $\calG$ and a value vector $\mathbf{v} \in \mathbb{R}^n$, we let the pure NE decision problem be to determine if there exists pure NE $\ppi$ such that $\upsilon_i(\ppi) \ge v_i$ for all $i \in [n]$.
\end{definition}
We now show that this decision-problem is NP-hard, using the proper choice of value vector $\mathbf{v}$.

\begin{proposition}\label{thm:np-hard-decision}
The pure NE decision problem defined in~\Cref{def:decisionproblemotbsg} is NP-hard in  $\ronly$ $\onestate$s.
\end{proposition}
\begin{proof}
We show the reduction of existence of a Hamiltonian cycle, a well-known NP-hard problem, to this decision problem. 
Given any directed graph $G=(V,E)$, we construct a corresponding instance of $\onestate$ instance $\calG_{G} = (n,\calS,\mathcal{A},\pp,\r,\gamma)$ with the following specifications.

\begin{itemize}
\item Number of players: We have $n=|V|$ since the instance is $\onestate$, and use $i\in V$ to denote a player.
\item State space: Each vertex corresponds to a distinct state. We use $s_i\in \calS$ to denote the state corresponding to vertex $i\in V$.	
\item Reward: Independent of the actions taken, each player $i$ receives negative instantaneous reward $r_{i,s_i,a} = -1<0$ for any $a\in\calA_i$ when she is taking actions at her own state $s_i$, and receives $0$ at all other states regardless of which action is taken.
\item  Action space and deterministic transitions: We let $a(s)$ denote the action that with probability $1$ transits to state $s\in\calS$. For each state $s_i\in\calS$, we specify their action set as $\calA_i = \cup_{j:(i,j)\in E(G)}\{a(s_j)\}$.
\item Utility function: For simplicity and given the equivalences of NE at different states as given in~\Cref{ssec:bellman}, we define the utility functions for each player defined as the value function of the state $s_i$ a player $i\in V$ controls, i.e., $\upsilon_i(\ppi) = V_i^{\ppi}(s_i)$.
\end{itemize}

We now claim that $G$ has a Hamiltonian cycle if and only if $\calG_G$ has an exact NE strategy $\ppi$ such that 
\begin{equation}\label{eq:NPhard-decision-value}
	\upsilon_i(\ppi)\ge -\frac{1}{1-\gamma^{|V|}}, ~~\text{for each}~i\in V.
\end{equation}
This claim, if true, implies the NP-hardness of the decision problem in~\Cref{def:decisionproblemotbsg} for $\mathbf{v} = -\frac{1}{1-\gamma^{|V|}} \mathbf{1}$.
We now prove the claim.
First, we show that if $G$ has a Hamiltonian cycle, we can construct a pure NE that satisfies Equation~\eqref{eq:NPhard-decision-value}.
We can construct a deterministic strategy $\ppi$ based on the Hamiltonian cycle so that $\ppi_{s_i} = \ee_{a(s_j)}$ if and only if $(i,j)$ is in the cycle. 
Note that by the definition of a Hamiltonian cycle, the utility of all players when they play $\ppi$ is equal to $-\frac{1}{1 - \gamma^{|V|}}$ which satisfies Equation~\eqref{eq:NPhard-decision-value}.
It is clear that this strategy is a NE because for any state $s_i$, switching to other action would only lead to shorter cycle starting from $s_i$, and consequently worse utility $\upsilon_i(\ppi_i', \ppi_{-i})\le -\frac{1}{1-\gamma^{|V|-1}} < -\frac{1}{1 - \gamma^{|V|}}$. 

On the other hand, suppose there exists a pure NE strategy $\ppi$ such that $\upsilon_i(\ppi)\ge -\frac{1}{1-\gamma^{|V|}}$, for all $i\in V$, we show there must be a Hamiltonian cycle in the graph. 
If this were not true, the pure strategy would induce some shorter cycle with length $\le |V|-1$ for some player $i$ on the directed graph. 
This implies that the utility for such player $i$ is $\upsilon_i(\ppi)\le -\frac{1}{1-\gamma^{|V|-1}}$, which leads to the the desired contradiction.
 
Combining two cases proves the desired claim, and thus the NP-hardness follows from the fact that finding a Hamiltonian cycle of any directed graph is NP-hard~\cite{garey1979computers}.
\end{proof}
We now provide a more complex graph-theoretic construction in order to prove the NP-hardness of determining the existence of pure NEs for the class of mixed-sign $\ronly$ $\onestate$s. 
Given a directed graph $G = (V,E)$, we let $L=|V|$ and describe the construction of such instance $\calG_{G} = (n,\calS,\mathcal{A},\PP,\RR,\gamma)$ with the following specifications.
(For this construction, we index by state $s$ rather than player $i$ for notational convenience.)

\begin{itemize}
	\item State space: The state space is composed of $\calS = \cup_{i\in V}\calS_i$, where each \[\calS_i \defeq \{s_i^\lrm,s_i^\srm\}\cup \{s_i^{(l),{\srm}}\}_{l\in[2L+1]}\cup \{s_i^{(l),{\lrm}}\}_{l\in[2L-2]}\cup\Par{\cup_{j\neq i:j\in V}\{s_{i,j}^{(l),{\slrm}}\}_{l\in[2L-1]}}.\] 
	For simplicity of notation, we define as shorthand
\begin{align*}
 	\calS^{\aux,\srm}_i & \defeq \{s_i^{(l),\srm}\}_{l\in[2L+1]},\\
 	\calS^{\aux,\lrm}_i & \defeq \{s_i^{(l),\lrm}\}_{l\in[2L-2]},\\
 	\calS^{\aux,\slrm}_i & \defeq \cup_{j\neq i:j\in V}\{s_{i,j}^{(l),{\slrm}}\}_{l\in[2L-1]},\\
 	\calS^\aux & \defeq \cup_{i\in V}\Par{\calS^{\aux,\srm}_i\cup \calS^{\aux,\lrm}_i\cup \calS^{\aux,\slrm}_i},\\
 	\calS^{\lrm} & \defeq \cup_{i\in V}\{s_i^\lrm\}~~\text{and}~~\calS^{\srm} \defeq \cup_{i\in V}\{s_i^\srm\}.
 \end{align*}
 The total number of states is equal to $O(L|V|\cdot |V|) = O(L|V|^2)$. 
	\item Reward: Independent of the action that is taken, each player controlling states $s\in\calS^\srm$ receives positive instant reward $r_s(s) = 1>0$ , each player controlling states $s\in\calS^\lrm$ receives negative instant reward $r_s(s) = -1<0$. In all other cases, the instant reward is equal to $0$ everywhere.\footnote{Note that we have grouped all states (players) with positive instant rewards as $\calS^\srm$ as they have incentives to get into shorter cycles; and all states (players) with negative rewards as $\calS^\lrm$ as they are incentivized to get in longer cycles. The other states (players) are neutral as they receive reward equal to $0$ everywhere.}
	\item  Action space and (deterministic) transitions: We let $a(s)$ denote the action that with probability $1$ transits to state $s\in\calS$. For each state $s\in\calS$, we specify their corresponding action set as follows:
	\begin{equation}\label{eq:Hamiltonian-action-set}
	\calA_s=\begin{cases}
	\Par{\cup_{j:(i,j)\in E(G)}\{a(s_j^\srm)\}}\cup\{a(s_i^\srm)\}\cup\{a(s^{(1),\lrm}_{i})\},~~~~~\text{if}~~s = s_i^\lrm, i\in V;\\
	\Par{\cup_{j:j\in V(G)\setminus\{i\}}\{a(s_{i,j}^{(1),\slrm})\}}\cup\{a(s_i^\lrm)\}\cup\{a(s^{(1),\srm}_{i})\},~~~~\text{if}~~s = s_i^\srm, i\in V;\\
	\{a(s_i^{(l+1),\lrm})\}, \quad\quad\quad~~\text{if}~~s = s_i^{(l),\lrm}, l\in[2L], i\in V; \\
	\{a(s_i^{\lrm})\}, \quad\quad\quad\quad\quad~~\text{if}~~s = s_i^{(2L+1),\lrm}, i\in V; \\
	\{a(s_i^{(l+1),\srm})\}, \quad\quad\quad~~\text{if}~~s = s_i^{(l),\srm}, l\in[2L-3], i\in V; \\
	\{a(s_i^{\srm})\}, \quad\quad\quad\quad\quad~~\text{if}~~s = s_i^{(2L-2),\srm}, i\in V;  \\
	\{a(s_{i,j}^{(l+1),\slrm})\}, \quad\quad\quad\quad\quad~~\text{if}~~s = s_{i,j}^{(l),\slrm}, l\in[2L-2], i\neq j\in V; \\
	\{a(s_{i,j}^{\lrm})\}, \quad\quad\quad\quad\quad~~\text{if}~~s = s_{i,j}^{(2L-1),\slrm}, i\neq j\in V. \\
\end{cases}
	\end{equation}
\item Utility function: As before, we consider utility functions for each player $\upsilon_s$ defined as the value function of the state $s$ some player controls for that player.
\end{itemize}

\begin{figure}[h]
\centering
\centerline{\includegraphics[width=\textwidth]{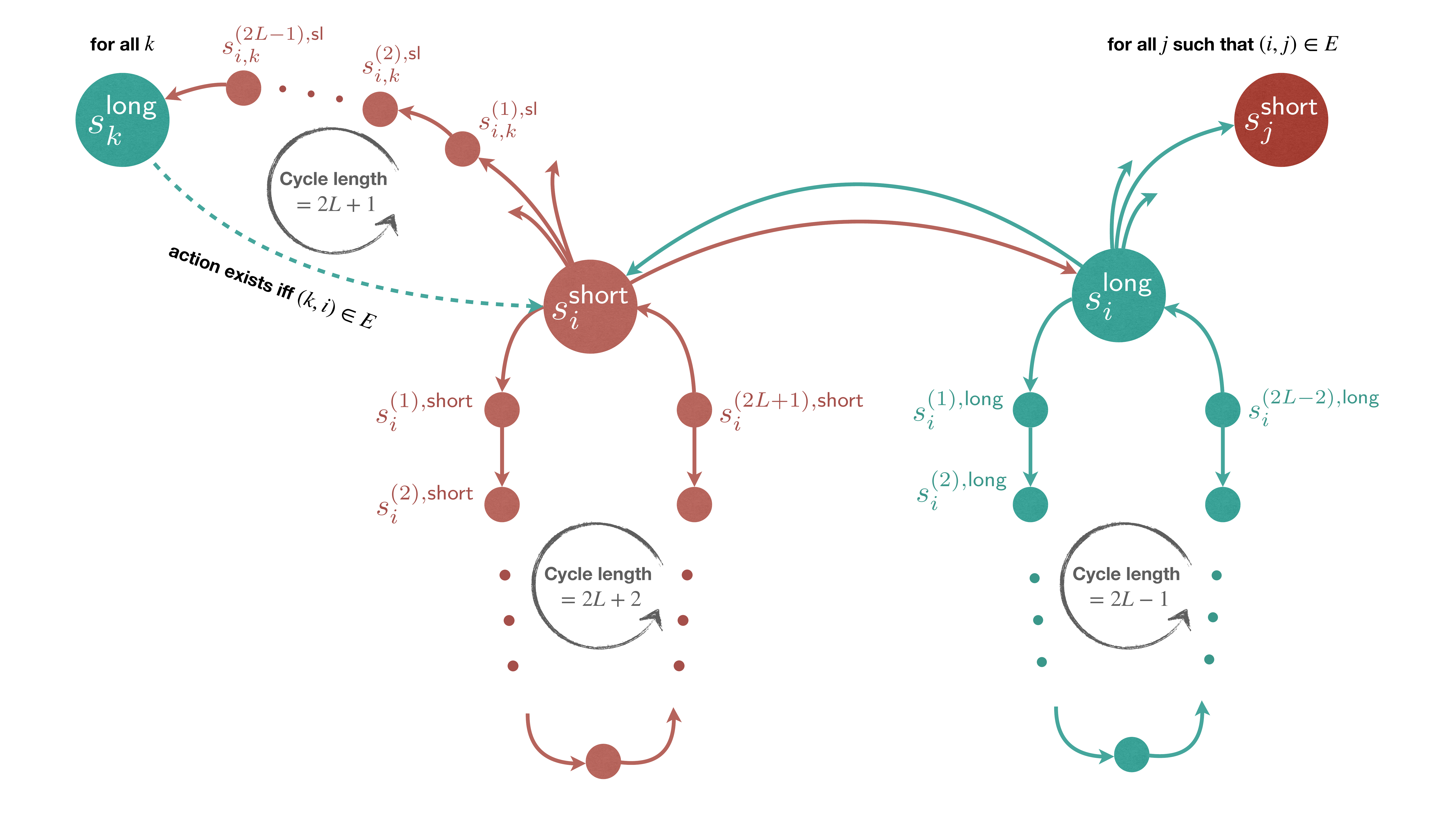}}
\caption{Illustration of construction of $\onestate$ $\calG_G$ based on a directed graph $G=(V,E)$.
Here we focus on the corresponding states $s_i^\srm$ (red center node which is incentivized to go in shorter cycles), $s_i^\lrm$ (blue center node which is incentivized to go in longer cycles) and players of a single vertex $i\in V$ of graph $G$. We show its different actions (all with deterministic transitions to some other states) using the correspondingly-colored arrows. Note that the player at $s_i^\srm$ has actions to all $s_k^\lrm$ for $k\in V, k\neq i$, while the player at $s_i^\lrm$ only has action to $s_j^\srm$ when $(i,j)\in E$. The construction of $\calG_G$ corresponding to the other vertices follow similarly.}
\label{fig:mixedgamegadgets}
\end{figure}

We first define a subgraph of $G$ induced by the strategy of $\calG_G$.

\begin{definition}\label{def:subgraph-induce-by-pi}[Induced Subgraph]
	Given a deterministic strategy $\ppi$ of above $\calG_{G}$, we define its \emph{induced subgraph} $G'(\calG_G, \ppi)$ on $G$, so that $e=(i,j)\in G'$ if and only if player at $s_i^\lrm$ picks $a(s_j^\srm)$ in the strategy $\ppi$, i.e.\ $\ppi_{s_i^\lrm} = \ee_{a(s_j^\srm)}$. 
\end{definition}

Based on this connection, we show the one-one correspondence between pure NE of $\calG_G$, and Hamiltonian path of original graph $G$, formally in~\Cref{prop:Hamiltonian}.

\begin{proposition}\label{prop:Hamiltonian}
Consider any directed graph $G=(V,E)$ with $L = |V|$ and the $\ronly$ $\onestate$ instance $\calG_G  = (n,\calS = \cup_{i\in[n]}\{s_i\},\mathcal{A},\PP,\RR,\gamma)$ constructed as above with $n = O(L^3)$ and instantaneous rewards $1$, $-1$ or $0$. 
Given accuracy $\delta<\frac{\gamma^{2L+2}(1-\gamma)}{(1-\gamma^{2L+2})(1-\gamma^{2L+3})}$, the game $\calG_G$ has a $\delta$-approximate pure NE $\ppi$ \emph{if and only if} the induced subgraph $G'(\calG_G,\ppi)$ defined in~\eqref{def:subgraph-induce-by-pi} is a Hamiltonian cycle of $G$. 
Further, any such pure strategy $\ppi$ must also be an exact pure NE.
\end{proposition}

To prove the proposition, we first provide a few lemmas that characterize the structural properties of an $\delta$-approximate pure NE (for sufficiently small $\delta$) of the game $\calG_G$, due to its construction.
The lemmas build on one another in sequence.

\begin{lemma}\label{lem:Hamiltonian-1}
Consider any $\delta< \frac{\gamma^{2L+2}}{1-\gamma^{2L+2}}$.
Then, at any $\delta$-approximate pure NE $\ppi$ of $\calG_G$, the player controlling $s_i^\lrm$ cannot take action $a(s_i^{(1), \lrm})$ for $i\in V$.	
\end{lemma}
\begin{proof}
We prove by contradiction. 
Suppose, instead, that there is some $\delta$-approximate pure NE $\ppi$ such that for some $i\in V$, the player controlling $s_i^\lrm$ takes action $a(s_i^{(1), \lrm})$. Then, the utility of player $s_i^\lrm$ is equal to $\upsilon_{s_i^\lrm}(\ppi) = -\frac{1}{1-\gamma^{2L-1}}$ given the actions taken by $\calS^{\aux,\lrm}_i$ defined in~\eqref{eq:Hamiltonian-action-set}. 
Now for player controlling state $s_i^\srm$, we denote its utility when taking action $a$ at its own state by $\upsilon_{s_i^\srm}(\ee_{a},\ppi_{-s_i^\srm})$.
We then get
\[\upsilon_{s_i^\srm}\Par{\ee_{a},\ppi_{-s_i^\srm}} = \begin{cases}
 1~&~\text{if}~~a = a(s_i^\lrm),\\
 \frac{1}{1-\gamma^{2L+2}}~&~\text{if}~~a = a(s_i^{(1),\srm}).
 \end{cases}
\]
Now recall that we defined $\delta< \frac{\gamma^{2L+2}}{1-\gamma^{2L+2}} = \frac{1}{1-\gamma^{2L+2}}-1$.
From the above, we note that player cannot pick action $a(s_i^\lrm)$ at state $s_i^\srm$ at any policy $\ppi$ that is a $\delta$-approximate pure NE. 
Therefore, if the game starts from state $s_i^\srm$ and all players follow the strategy $\ppi$, then the game will never transit to $s_i^\lrm$. 
As a consequence of this, we have
\[\upsilon_{s_i^\lrm}\Par{\ee_{a}, \ppi_{-s_i^\lrm}}=-1,~~\text{if}~~a = a(s_i^\srm),\]
and so $a(s_i^\srm)$ provides an improvement $\geq \delta$ over $s_i^{(1),\lrm}$ for the state $s_i^\lrm$.
This breaks the property of $\delta$-approximate NE and completes the proof.
\end{proof}

\begin{lemma}\label{lem:Hamiltonian-2}
Consider $\delta<  \frac{\gamma^{2L+2}(1-\gamma)}{(1-\gamma^{2L+2})(1-\gamma^{2L+1})}$.
Then, any $\delta$-approximate pure NE $\ppi$ of $\calG_G$ must satisfy the following property: if there is some $i\in V$ such that the player at state $s_i^\lrm$ takes action $a(s_j^\srm)$ for $j\neq i$, then the player at state $s_j^\srm$ must not take action $a(s_j^{(1),\srm})$ or $a(s_{j,i}^{(1),\slrm})$.	
\end{lemma}
\begin{proof}
Consider a $\delta$-approximate pure NE $\ppi$ such that the player at state $s_i^\lrm$ transits to $s_j^\srm$.
Then, the utilities of player $s_j^\srm$ are given by
\begin{align*}
\upsilon_{s_j^\srm}\Par{\ee_{a},\ppi_{-s_j^\srm}} = \begin{cases}
 \frac{1}{1-\gamma^{2L+1}}~&~\text{if}~~a = a(s_{j,i}^{(1),\slrm}),\\
 \frac{1}{1-\gamma^{2L+2}}~&~\text{if}~~a = a(s_j^{(1),\srm}).
 \end{cases}
\end{align*}
Noting that we need a $\delta$-approximate NE for the choice $\delta<  \frac{\gamma^{2L+1}(1-\gamma)}{(1-\gamma^{2L+2})(1-\gamma^{2L+1})}$, this rules out the possibility of player $s_j^\srm$ taking action $a(s_j^{(1),\srm})$.
To rule out the other action $a(s_{j,i}^{(1),\slrm})$, we prove by contradiction.
Suppose instead that the player $s_j^\srm$ takes action $a(s_{j,i}^{(1),\slrm})$ at $\ppi$.
Then, the state transits from $s_j^\srm \to s_i^\lrm$, and there is consequently an induced cycle $s_i^\lrm \to s_j^\srm \to s_i^\lrm$ of length equal to $2L+1$.

We now consider the action that player of $s_{i}^\srm$ takes at the $\delta$-approximate pure NE. 
Recall that player $s_i^\srm$ wishes to minimize the length of her own cycle.
Consequently, in an argument similar to the proof of~\Cref{lem:Hamiltonian-1}, player $s_i^\srm$ will never visit $s_i^\lrm$ at a $\delta$-approximate NE for any $\delta< \frac{\gamma^{2L+2}}{(1-\gamma^{2L+2})}$.
Finally, recall that player $s_i^\lrm$ wishes to maximize the length of her own cycle, or ideally find a path.
Consequently, player $s_i^\lrm$ is strictly incentivized to take action $a(s_i^\srm)$ in $\delta$-approximate NE which causes a contradiction with the assumption that player $s_i^\lrm$ takes action $a(s_j^\srm)$.
This completes the proof.
\end{proof}

\begin{lemma}\label{lem:Hamiltonian-3}
Consider any $\delta<  \frac{\gamma^{2L+2}(1-\gamma)}{(1-\gamma^{2L+2})(1-\gamma^{2L+3})}$ and a $\delta$-approximate pure NE $\ppi$ of $\calG_G$.
Then, for any $i\in V$, no player $s_i^\srm$ can take action $a(s_{i,j}^{(1),\slrm})$, for all $j\neq i$ and $j\in V$.	
\end{lemma}
\begin{proof}
We prove by contradiction. 
Consider, instead, a $\delta$-approximate pure NE $\ppi$ where state $s_i^\srm$ takes action $a(s_{i,j}^{(1),\slrm})$ for some $i\neq j\in V$. 
Given~\Cref{lem:Hamiltonian-2}, we know that state $s_j^\lrm$ must not take action $a(s_i^\srm)$ even if $(j,i)\in E$. 
Now we consider the value function of player $s_i^\srm$ when fixing the policies of other players in $\ppi$.
We have
\begin{align}\label{eq:valueshort}
\upsilon_{s_i^\srm}\Par{\ee_{a},\ppi_{-s_i^\srm}} \begin{cases}
 \le \frac{1}{1-\gamma^{2L+3}}~&~\text{if}~~a = a(s_{i,j}^{(1),\slrm}),\\
 =\frac{1}{1-\gamma^{2L+2}}~&~\text{if}~~a = a(s_i^{(1),\srm}).
 \end{cases}
\end{align}
Above, the first inequality is due to the fact that if $s_i^\srm$ still revisits itself, it must revisit itself with cycle length $\ge 2L+3$. This is because the state $s_j^\lrm$ that it visits does not transit back to $s_i^\srm$ directly under strategy $\ppi$. 
Noting that $\delta<\frac{\gamma^{2L+2}(1-\gamma)}{(1-\gamma^{2L+2})(1-\gamma^{2L+3})} $, Equation~\eqref{eq:valueshort} implies that $\ppi$ cannot be a $\delta$-approximate pure NE for player $s_i^\srm$. This proves the lemma.
\end{proof}

\begin{lemma}\label{lem:Hamiltonian-4}
Consider any $\delta<  \frac{\gamma^2(1-\gamma^{2L-3})}{(1-\gamma^2)(1-\gamma^{2L+2})}$.
Then, for any $i \in V$, no player of state $s_i^\lrm$ will take action $a(s_i^\srm)$ at a $\delta$-approximate pure NE $\ppi$ of $\calG_G$. 
\end{lemma}
\begin{proof}
	We prove by contradiction. 
	Suppose, instead, that there is a $\delta$-approximate pure NE $\ppi$ where for some $i\in V$, the player of $s_i^\lrm$ takes action $a(s_i^\srm)$. Then, the value function of player $s_i^\srm$ under different actions when fixing other players' strategies in $\ppi$ is equal to
	\[
	\upsilon_{s_i^\srm}\Par{\ee_{a},\ppi_{-s_i^\srm}}\begin{cases}
  = \frac{1}{1-\gamma^{2}} ~&~\text{if}~~a = a(s_{i}^{\lrm}),\\
  \le \frac{1}{1-\gamma^{2L+2}}~&~\text{if}~~a = a(s_{i,j}^{(1),\slrm}),\\
 =\frac{1}{1-\gamma^{2L+2}}~&~\text{if}~~a = a(s_i^{(1),\srm}).
 \end{cases}
	\]
Then, in order to satisfy a $\delta$-approximate NE with $\delta<\frac{\gamma^2(1-\gamma^{2L})}{(1-\gamma^2)(1-\gamma^{2L+2})}$, the player of $s_i^\srm$ must take action $a(s_i^\lrm)$ in $\ppi$. 
In that case, the value function for player $s_i^\lrm$ is equal to
	\[
	\upsilon_{s_i^\lrm}\Par{\ee_{a},\ppi_{-s_i^\lrm}} = \begin{cases}
   -\frac{1}{1-\gamma^{2}} ~&~\text{if}~~a = a(s_{i}^{\srm}),\\
 -\frac{1}{1-\gamma^{2L-1}}~&~\text{if}~~a = a(s_i^{(1),\lrm}).
 \end{cases}
	\]
	Consequently, for $\delta<\frac{\gamma^2(1-\gamma^{2L-3})}{(1-\gamma^2)(1-\gamma^{2L-1})}$, we have the player of $s_i^\lrm$ must not take action $a(s_{i}^{\srm})$ which leads to a contradiction.
	This concludes the proof.
\end{proof}

With these lemmas we are ready to formally characterize the necessary structural properties for a $\delta$-approximate pure NE $\ppi$ of game $\calG_G$, and thus proving~\Cref{prop:Hamiltonian}.

\begin{proof}[Proof of~\Cref{prop:Hamiltonian}]	
We first prove that when $\ppi$ is a $\delta$-approximate pure NE, then its induced subgraph $G'(\calG_G, \ppi)$ must be a Hamiltonian cycle. 
From Lemmas~\ref{lem:Hamiltonian-1},~\ref{lem:Hamiltonian-2},~\ref{lem:Hamiltonian-3} and~\ref{lem:Hamiltonian-4}, we know that if $\ppi$ is a $\delta$-approximate pure NE of $\calG_G$, then it must satisfy
\begin{align*}
	\ppi_s = \begin{cases}
	\ee_{a(s_j^\srm)}~~\text{for some}~~(i,j)\in E&\text{if}~~s = s_i^\lrm,\\
	\ee_{a(s_i^\lrm)}&\text{if}~~s = s_i^\srm~\text{and}~\exists k\in V~\text{such that}~\ppi_{s_k^\lrm}=\ee_{a(s_i^\srm)},\\
	\ee_{a(s_i^\srm)}~\text{or}~\ee_{a(s_i^{(1),\srm})}&\text{if}~~s = s_i^\srm~\text{without such}~k~\text{in previous case}.
	\end{cases}
\end{align*}

Now, let us consider the induced subgraph $G'(\calG_G, \ppi)$.
Suppose it is not a Hamiltonian cycle of original graph $G$.
Then, there must exist a cycle of smaller length $\le L-1$, which corresponds to a cycle of length $\le 2L-2$ in game $\calG'$. 
For any player of $s_i^\lrm$ that is on the cycle, the value function is equal to
\[
\upsilon_{s_i^\lrm}\Par{\ee_{a},\ppi_{-s_i^\lrm}} = \begin{cases}
   -\frac{1}{1-\gamma^{2L-2}} ~&~\text{if}~~a = a(s_{j}^{\srm})~\text{for some}~j:(i,j)\in E,\\
 -\frac{1}{1-\gamma^{2L-1}}~&~\text{if}~~a = a(s_i^{(1),\lrm}).
 \end{cases}
\]
This contradicts with the fact that $s_i^\lrm$ takes action $a(s_j^\srm)$ in a $\delta$-approximate pure NE $\ppi$ with $\delta<\frac{\gamma^{2L-2}(1-\gamma)}{(1-\gamma^{2L-2})(1-\gamma^{2L-1})}$. Consequently we conclude that the induced subgraph $G'(\calG_G,\ppi)$ must be a Hamiltonian cycle.

We now show the reverse argument, i.e. that we can construct a pure NE out of any Hamiltonian cycle of the graph $G$ if it exists.
Consider a Hamiltonian cycle $G'$ on the original graph $G$.
Then, we can construct a pure strategy for game $\calG_G$ as below:
\begin{align*}
	\ppi_s = \begin{cases}
	\ee_{a(s_j^\srm)}&\text{if}~~s = s_i^\lrm,~(i,j)\in E(G')\\
	\ee_{a(s_i^\lrm)}&\text{if}~~s = s_i^\srm.
	\end{cases}
\end{align*}
Since $G'$ is a Hamiltonian cycle, it uniquely defines a pure strategy $\ppi$ for all players.
Now for each state $s_i^\lrm$, the value function under different actions of the player $s_i^\lrm$ when fixing other players' action in $\ppi$ becomes equal to
\[
\upsilon_{s_i^\lrm}\Par{\ee_{a},\ppi_{-s_i^\lrm}}  \begin{cases}
   = -\frac{1}{1-\gamma^{2L}} ~&~\text{if}~~a = a(s_{j}^{\srm})~\text{for}~j:(i,j)\in E(G'),\\
   \le -\frac{1}{1-\gamma^{2(L-1)}} ~&~\text{if}~~a = a(s_{j}^{\srm})~\text{for}~j:(i,j)\notin E(G'),\\
   = -\frac{1}{1-\gamma^{2}}~&~\text{if}~~a = a(s_i^\srm),\\
 = -\frac{1}{1-\gamma^{2L-1}}~&~\text{if}~~a = a(s_i^{(1),\lrm}),
 \end{cases}
\]
which implies the current action $a_k^\srm$ for some $(i,k)\in E(G')$ of $s_i^\lrm$ is a best response to the other players' strategy.

Similarly, for each state $s_i^\srm$, the value function under different actions of the player $s_i^\srm$ when fixing other players' action in $\ppi$ becomes equal to
\[
\upsilon_{s_i^\srm}\Par{\ee_{a},\ppi_{-s_i^\srm}}  \begin{cases}
   \le \frac{1}{1-\gamma^{2L+2}} ~&~\text{if}~~a = a(s_{i,j}^{(1),\slrm})~\text{for}~j\in V(G)\setminus\{i\},\\
   = \frac{1}{1-\gamma^{2L}}~&~\text{if}~~a = a(s_i^\lrm),\\
 = \frac{1}{1-\gamma^{2L+2}}~&~\text{if}~~a = a(s_i^{(1),\srm}),
 \end{cases}
\]
where the first inequality is due to the fact that the cycle such action leads to when fixing other players' action must have length at least equal to $2L+2$. 
This also implies the current action of $s_i^\srm$ is a best response for other players' strategy.

Combining these together, we conclude that we can modify any Hamiltonian cycle $G'$ of $G$ into an exact pure NE of game $\calG_G$.
This proves the second part of the proposition.
\end{proof}

Following the NP-hardness of finding a Hamiltonian cycle of any directed graph~\cite{garey1979computers}, we obtain the NP-hardness of finding $\epsilon$-approximate pure NE for small enough $\epsilon$. We state the formal theorem as an immediate corollary of~\Cref{prop:Hamiltonian} here for completeness.

\begin{theorem}
It is NP-hard to find an $\epsilon$-approximate pure NE in $\gamma$-discounted mixed-sign $\ronly$ $\onestate$s with both positive and negative rewards for any accuracy $\epsilon <\frac{\gamma^{2\Atot+2}(1-\gamma)}{(1-\gamma^{2\Atot+2})(1-\gamma^{2\Atot+3})}$.
\end{theorem}
In particular, setting $\gamma = 1-\poly(\Atot)^{-1}$, this implies the NP-hardness of finding approximate pure NE with accuracy under some small enough constant $\epsilon\le \epsilon_0\in(0,1)$ in the class of $\gamma$-discounted mixed-sign $\ronly$ $\onestate$s.

\section*{Acknowledgment}
The authors thank Constantinos Daskalakis, Noah Golowich, and Kaiqing Zhang for kindly coordinating on uploads to arXiv. The authors also thank Aviad Rubinstein and anonymous reviewers for helpful feedback.

Researchers are supported in part by an Adobe Data Science Research Award,  a Danzig-Lieberman Graduate Fellowship, a Google Research Colabs Award, a Microsoft Research Faculty Fellowship, NSF CAREER Award CCF-1844855, NSF Grant CCF-1955039, NSF Grant IIS-2212182, a PayPal research
award, a Sloan Research Fellowship and a Stanford Graduate Fellowship.  Part of this work was conducted while the authors were
visiting the Simons Institute for the Theory of Computing.

\newpage

\bibliographystyle{alpha}
\newcommand{\etalchar}[1]{$^{#1}$}

\onecolumn
\newpage
\appendix

\part*{Supplementary material}
\section{A counterexample to pseudo-linearity for $\tbsg$}\label{app:qmcounterexample}

In this section, we provide a simple counterexample to the pseudolinearity property for $\tbsg$ when one player can control multiple states.
We consider a $2$-player $\tbsg$ with $3$ states: $\calS = \{A,B,C\}$, and $2$ actions per state. 
We denote the action set for both players as $\calA = \{1,2\}$, and let player $1$ control states $\{A,B\}$ while player $2$ controls state $C$.
We fix the two candidate policies as $\ppi = (\ee_1, \ee_1, \ee_2)$ and $\ppi' = (\ee_2, \ee_2, \ee_2)$.
We will now specify the parameters of the $\tbsg$ (i.e. the transition probabilities and instantaneous reward functions for player $1$) such that her value function is \emph{not} pseudo-linear, i.e. for any $\lambda \in [0,1]$, $V^{\lambda \ppi + (1 - \lambda) \ppi'}$ is \emph{not} maximized at $\lambda \in \{0,1\}$.

We fix the following details for the $\tbsg$:
\begin{itemize}
\item We specify the transition probabilities from states that are controlled by player $1$ (i.e. $\{A,B\}$) as: 
\begin{align*}
\pp_{A,2} &= \begin{bmatrix} 1/5 & 2/5 & 2/5 \end{bmatrix} \\
\pp_{A,1} &= \begin{bmatrix} 2/5 & 1/5 & 2/5 \end{bmatrix} \\
\pp_{B,2} &= \begin{bmatrix} 2/5 & 1/5 & 2/5 \end{bmatrix} \\
\pp_{B,1} &= \begin{bmatrix} 1/15 & 4/5 & 2/15 \end{bmatrix} \\
\pp_{C,2} &= \begin{bmatrix} 2/5 & 2/5 & 1/5 \end{bmatrix}.
\end{align*}
(Note that we do not need to specify the transition probabilities $\pp_{C,1}$, as we are considering only the case where player $2$ takes action $2$ at his state.)

As a result of this, we get the following transition kernels:
\begin{align*}
\PP^{\ppi} &= \begin{bmatrix}
2/5 & 1/5 & 2/5 \\
1/15 & 4/5 & 2/15 \\
2/5 & 2/5 & 1/5 
\end{bmatrix} \\
\PP^{\ppi'} &= \begin{bmatrix}
1/5 & 2/5 & 2/5 \\
2/5 & 1/5 & 2/5 \\
2/5 & 2/5 & 1/5.
\end{bmatrix}
\end{align*}
\item We consider discount factor $\gamma = 5/6$.
\item For player $1$, we specify the reward function
\begin{align*}
r_{1,A,1} &= 1 \\
r_{1,A,2} &= 1.1 \\
r_{1,B,a} &= 0 \text{ for } a \in \{1,2\}.
\end{align*}
\end{itemize}
We assume that play starts at state $A$.
From these specifications, simple algebra gives
\begin{align*}
V^{(\lambda \ppi + (1 - \lambda) \ppi')} = \frac{(63 - 37 \lambda)(1 + 0.1 \lambda)}{(4 \lambda^2 - 84 \lambda + 147)}
\end{align*}
which is \emph{not} maximized at $\lambda \in \{0,1\}$, thus violating the pseudolinear property.

\section{Bellman optimality from Nash equilibrium (NE)}
\label{ssec:bellman}

In this section we show the connection of NE of $\ssg$s and their corresponding Bellman equations. We define the Bellman equations first for $\ssg$s and thus the conditions for $\tbsg$s follow as special cases.

By definition, a strategy $\ppi$ is an exact NE of an $\ssg$ instance  for some initial distribution $\qq\in\Delta^\calS$, $q(s)>0$ for any $s\in\calS$, if and only if the following Bellman optimality conditions are satisfied:
\begin{equation}\label{eq:NE-bellman}
V_i^{\ppi}(s) = 
 	\begin{cases}
 	\max_{a\in\calA_{i,s}} \Brack{r^{(\ee_{a},\ppi_{(i,-s)},\ppi_{-i})}_i(s)+\gamma \PP^{(\ee_{a},\ppi_{(i,-s)},\ppi_{-i})}(s,\cdot) \VV_i^{\ppi}},~&~\text{for any}~i\in[n],~s\in\calS_i,\\
 	r^{\ppi}_i(s)+\gamma \PP^{\ppi}(s,\cdot) \VV_i^{\ppi},~&~\text{for any}~i\in[n],~s\notin\calS_i.
 \end{cases}
\end{equation}

This implies that whether a strategy $\ppi$ is an exact NE or not of the $\ssg$ instance does not depend on the particular initial distribution $\qq$ as long as it satisfies non-degeneracy. In particular, it also shows for $\onestate$ ($\onestatessg$), the notion of exact NE is unchanged if we define utilities as $u_i(\ppi) = V_i^{\ppi}(s_i)$ instead.

We next prove the following theorems that connect Bellman optimality equations with $\eps$-approximate NEs, which we use in the derivations of~\Cref{sec:ppad-membership,sec:tbsg}.

\begin{lemma}[Necessary Bellman condition]\label{lem:bellman-nece}
Considering the Bellman optimality systems, a necessary condition of a strategy $\ppi$ being an $\epsilon$-approximate NE (when utilities are defined under the fixed initial uniform distribution) is that the corresponding value functions under strategy $\ppi$ satisfy
\begin{equation}\label{eq:NE-bellman-nece}
V_i^{\ppi}(s)  \begin{cases}
 	\ge \max_{a\in\calA_{i,s}} \Brack{r^{(\ee_{a},\ppi_{(i,-s)},\ppi_{-i})}_i(s)+\gamma \PP^{(\ee_{a},\ppi_{(i,-s)},\ppi_{-i})}(s,\cdot) \VV_i^{\ppi}}-|\calS|\epsilon,~~\text{for any}~i\in[n],~s\in\calS_i,\\
 	= r^{\ppi}_i(s)+\gamma \PP^{\ppi}(s,\cdot) \VV_i^{\ppi},~~\text{for any}~i\in[n],~s\notin\calS_i.
 \end{cases}
\end{equation}
\end{lemma}

\begin{proof}
	Suppose Eq.s~\eqref{eq:NE-bellman-nece} fail to hold, i.e. there exists some player $i\in[n]$, state $s\in\calS_i$, and $a\in\calA_{i,s}$ such that $V_i^{\ppi}(s)  < r^{(\ee_{a},\ppi_{(i,-s)},\ppi_{-i})}_i(s)+\gamma \PP^{(\ee_{a},\ppi_{(i,-s)},\ppi_{-i})}(s,\cdot) \VV_i^{\ppi}-|\calS|\epsilon$. Now apply the Bellman equality under the new strategy $\ppi' = (\ee_a, \ppi_{(i,-s)}, \ppi_{-i})$, by monotonicity and convergence we know that $\VV_i^{\ppi'}(s')\le \VV_i^{\ppi}(s')$ for all $s'\in\calS$, and in particular $\VV_i^{\ppi'}(s')> \VV_i^{\ppi}(s')+|\calS|\epsilon$. 
	
	Now for initial distribution $\qq = \frac{1}{|\calS|}\ee_\calS$, one has $u_i^{\ppi} = \langle \qq,\VV_i^{\ppi}\rangle$ and consequently $u_i^{\ppi'}> u_i^{\ppi}+\epsilon$. This implies that $\ppi$ is not an $\epsilon$-approximate NE and leads to a contradiction.
\end{proof}

\begin{lemma}[Sufficient Bellman condition]\label{lem:bellman-suff}A sufficient condition of a strategy $\ppi$ being an $\epsilon$-approximate NE is that there exists 
\begin{equation}\label{eq:NE-bellman-suff}
V_i^{\ppi}(s) \begin{cases}
 	\ge \max_{a\in\calA_{i,s}} \Brack{r^{(\ee_{a},\ppi_{(i,-s)},\ppi_{-i})}_i(s)+\gamma \PP^{(\ee_{a},\ppi_{(i,-s)},\ppi_{-i})}(s,\cdot) \VV_i^{\ppi}}-(1-\gamma)\epsilon,~~\text{for any}~i\in[n],~s\in\calS_i,\\
 	= r^{\ppi}_i(s)+\gamma \PP^{\ppi}(s,\cdot) \VV_i^{\ppi},~~\text{for any}~i\in[n],~s\notin\calS_i.
 \end{cases}
\end{equation}
\end{lemma}

\begin{proof}
	Consider a strategy $\ppi$ satisfying Eq.s~\eqref{eq:NE-bellman-suff}, we have for any player $i\in[n]$, for any alternative strategy $\ppi' = (\ppi'_{i}, \ppi_{-i})$ we have
	\begin{align*}
		V_i^{\ppi}(s) \begin{cases}
 	\ge \Brack{r^{\ppi'}_i(s)+\gamma \PP^{\ppi'}(s,\cdot) \VV_i^{\ppi}}-(1-\gamma)\epsilon,~~\text{for any}~i\in[n],~s\in\calS_i,\\
 	= r^{\ppi'}_i(s)+\gamma \PP^{\ppi'}(s,\cdot) \VV_i^{\ppi},~~\text{for any}~i\in[n],~s\notin\calS_i.
 	 \end{cases}
	\end{align*}
	
Combined the above with $V_i^{\ppi'}(s) = r^{\ppi'}_i(s)+\gamma \PP^{\ppi'}(s,\cdot) \VV_i^{\ppi'}$ for all $s\in\calS$, we have for any $s$, $V_i^{\ppi}(s)-V_i^{\ppi'}(s)\ge \gamma\min_{s\in\calS}\left[V_i^{\ppi}(s)-V_i^{\ppi'}(s)\right]-(1-\gamma)\epsilon$. Taking minimum over $s\in\calS$ on LHS and rearranging terms we obtain $V_i^{\ppi}(s)-V_i^{\ppi'}(s)\ge -\epsilon$ for all states $s$, given a player $i\in[n]$. Thus, by definition of utility function under the fixed initial uniform distribution, 
\[
u_i^{\ppi}=\left\langle \frac{1}{|\calS|}\ee_{\calS}, V^{\ppi}\right\rangle,
\] 
we can conclude that
\[
u_i^{\ppi}\ge u_i^{(\ppi'_i, \ppi_{-i})}-\epsilon,~~\text{for any}~~ \ppi_{i,s}' \in \Delta^{\calA_{i,s}}.
\]
\end{proof} 

As an immediately corollary of the above results, we show the equivalence between NEs for $\onestatessg$s (and $\onestate$s) under the uniform distribution $\qq = \frac{1}{|\calS|}\ee_\calS$ and under a single support $\ee_{s_i}$ for each player $i$, up to polynomial factors.

\begin{corollary}[Equivalence of NE notions for $\onestatessg$]\label{coro:equivalent-notion-NE}
	For any given $\onestatessg$ (or $\onestate$) instance, any $\epsilon$-approximate NE under uniform distribution $u_i(\ppi) = v_i^{\ppi,\qq}$ for $\qq = \frac{1}{|\calS|}\ee_\calS$ is also an $|\calS|\epsilon$-approximate NE under single-support distribution $u_i^{\ppi} = V_i^{\ppi}(s_i)$; any $\epsilon$-approximate NE under single-support distribution is also an $\epsilon/(1-\gamma)$-approximate NE under uniform distribution.
\end{corollary}

\section{Infinite-horizon un-discounted stochastic games (when $\gamma\rightarrow$ 1)}\label{apdx:average}

In this section, we show how some of the main results in the main paper generalize to infinite-hor izon general-sum average-reward games ($\ar$-$\ssg$ and $\ar$-$\tbsg$) and absorbing games ($\Abs$-$\ssg$ and $\Abs$-$\tbsg$). In~\Cref{ssec:ar-prelims}, we provide some preliminaries for defining these variants of $\ssg$s and $\tbsg$s. In~\Cref{ssec:redx-average-to-discounted}, we provide a generic reduction from solving average-reward games to solving discounted games approximately, and show $\ppad$-membership of finding $\eps$-approximate NEs for average-reward games. In~\Cref{ssec:redx-discounted-to-average}, we provide a generic reduction from discounted games to absorbing games, and show $\ppad$-hardness of finding approximate NEs for $\Abs$-$\tbsg$.

\subsection{Definitions of infinite-horizon un-discounted general-sum games}\label{ssec:ar-prelims}

In this section, we introduce the basic concepts of infinite-horizon general-sum average-reward games ($\ar$-$\ssg$ and $\ar$-$\tbsg$) and absorbing games ($\Abs$-$\ssg$ and $\Abs$-$\tbsg$). The main difference of the general setup is that we use tuple $\mathcal{G} = (n, \calS ,\mathcal{A},\pp,\r)$, while omiting the last entry $\gamma$ since we let $\gamma=1$, i.e. consider no discounting in future instantaneous rewards for the average-reward setup. All the other elements including states, action space, $\calS_i$ and $\calI_s$ are defined identically.

The goal of the players is still to maximize the expected infinite-horizon reward, which is defined differently though from the discounted case as 
\begin{equation}\label{eq:def-ar-v-i}
v_i = 
\begin{cases}\lim_{T\rightarrow\infty}\E\left[\frac{1}{T}\sum_{t<T}r_{i,s^t, \aa^t}\right],~~\forall i\in[n]~~\text{for average-reward games},\\
\lim_{T\rightarrow\infty}\E\left[\sum_{t<T}r_{i,s^t, \aa^t}\right],~~\forall i\in[n]~~\text{for absorbing games}.
\end{cases}
\end{equation}

Under a strategy $\ppi$, the value functions of each player are similarly defined to be 
\begin{equation}\label{eq:def-ar-V-i-pi}
\begin{aligned}
V_i^{\ppi}(s) = 
\begin{cases}
	\lim_{T\rightarrow\infty}\E\left[\frac{1}{T}\sum_{t<T} r_{i,s^t,\aa^t}|s_0 = s, a^t_{j,s^t}\sim \ppi_{j,s^t}~\text{for all}~j,t\right],~\text{for each}~i\in[n], s\in\calS\\
\lim_{T\rightarrow\infty}\E\left[\sum_{t<T} r_{i,s^t,\aa^t}|s_0 = s, a^t_{j,s^t}\sim \ppi_{j,s^t}~\text{for all}~j,t\right],~\text{for each}~i\in[n], s\in\cal
\end{cases},\\
~~\text{and also satisfies}~~v_i^{\ppi,\qq} =\langle \qq,\VV_i^\pi\rangle.
\end{aligned}
\end{equation}

The NE for both $\ssg$ and $\tbsg$ are defined identically from definition, using the new definition of $\VV_i^{\ppi}$ in~\eqref{eq:def-ar-V-i-pi}. The equivalence between exact NE under utilities defined from initial $\qq = \frac{1}{|\calS|}\ee_\calS$ and from $V_i^{\ppi}(s_i)$ also applies.

We now introduce the notion of mixing un-discounted $\ssg$s, as a generalization of the mixing average-reward Markov Decision processes~\cite{jin2021v}, as follows:

\begin{definition}\label{def:mixing}
	An un-discounted $\tbsg$ instance with state space $\calS$ is \emph{multi-chain mixing} if for any initial distribution $\ee_s$, for some $s\in\calS$ and some strategy $\ppi$,  for any initial distribution $\ee_s\in\Delta^{\hat{\calS}}$ there exists a stationary distribution $\llam^{s,\ppi}$ so that the induced Markov chain has mixing time bounded by $\tmix<\infty$ , where $\tmix$ is defined as 
	 \begin{equation*}\label{def-multichain-mixingtime}
	\tmix\defeq\max_{\ppi}\left[\argmin_{t\ge1 }\left\{\max_{\ee_s\in\Delta^{\calS}} \norm{({\PP^{\ppi}}^{\top})^t\qq-\llam^{s,\ppi}}_1\le \tfrac{1}{2}\right\}\right].
	\end{equation*}
	Specifically, when $\llam^{s,\ppi} = \llam^{\ppi}$ for all $s\in\calS$, we further call it \emph{unichain mixing}, and $\tmix$ can be equivalently defined as,
	\begin{equation*}\label{def-mixingtime}
	\tmix\defeq\max_{\ppi}\left[\argmin_{t\ge1 }\left\{\max_{\qq\in\Delta^{\calS}} \norm{({\PP^{\ppi}}^{\top})^t\qq-\llam^{\ppi}}_1\le \tfrac{1}{2}\right\}\right].
	\end{equation*}
\end{definition}

\subsection{Reduction from average-reward $\ssg$ to discounted $\ssg$ with error}\label{ssec:redx-average-to-discounted}

In this section, we show finding an approximate NE for average-reward $\ssg$ can be reduced to finding an approximate NE for discounted $\ssg$. This approximation works for any multi-chain mixing average-reward $\ssg$ as long as there is a polynomially-bounded mixing time $\tmix$ (see~\Cref{def-mixingtime}). It is a slight generalization of Lemma 1 and Lemma 2 in~\cite{jin2021towards}.

 \begin{lemma}
 	Given any multi-chain mixing average-reward $\ssg$ 
 $\mathcal{G} = (n, \calS ,\mathcal{A},\pp,\r)$ with mixing time bound $\tmix<\infty$ and instant reward bound $|r_{i,s,\aa_s}|\le 1$ for all $i\in[n],s\in\calS, \aa_s\in\calA_s$, consider its corresponding discounted $\ssg$ $\hat{\calG} = (n, \calS ,\mathcal{A},\pp,\r, \gamma)$ with  $\gamma = 1-\eps/(9\tmix)$, we use $u$ and $\hat{u}$ to denote their corresponding utility functions under uniform initial distribution, then any $\frac{\eps}{3(1-\gamma)}$-approximate NE $\ppi$ of $\hat{\calG}$ is also an $\eps$-approximate NE of $\calG$.
 \end{lemma}

\begin{proof}
	Following similar approach in~\cite{jin2021towards}, for any stationary strategy $\ppi$, we first let $\LLam^{\ppi}$ is the matrix where its $s$th row corresponds to $\llam^{s,\ppi}$ and observe that
	\begin{align}\label{eq:ar-contracitivity-linf}
		\norm{\left(\PP^{\ppi}\right)^k-\LLam^{\ppi}}_\infty\le \left(\frac{1}{2}\right)^{\lfloor k/\tmix \rfloor}.
	\end{align}
	
	Next, we show for the value functions of some player $i$, we have $\norm{\VV_i^{\ppi}-(1-\gamma)\hat{\VV}^{\ppi}_i}_\infty\le 3(1-\gamma)M\tmix\le \epsilon/3$. To see this, we let $\LLam^{\ppi}$ be the matrix where its $s$th row corresponds to $\llam^{s,\ppi}$, note
\begin{align*}
\VV^\pi_i-(1-\gamma)\hat{\VV}^\pi_i 
	&= (1-\gamma)\sum_{t\ge 0}\gamma^t\LLam^{\ppi}\r_i^{\ppi} - (1-\gamma)\sum_{t\ge 0}\gamma^t(\PP^{\ppi})^t\r_i^\pi\\
	&= (1-\gamma)\sum_{t=0}^\infty \gamma^t\left[(\PP^{\ppi})^t-\LLam^{\ppi}\right]\r_i^{\ppi}.
\end{align*}
For all $t \geq \tmix$ we have $\norm{(\PP^{\ppi})^t-\LLam^{\ppi}}_\infty \leq 2^{- \lfloor k/\tmix \rfloor}$ by  \Cref{eq:ar-contracitivity-linf}. Plugging this back yields the desired bound of
\begin{align*}
\norm{\VV_i^{\ppi}-\hat{\VV}_i^{\ppi}}_\infty  
&\leq
(1-\gamma)\sum_{t=0}^{\tmix - 1} \gamma^t\norm{(\PP^{\ppi})^t-\LLam^{\ppi}}_\infty 
+ (1-\gamma)\sum_{t \geq  \tmix} \gamma^t \norm{(\PP^{\ppi})^t-\LLam^{\ppi})^\top}_\infty  \\
&\leq (1-\gamma)\sum_{t=0}^{ \tmix - 1} 2 \gamma^t 
+ (1-\gamma)\sum_{t \geq  \tmix} \frac{1}{2^{\lfloor k/\tmix\rfloor}}
\leq 3(1 - \gamma)\tmix\le \frac{\eps}{3}\,.
\end{align*}

Note the above inequality holds for all players $i\in[n]$, consequently, given an $\frac{\eps}{3(1-\gamma)}$-approximate NE $\ppi$ of $\hat{\calG}$, we have for average-reward game $\calG$, it holds that for any player $i\in[n]$,

\begin{align*}
u_i(\ppi)+\tfrac{\eps}{3} & \stackrel{(i)}{\ge} (1-\gamma)\hat{u}_i(\ppi)\stackrel{(ii)}{\ge} (1-\gamma)\hat{u}_i(\ppi_i',\ppi_{-i})-\tfrac{\eps}{3} \stackrel{(i)}{\ge} \left(u_i(\ppi_i',\ppi_{-i})-\tfrac{\eps}{3}\right)-\tfrac{\eps}{3},~~\text{for any}~\ppi_{i,s}\in\Delta^{\calA_{i,s}}.
\end{align*}
Here we use $(i)$ the definition of $u$ under uniform initial distribution and inequality of $\norm{\VV_i^{\ppi}-\hat{\VV}_i^{\ppi}}_\infty \le \epsilon/3$, and $(ii)$ the definition of $\ppi$ as approximate NE of $\hat{\calG}$. Altogether this shows that $\ppi$ is an $\eps$-approximate NE for average-reward $\ssg$ $\calG$.
\end{proof}

This reduction shows: we can generalize the $\ppad$-membership of discounted $\ssg$ with polynomially-bounded discount factor $1/(1-\gamma)$ also  to the $\ppad$-membership of average-reward $\ssg$ with polynomially-bounded mixing time bound $\tmix$, which we state formally below for completeness.

\begin{corollary}[$\ppad$-membership of $\ar$-$\ssg$]\label{coro:ppad-membership}
The problem of $\eps$-approximate NE computation in multi-chain average-reward $\ssg$s  with mixing time bounded by $\tmix$ is in $\ppad$ for $\eps = \Omega(1/\poly(\Atot))$ and $\tmix = \poly(\Atot)$. 
\end{corollary}

\subsection{Reduction from discounted $\tbsg$ to absorbing $\tbsg$}\label{ssec:redx-discounted-to-average}

In this section, we provide a formal reduction from $\gamma$-discounted $\tbsg$ to absorbing $\tbsg$ with fast mixing time. Given any $\gamma$-discounted $\tbsg$ instance $\mathcal{G} = (n,\calS,\mathcal{A},\pp,\r,\gamma)$, we consider a $\gamma'$-discounted or absorbing undiscounted instance with $\hat{\calG} = (n',\hat{\calS},\hat{\calA},\hat{\pp},\hat{\r},\gamma')$ with the following specifications:
\begin{itemize}
\item State space: the state space consists of all states in $\calS$ and one extra absorbing state which we call $s_0$, so we have $\hat{\calS} = \{s_0\}\cup\calS$. We let $n' = n+1$ with an additional player controlling the new state $s_0$. 
\item Action space: the action set for each state $s\in \calS$ remains the same, i.e.\ $\hat{\calA}_{i,s}=\calA_{i,s}$ for all $s\neq s_0$, and $\hat{\calA}_{n+1,s_0}=\{a_0\}$, $\hat{\calA}_{i,s_0}=\emptyset$ for all $i\in[n]$. In other words, we let a new player to control the new state $s_0$ with a degenerate single action $a_{s_0}$ in its actions set. 
\item Probability transition kernels: For each state $(s,a_s)$ where $s\in\calS$, $a_s\in\calA_s$, we have transition probability $\hat{\pp}_{s,a_s} = \frac{\gamma}{\gamma'}\pp_{s,a_s}+(1-\frac{\gamma}{\gamma'})\ee_{s_0}$. For state $s_0$ and its only action $a_{0}$ we have  $\hat{\pp}_{s_0,a_{s_0}}=\ee_{s_0}$. 
\item Instantaneous rewards: For each state $(s,a)$ where $s\in\calS$, $a_s\in\calA_s$, the instantaneous reward remains unchanged for all player, i.e. $\hat{r}_{i,s,a_s} = r_{i,s,a_s}$. For state $s_0$, the instant reward is $0$, i.e. $\hat{r}_{i,s_0,a_{s_0}}=0$ for all $i\in[n]$.
\item Discount factor $\gamma'$: When $\gamma'\in(\gamma,1)$ we still consider the discounted reward model; when $\gamma'=1$ we alternatively consider the average-reward model.
\end{itemize}
Note in this construction we will get an $\onestate$ (or $\Abs$-$\onestate$) instance from any $\onestate$ instance, which  preserves the fact that each player only controls a single state.

Given a strategy $\ppi$ for the original discounted $\tbsg$ $\calG$, we keep $\ppi$ as the notation for a corresponding strategy for $\gamma'$-discounted or average-reward $\tbsg$ $\hat{\calG}$ where all policies for $s\in\calS$ remains unchanged, and player $n+1$ controlling $s_0$ will just take the single action available $a_{s_0}$ at state $s_0$. In this setting (as a typical generalization for discounted or absorbing MDPs), we let the utility function for each player $i\in[n]$ under some joint strategy $\ppi$ and initial distribution $\qq\in\Delta^\calS$ be 
\[
\hat{u}_i(\ppi) = \hat{v}_i^{\ppi,\qq} = \begin{cases}
 	\lim_{T\rightarrow\infty}\E\left[\sum_{t=0}^T\hat{r}_{i,s^t,a^t_{s^t}}|s^0 \sim\qq, a^t_{s^t}\sim \ppi_{s^t}~\text{for all}~t\right],~\text{for each}~i\in[n],~~\text{if}~\gamma'=1\\
 	\E\left[\sum_{t\ge0}(\gamma')^t\cdot \hat{r}_{i,s^t,a^t_{s^t}}|s^0 \sim\qq, a^t_{s^t}\sim \ppi_{s^t}~\text{for all}~t\right],~\text{for each}~i\in[n],~~\text{if}~\gamma\le \gamma'<1
 \end{cases}
.
\] 
We note that when $\gamma'=1$, the value (and utility) function we consider for the model defined especially for absorbing MDP, as in~\eqref{eq:def-ar-V-i-pi}. By construction our instance is also clearly $O(1/(1-\gamma))$-mixing, which we prove formally below

\begin{lemma}\label{lem:redx-mixing-time}
	Given a $\gamma$-discounted $\tbsg$ $\calG$, the absorbing $\tbsg$ $\hat{\calG}$ constructed above has mixing time bound $\tmix = O(1/(1-\gamma))$.
\end{lemma}

\begin{proof}
	By construction, we have the probability transition matrix of $\hat{\calG}$ under any strategy $\ppi$ can be expressed as \[\hat{\PP}^{\ppi} = \begin{bmatrix}
 1 & \mathbf{0_\calS}^\top\\
 (1-\gamma)\ee_\calS & 	\gamma \PP^{\ppi}
 \end{bmatrix} = (1-\gamma)\ee_{\calS'}\begin{bmatrix}
1 & \mathbf{0}_\calS^\top	
\end{bmatrix}+\gamma \begin{bmatrix}
 1 & \mathbf{0}_\calS^\top\\
 \mathbf{0}_\calS & \PP^{\ppi}
 \end{bmatrix}.
\]

Consequently, for any initial distribution $\qq,\qq'\in\Delta^{\calS'}$, we have 
\[
\norm{\left(\Par{\hat{\PP}^{\ppi}}^t\right)^\top\Par{\qq-\qq'} }_1\le 2\gamma^t \implies \tmix\le O\Par{\frac{1}{1-\gamma}},
\]
by definition of $\tmix$, which concludes the proof.
\end{proof}

Next, we show the formal lemma that reduces solving $\gamma$-discounted $\tbsg$ to solving $\gamma'$-discounted $\tbsg$ or $O(1/(1-\gamma))$-mixing absorbing $\tbsg$. This is a pretty straight-forward generalization of the standard reduction for MDPs, e.g. see~\cite{puterman2014markov}.

\begin{lemma}\label{lem:hardness-ineq}
Under some fixed strategy $\ppi$, for any player $i\in[n]$ and some initial probability $\qq\in\Delta^\calS$, let $\hat{v}$, $v$ be the value functions of $\hat{\calG}$ (for some $\gamma'\in[\gamma,1]$), $\calG$, respectively, it holds that, 
\[
\hat{v}_i^{\ppi,\qq} = v_i^{\ppi,\qq}.
\]
Consequently, when utility functions are defined under the same uniform initial distributions $\qq = \frac{1}{|\calS|}\ee_\calS$, then $\ppi$ is a $\delta$-approximate NE for $\gamma$-discounted $\tbsg$ $\calG$ if and only if it is also a $\delta$-approximate NE for the correspondingly-constructed $\gamma'$-discounted $\tbsg$ or absorbing $\tbsg$ $\hat{\calG}$.
\end{lemma}
\begin{proof}
To see this, note it holds that $\hat{\PP}^{\ppi} = \begin{bmatrix}
 1 & \mathbf{0_\calS}^\top\\
 (1-\frac{\gamma}{\gamma'})\ee_\calS & 	\frac{\gamma}{\gamma'} \PP^{\ppi}
 \end{bmatrix}
$ and $\hat{\r}^{\ppi} = \begin{bmatrix}
 0\\
 \r^{\ppi}	
 \end{bmatrix}
$, we have the equality that for any initial distribution $\qq\in\Delta^\calS$ and $t\ge0$, 
\begin{equation}\label{eq:redx-average-to-discounted-equality}
\qq^\top\Par{\PP^{\ppi}}^t\r_i^{\ppi} = \begin{bmatrix}
 	0 & \qq^\top
 \end{bmatrix}
(\gamma')^t\Par{\hat{\PP}^{\ppi}}^t\hat{\r}_i^{\ppi}.	
\end{equation}
Now expanding out the expression of $\hat{V}$ and $V$, respectively, for the discounted case we have
\[
v_i^{\ppi,\qq} = \qq^\top\Par{\sum_{t\ge 0}\gamma^t\Par{\PP^{\ppi}}^t\r_i^{\ppi}}.
\]
For the average-reward case
\[
\hat{v}_i^{\ppi,\qq} = \begin{bmatrix}
 0 & \qq^\top 	
 \end{bmatrix}\Par{\sum_{t = 0}^\infty(\gamma')^t\Par{\hat{\PP}^{\ppi}}^t\hat{\r}_i^{\ppi}} = v_i^{\ppi,\qq},
\]
where for the last equality we use~\eqref{eq:redx-average-to-discounted-equality}. 
\end{proof}

Combining reduction from solving $\gamma$-discounted $\tbsg$ to solving $O((1-\gamma)^{-1})$-mixing absorbing $\tbsg$ approximately with the hardness result in~\Cref{thm:PPAD-hard}, we immediately obtain the following corollary.
\begin{corollary}[$\ppad$-hardness of $\Abs$-$\tbsg$]\label{coro:ar-hardness}
Given any mixing time $\tmix$ that is lower bounded by some constant and some small enough constant accuracy $\epsilon$,  the problem of finding an $\epsilon$-approximate NE in $\tmix$-mixing $\Abs$-$\onestate$s is $\ppad$-hard. 
\end{corollary}

\end{document}